\pdfoutput=1
\documentclass[10pt,letterpaper]{article}
\usepackage[T1]{fontenc}
\usepackage[utf8]{inputenc}
\usepackage[textwidth=5.8in,textheight=8.9in]{geometry}
\usepackage{amsmath,amsthm}
\usepackage{newpxtext}
\usepackage{newpxmath}
\usepackage{microtype}
\usepackage{xcolor}
\usepackage{enumitem}
\usepackage{authblk}
\usepackage{algorithm}
\usepackage{algpseudocode}
\makeatletter
\AtBeginDocument{\def\theHALG@line{\thealgorithm.\arabic{ALG@line}}}
\makeatother
\usepackage{graphicx}
\usepackage{nicefrac}
\usepackage{booktabs}
\usepackage[font=small,labelfont=bf,skip=6pt]{caption}

\usepackage{titlesec}
\titleformat*{\section}{\large\bfseries}
\titleformat*{\subsection}{\normalsize\bfseries}
\titleformat*{\subsubsection}{\normalsize\itshape}
\titlespacing*{\section}{0pt}{3.2ex plus 1ex minus .2ex}{1.4ex plus .2ex}
\titlespacing*{\subsection}{0pt}{2.6ex plus .8ex minus .2ex}{1ex plus .2ex}
\titlespacing*{\subsubsection}{0pt}{2.2ex plus .6ex minus .2ex}{.8ex plus .2ex}
\definecolor{linkblue}{RGB}{0,51,128}
\usepackage{hyperref}
\hypersetup{colorlinks=true, linkcolor=linkblue, citecolor=linkblue, urlcolor=linkblue}
\usepackage[capitalize,noabbrev]{cleveref}
\makeatletter
\g@addto@macro{\UrlBreaks}{\UrlOrds}
\makeatother

\newtheorem{theorem}{Theorem}

\newtheorem{lemma}{Lemma}

\newtheorem{corollary}{Corollary}

\theoremstyle{definition}
\newtheorem{definition}{Definition}

\theoremstyle{remark}
\newtheorem{remark}{Remark}

\newcommand{\SCP}{\textrm{SCP}}
\newcommand{\SCT}{\textrm{SCT}}
\newcommand{\BM}{B_{\max}}
\newcommand{\Bz}{B_0}
\newcommand{\algAone}{\hyperref[alg:A1]{Algorithm~\ensuremath{\mathcal{A}_1}}}
\newcommand{\algAtwo}{\hyperref[alg:mmcr]{Algorithm~\ensuremath{\mathcal{A}_2}}}

\title{\bfseries Job Scheduling with Battery Recharging Constraints}
\author[1]{Rudransh Kumar}
\author[2]{Nima Nasiri}
\author[3]{Jared Paul}
\author[1]{Sathish Gopalakrishnan\thanks{Corresponding author: \texttt{sathish@ece.ubc.ca}}}
\affil[1]{Department of Electrical and Computer Engineering,
          The University of British Columbia, Vancouver, BC, Canada}
\affil[2]{Lyft, Inc.}
\affil[3]{Amazon.com, Inc.}

\date{}

\begin{document}
\maketitle
\begin{abstract}
A battery-powered device, such as a delivery drone that works from a depot,
must stop to recharge between jobs, and the time spent recharging delays
every job that follows. Scheduling models that fix the duration of a
recharge do not describe a device whose recharge takes longer when it
acquires more energy. We study a single device that executes a batch of
non-preemptive jobs with known execution times, energy demands, and
optional deadlines. Under \emph{partial recharging} the device may acquire
any amount of energy between jobs; under \emph{complete recharging} every
recharge fills the battery. Four objectives and four relationships between
execution time and energy demand give 32 variants. When acquiring $q$ units
of energy takes $q$ time units, we show that 14 of the 16 partial-recharging
variants are polynomial, and we give a tight 2-approximation for the average completion time, one of the two NP-hard variants. Under complete
recharging, the four variants with equal energy demands are polynomial and
the other 12 are NP-hard; for makespan we give a $5/4$-approximation. When
each recharge also incurs a fixed \emph{setup time} $h$, the equal-energy
variants remain polynomial and the other 24 are strongly NP-hard if $h$ is
part of the input. For these we give exact algorithms that are exponential
only in the number of jobs, and approximation algorithms for makespan and,
when the battery starts empty, for the weighted average completion time.
Experiments on synthetic and trace-derived job sets compare the algorithms
with exact optima. The model is offline and deterministic, and we have not
validated the schedules on hardware.

\end{abstract}

\section{Introduction}
\label{sec:intro}

A delivery drone that works from a single depot executes a batch of orders,
one round trip per order. The operator fixes the routes before the batch
starts, and therefore knows how long each trip takes and how much battery
charge it consumes: a heavy order to a distant address needs more time and
more charge than a light order to a nearby one. Some orders carry a
promised delivery time. The drone has one battery, which may be only partly
charged when the batch starts, and the depot has a charging pad. The drone
can recharge between trips, but it cannot deliver while it recharges.

The operator must choose the order of the trips and decide when, and for
how long, to recharge. These two decisions interact. A trip that fits
within the remaining charge can start immediately; a trip that needs more
charge waits for the charger, and so does every trip that follows it.
Suppose that the battery is empty, that trip $A$ takes one minute and needs
ten minutes of charging, and that trip $B$ takes two minutes and needs one
minute of charging. If $A$ runs first, then the two deliveries complete at
minutes 11 and 14. If $B$ runs first, then they complete at minutes 3
and 14. Running the shorter trip first, which minimizes the average
completion time when no recharging is involved, is therefore not optimal
here. The choice is harder when the battery starts with some charge,
because the operator must also decide which trips should consume that
charge.

The charging arrangement determines which recharges are possible. Some
operators allow the drone to acquire only the energy that the next trip
needs. Others require the drone to charge to a fixed level whenever it
docks; that level may be the full capacity or a lower target such as 80\%.
We call the first policy \emph{partial recharging} and the second
\emph{complete recharging}. Complete recharging removes the choice of
amount but not the choice of when to dock, and we will show that this
remaining choice determines the total recharging time of the batch.

Starting a recharge may itself take time. The drone must land on the pad
and establish a connection, and this preparation takes the same time
whether the drone then recharges to a small extent or to a more
substantial extent. We call it the
\emph{setup time}. Setup time is a larger fraction of the schedule when
trips are short; one operator reports an average delivery time under five
minutes in initial tests~\cite{doordash2026air}. Consider two trips that
each take one minute and each consume the charge acquired in one minute, an
empty battery that holds two such units of charge, and a setup time of
three minutes. Charging one unit before each trip completes the deliveries
at minutes 5 and 10. Charging two units before the first trip completes
them at minutes 6 and 7: the first delivery is one minute later, but the
drone avoids the second setup, and the average completion time falls
from 7.5 to 6.5 minutes. With setup time, the operator must therefore also decide
which trips share a recharge.

The same decisions arise for a warehouse robot that executes
pick-and-deliver tasks. The time and energy of a task depend on the load
and the distance, and the robot can visit a charger between tasks. A batch
of orders with promised dispatch times matches the description above once
routing has fixed the time and energy demand of each task.

Payload tasks on a satellite and reports from a sensor node also compete
for stored energy and may have deadlines. Orbital charging windows and
intermittent energy harvesting impose constraints that we do not consider,
so our results apply to these settings only when the operator can start a
recharge between any two jobs.

We ask which versions of this scheduling problem can be solved exactly in
polynomial time, which are NP-hard, and how closely the hard versions can
be approximated. Related models either fix the duration of each
recharge~\cite{yu2025recharging} or treat the battery as a budget within a
routing problem (\cref{sec:related}). In the problem that we study, the
duration of a recharge depends on the energy acquired, and the scheduler
chooses both the job order and the recharges. Setup times have also been
studied for a processor that switches between job families, where they
model cold start delays~\cite{gopalakrishnan2025coldstart}. In that
setting the job sequence determines when a setup occurs; here the
scheduler incurs a setup whenever it chooses to recharge.

\subsection{The Scheduling-and-Charging Problem}

We call this problem the \emph{Scheduling-and-Charging Problem} (\SCP{})
and define it formally in \cref{sec:model}. A device with battery capacity
$\BM$ and initial energy $\Bz$ must execute $n$ non-preemptive jobs, all
available at time zero. Job $J_i$ has an execution time $e_i$, an energy
demand $b_i$, and, where applicable, a deadline $d_i$. The device can
recharge between jobs at a constant rate, under either the partial or the
complete policy. For complete recharging we write $\BM$ for the target
level, so the analysis applies to any fixed target and not only to a full
battery.

We consider four objectives: the average completion time; the makespan,
which is the completion time of the last job; the maximum tardiness; and
the number of tardy jobs. (A job is tardy if it completes after its
deadline.) Because the number of jobs is fixed, minimizing the average
completion time is equivalent to minimizing the sum of completion times
(\SCT{}), with the same optimal schedules and the same approximation
factors. We refer to the average in this introduction and state the
formal results for the sum. We combine these objectives with four relationships between $e_i$ and $b_i$:
equal energy demands, equal execution times, energy demand equal to
execution time, and unrelated parameters. With two policies, this gives 32
variants (\cref{tab:summary}).

We treat the 32 variants twice. We first assume that acquiring $q$ units
of energy takes $q$ units of time
(\cref{sec:complete,sec:partial,sec:experiments}). We then add the setup
time $h$ to every recharge
(\cref{sec:setup-partial,sec:setup-complete,sec:hardness,sec:selection,sec:setup-experiments}).
The two settings require different arguments, and several ordering rules
that are optimal without setup time fail with it.

\subsection{Results without Setup Time}

Under partial recharging without setup time, the device needs to acquire
only the energy that the next job lacks. Once we choose the job order, the
recharge durations follow (\cref{lem:min-recharge}), and each objective
reduces to an ordering problem. We show that 14 of the 16
partial-recharging variants are solvable in polynomial time
(\cref{sec:partial}). Earliest-deadline-first order minimizes the maximum
tardiness for every job class, every order has makespan
$\sum_i e_i+\max(0,\sum_i b_i-\Bz)$, and shortest-job-first order or a
Moore--Hodgson rule solves the three structured job classes for the other
two objectives. The two remaining variants have unrelated $(e_i,b_i)$ and
minimize either the average completion time or the number of tardy jobs.

The first of these two variants was left open in the conference version of
this paper~\cite{gopalakrishnan2022rtss}. The case is interesting because
the difficulty arises only from the initial charge: with an empty battery,
sorting the jobs by $e_i+b_i$ is optimal. With $\Bz>0$, the scheduler must
decide which jobs execute on the initial charge, and we show that this
decision is NP-hard through a reduction from Partition that uses $2n+3$
jobs (\cref{thm:sct-arb-nph}). An exact cost identity shows that a schedule
meets the decision threshold only if it encodes a valid partition. We then
give an $O(n\log n)$ algorithm whose average completion time is at most
twice the optimum, show that the factor two is tight for this algorithm
(\cref{thm:2approx}), and improve the factor to $2-2/(n+1)$ by enumerating
the final job (\cref{thm:final-job}). We also show that minimizing the
number of tardy jobs is NP-hard for unrelated parameters and give a
pseudo-polynomial dynamic program for it
(\cref{thm:numtardy-arb-nph,thm:numtardy-dp}).

Under complete recharging, the scheduler must also decide which jobs
execute between consecutive recharges (\cref{sec:complete}). The four
equal-energy variants are solvable in polynomial time: $O(n\log n)$ for
the average completion time, $O(n)$ for makespan, $O(n^2)$ for maximum
tardiness, and $O(n^3)$ for the number of tardy jobs. The other 12 variants
are NP-hard. For makespan, we show that the total recharging time depends
only on which jobs execute after the last recharge. Choosing that final
group is a subset-sum problem, which yields both the hardness proof and an
$O(n\log n)$ algorithm whose makespan is at most $5/4$ of the optimum
(\cref{thm:makespan-nph,thm:ssmin-reduct,thm:ssmax-reduct,thm:a2-guarantee}).

\subsection{Results with Setup Time}

With setup time, even partial recharging requires the scheduler to group
jobs into charging sessions, as the two-trip example shows. The eight
equal-energy variants remain solvable in polynomial time under both
policies, and the other 24 are strongly NP-hard when $h$ is part of the
input (\cref{thm:setup-hardness}). Ten partial-recharging variants that are
polynomial without setup time therefore become hard. The makespan and
deadline reductions already hold at $h=1$. The reduction for the average completion time chooses $h$ from the input, so it does not establish
hardness for every fixed positive value of $h$.

Equal energy demands allow us to exchange jobs without changing which
charging patterns are feasible. Shortest-job-first and
earliest-deadline-first orders therefore remain optimal for the average completion time and for maximum tardiness, respectively. Minimizing the
number of tardy jobs requires a dynamic program, because removing the
longest job from an infeasible prefix can fail even when demands are
equal. The resulting algorithms take $O(n\log n)$ time for the average completion time, $O(n)$ for makespan, and $O(n^3)$ for the two deadline
objectives (\cref{sec:setup-partial,sec:setup-complete}). These bounds
count arithmetic operations.

For unequal demands (\cref{sec:selection}), we give exact algorithms that
are exponential in $n$ but independent of the numerical battery capacity,
and polynomial algorithms that optimize the recharges for any prescribed
job order. We approximate makespan within a factor of two for any initial
energy and within $3/2$ when the battery starts empty. For the average completion time with an empty battery and positive demands, a sorting key
that depends on $h$, followed by next-fit grouping, gives a factor of three
under either policy, even with job weights (\cref{thm:sct-approx}). The
dependence on $h$ is necessary in the following sense: the zero-setup order
with optimal recharging has ratio $\Omega(\sqrt n)$
(\cref{thm:a1-order-lower-bound}).

\subsection{Scope}

The model applies to a depot-based device after routing has fixed the
duration and energy demand of each job. It does not optimize routes,
resolve competition among several devices for chargers, or represent
charging windows. Setup consumes time but no energy, the charging rate is
constant, and charging cannot overlap execution. The time to travel to a
charger is covered only if it is the same for every session, in which case
it can be included in $h$. These assumptions isolate the scheduling effect
of recharge duration and setup time; we have not validated them against a
physical battery or charger.

\subsection{Relation to the Conference Version}

A preliminary version of this work appeared at the IEEE Real-Time Systems
Symposium (RTSS) in 2022~\cite{gopalakrishnan2022rtss}. That version considered the 24
variants without makespan and without setup time, and left open the average completion time under partial recharging with unrelated parameters. Here
we settle that case, add the 2-approximation, study the eight makespan
variants including the $5/4$-approximation, revise the algorithms and
proofs for the deadline objectives, and add the computational evaluation.
The results on setup time did not appear in the conference version.

\section{System Model}
\label{sec:model}

We now state the problem precisely. We first describe the jobs and the
battery, then the recharging policies and the setup time, and finally the
schedules and the objectives that we use to compare them.

\subsection{Jobs and Battery}

One battery-powered device executes $n$ jobs, one at a time, and it cannot
execute a job while it recharges. We write
$\mathcal{J}=\{J_1,\ldots,J_n\}$ for the job set. All jobs are available at
time zero, which corresponds to a batch that the operator holds before the
device starts. Jobs are \emph{non-preemptive}: once a job starts, it runs
to completion. This assumption is reasonable for the applications in
\cref{sec:intro}, because a delivery trip or a pick-and-deliver task cannot
be suspended midway and resumed after a recharge.

Each job $J_i$ has an \emph{execution time} $e_i\ge0$, which is the time
for which the job occupies the device. It has an \emph{energy demand}
$b_i\ge0$, which is the charge that it consumes. For the deadline
objectives, it also has a \emph{deadline} $d_i$.

The battery has capacity $\BM>0$ and initial energy $\Bz$, with
$0\le\Bz\le\BM$. We assume that $b_i\le\BM$ for every job, because a job
that needs more than a full battery can never execute. All parameters are
integers, and deadlines may be negative. We model charging as a linear
process and choose the energy unit so that acquiring one unit of energy
takes one unit of time. Linear charging makes the duration of a recharge
proportional to the energy acquired; it does not represent the slower
charging of a nearly full battery, which we discuss in
\cref{sec:conclusions}.

Whether longer jobs also consume more energy affects which ordering rules
are optimal. We therefore distinguish four classes of job sets, according
to the relationship between $e_i$ and $b_i$.

\begin{description}
  \item[Equal energy ($b_i = b$).] Every job consumes the same energy,
        whatever its execution time.
  \item[Equal execution time ($e_i = e$).] Every job executes for the same
        time, whatever its energy demand.
  \item[Energy equal to execution time ($e_i = b_i$).] A job that executes
        for $e$ time units needs $e$ time units of charging. The equality
        fixes the constant of proportionality at one after normalization;
        it does not cover every proportional relationship.
  \item[Arbitrary.] No constraint relates $e_i$ and $b_i$.
\end{description}
These are classes of instances. An individual instance may satisfy several
of the restrictions, or may need no recharging at all.

\subsection{Recharging and Setup Time}

The device can recharge before the first job and between jobs. A recharge
that acquires $q>0$ units of energy first incurs a \emph{setup time}
$h\ge0$, so it occupies the device for
\begin{equation}
 L(q)=q+h\,\mathbf{1}\{q>0\}
 \label{eq:setup-duration}
\end{equation}
time units. Setup consumes no energy. Choosing $q=0$ incurs neither setup
nor charging time; in particular, a request to fill a battery that is
already full has no cost. We exclude a recharge after the last job because
it cannot improve any of the objectives below. The setup time is a
duration; a monetary charge per session would require a different
objective.

We study two policies, which differ in the amounts that the device may
acquire when the battery holds $B$ units.

\begin{description}
  \item[Partial recharging.] The device may acquire any integer amount
        $q\in[0,\BM-B]$.
  \item[Complete recharging.] The device may acquire either $q=0$ or
        $q=\BM-B$: a recharge, once started, fills the battery.
\end{description}

When $h=0$, a recharge of $q$ units takes exactly $q$ time units. This is
the model of the conference version~\cite{gopalakrishnan2022rtss}, and we
assume it throughout \cref{sec:complete,sec:partial,sec:experiments}. From
\cref{sec:setup-partial} onward, $h$ is an input.

\subsection{Schedules}

A \emph{schedule} consists of a permutation $\pi$ of the jobs and amounts
$q_1,\ldots,q_n$, where the device acquires $q_j$ units immediately before
the $j$-th job $\pi(j)$. We write $r_j=L(q_j)$ for the duration of that
recharge; without setup time, $r_j=q_j$. We omit idle time. All jobs are
available at time zero, a recharge is possible between any two jobs, and
no objective below improves when a completion time increases, so removing
idle time preserves feasibility and cannot worsen the objective.

Let $C_j$ denote the completion time of the $j$-th job, with $C_0=0$. Let
$B_j$ denote the battery level immediately after the $j$-th job, so that
$B_0$ is the initial energy. A schedule is \emph{feasible} if, for every
position $j$,
\begin{align}
 b_{\pi(j)}&\leq B_{j-1}+q_j\leq\BM,\label{eq:feasibility}\\
 B_j&=B_{j-1}+q_j-b_{\pi(j)},\label{eq:battery-level}\\
 C_j&=C_{j-1}+r_j+e_{\pi(j)},\label{eq:completion}
\end{align}
and if, under complete recharging, each $q_j$ is either zero or
$\BM-B_{j-1}$. Constraint (\ref{eq:feasibility}) requires that the job has
enough energy to execute and that the battery never exceeds its capacity.
When the job order is fixed or clear from the context, we index the jobs
by position and write $e_j$, $b_j$, and $d_j$ for the parameters of the
$j$-th job; then $C_j=\sum_{k\le j}(e_k+r_k)$.

A \emph{group} is a maximal sequence of consecutive jobs with no positive
recharge between them. The first group may execute on the initial energy
without any recharge.

\subsection{Objectives}

An operator may wish to reduce the average completion time, to complete
the whole batch early, or to limit deadline misses. We represent these
requirements through four objectives, each of which we minimize.

\begin{description}
  \item[Sum of completion times,] $\SCT=\sum_{j=1}^n C_j$. Minimizing
        \SCT{} minimizes the average completion time and, because the
        total execution time is fixed, the average waiting time. This
        objective applies to deliveries of equal priority without hard
        deadlines.
  \item[Makespan,] $C_{\max}=C_n$, the completion time of the last job.
        This objective applies when a whole batch must complete before a
        fixed time, such as the end of a shift.
  \item[Maximum tardiness,]
        $T_{\max}=\max\bigl(0,\max_j(C_j-d_{\pi(j)})\bigr)$. This
        objective bounds how late any single job can be relative to its
        deadline.
  \item[Number of tardy jobs,]
        $U=\sum_{j=1}^n\mathbf{1}\{C_j>d_{\pi(j)}\}$. This objective
        applies when a late job has no value, such as a sensor reading
        reported after its time window closes.
\end{description}
All four objectives equal zero for an empty job set.

Four totals recur in the analysis. We write $E=\sum_i e_i$ for the total
execution time and $S=\sum_i b_i$ for the total energy demand. For a given
schedule, $Q=\sum_j q_j$ is the acquired energy and
$K=\sum_j\mathbf{1}\{q_j>0\}$ is the number of positive recharges, which
we call \emph{charging sessions}. Energy
conservation gives $Q=S-\Bz+B_n$, and therefore
\begin{equation}
 C_{\max}=E+Q+hK.
 \label{eq:makespan-setup}
\end{equation}
Without setup time, minimizing makespan is equivalent to minimizing the
acquired energy. With setup time, the number of sessions also matters, and
a schedule that acquires the least energy need not have the least
makespan.

\subsection{Classification}

The complexity of the 32 variants differs between the two policies and
between the two settings (\cref{tab:summary}). Without setup time, partial
recharging is hard only for the sum of completion times and the number of
tardy jobs with arbitrary parameters, whereas complete recharging is hard
for every job class except equal energy. With setup time, the
classification is the same under both policies: equal energy is polynomial
and the other three classes are strongly NP-hard
(\cref{thm:setup-hardness}).

We analyze complete recharging without setup time first
(\cref{sec:complete}), because grouping jobs between full recharges is the
source of its hardness and the same grouping decision returns under both
policies once setup takes time. We then show that acquiring only the
energy needed reduces the partial-recharging variants without setup time
to ordering problems (\cref{sec:partial}).

\begin{table}[tbp]
\caption{Complexity of the 32 \SCP{} variants. P: solvable in polynomial
time. H: NP-hard. sH: strongly NP-hard when the setup time $h$ is part of
the input; for $C_{\max}$, $T_{\max}$, and $U$ this already holds at
$h=1$. Setup time changes ten entries, all under partial recharging.}
\label{tab:summary}
\centering
\begin{tabular}{lcccccccc}
\toprule
& \multicolumn{4}{c}{Partial recharging}
& \multicolumn{4}{c}{Complete recharging} \\
\cmidrule(lr){2-5}\cmidrule(lr){6-9}
Objective
  & $b_i{=}b$ & $e_i{=}e$ & $e_i{=}b_i$ & Arbitrary
  & $b_i{=}b$ & $e_i{=}e$ & $e_i{=}b_i$ & Arbitrary \\
\midrule
\multicolumn{9}{l}{\emph{Without setup time} ($h=0$)}\\[2pt]
$\SCT$     & P & P & P & H & P & H & H & H \\
$C_{\max}$ & P & P & P & P & P & H & H & H \\
$T_{\max}$ & P & P & P & P & P & H & H & H \\
$U$        & P & P & P & H & P & H & H & H \\
\midrule
\multicolumn{9}{l}{\emph{With setup time}}\\[2pt]
$\SCT$     & P & sH & sH & sH & P & sH & sH & sH \\
$C_{\max}$ & P & sH & sH & sH & P & sH & sH & sH \\
$T_{\max}$ & P & sH & sH & sH & P & sH & sH & sH \\
$U$        & P & sH & sH & sH & P & sH & sH & sH \\
\bottomrule
\end{tabular}
\end{table}

\section{Complete Recharging without Setup Time}
\label{sec:complete}

Having defined the model, we first study complete recharging without setup
time. Under this policy every recharge fills the battery, so the scheduler
does not choose a recharge amount; it chooses the job order and the
positions in that order at which the device recharges. In this section and
the next, $h=0$, and the recharge before the $j$-th job therefore lasts
$r_j=q_j\in\{0,\,\BM-B_{j-1}\}$, where $B_{j-1}$ is the battery level when
the recharge begins.

We show that equal energy demands make all four objectives solvable in
polynomial time and that unequal demands make all four NP-hard. We treat
the sum of completion times, the maximum tardiness, and the number of
tardy jobs in turn. We treat makespan last, because its analysis also
yields an approximation algorithm.

\subsection{Minimizing Sum of Completion Times (\SCT{})}

\subsubsection{Equal Energy Consumption ($b_i = b$)}

When every job consumes the same energy $b>0$, the device can execute at
most $m^* = \lfloor \BM / b \rfloor$ jobs consecutively without recharging.
(If $b=0$, then the device never recharges and shortest-job-first (SJF)
order is optimal.) For $b>0$, we first separate the choice of job order
from the choice of recharge positions.

\begin{lemma}[SJF within groups]
\label{lem:sjf-groups}
Suppose that the recharging intervals are fixed. Among all orderings of the
jobs within the resulting groups, SJF minimizes the sum of completion times.
\end{lemma}

\begin{proof}
Because all jobs consume the same energy, we can reorder the jobs within a
group without changing the recharge durations. We can therefore minimize
the group's contribution from execution times by placing its jobs in SJF
order. We give the decomposition into execution and recharge contributions
in the appendix (\cref{apx:lem:sjf-groups}).
\end{proof}

With SJF order inside each group, we still need to determine how many jobs
each group contains. The next lemma reduces this choice: we can fill the
interior groups to their maximum size without increasing the sum of
completion times.

\begin{lemma}[Full interior groups]
\label{lem:full-groups}
There exists an optimal schedule in which every group of jobs that lies strictly
between the first and last recharging intervals has exactly $m^*$ jobs.
\end{lemma}

\begin{proof}
If an interior group is not full, then we move jobs into it from the next
group while preserving their order. This changes where we recharge but
does not increase the recharge contribution to the objective
(\cref{apx:lem:full-groups}).
\end{proof}

These two lemmas leave a single choice, the number of jobs that execute
before the first recharge. Once we make that choice, every later group is
full except possibly the last. We show next that SJF order applies across
the whole schedule, so that after sorting we can evaluate every first-group
size in constant time.

\begin{theorem}[\SCT{}, equal energy]
\label{thm:sct-complete-eq}
Minimizing the sum of completion times under complete recharging with equal
energy demand is solvable in $O(n \log n)$ time.
\end{theorem}

\begin{proof}
By \cref{lem:sjf-groups}, we sort jobs in SJF order.
Because all jobs consume equal energy $b$, the recharging time contributed to
the objective depends only on the group sizes, not on which job occupies which
group. Global SJF order minimizes the execution contribution and therefore
minimizes the objective across groups as well as within each group.
If $b=0$ or $nb\le\Bz$, then the SJF schedule needs no recharge and is optimal.
Otherwise set $m=m^*=\lfloor\BM/b\rfloor$.
By \cref{lem:full-groups}, we can fill every group after the first recharge
to $m$ jobs except the last. Enumerate the number $p$ of jobs
before the first recharge from $0$ to $\lfloor\Bz/b\rfloor<n$.
For each $p$, put $N=n-p$ and $\ell=\lfloor(N-1)/m\rfloor$.
The first recharge lasts $\BM-\Bz+pb$ and delays all $N$ remaining jobs;
the $\ell$ later recharges each last $mb$, since a full group leaves residual
charge $\BM-mb$. Hence the recharge contribution is
\[
 R(p)=N(\BM-\Bz+pb)
       +mb\left(\ell N-\frac{m\ell(\ell+1)}2\right).
\]
The execution contribution is independent of $p$, so we evaluate each
choice in $O(1)$ time and all choices in $O(n)$ time.
Sorting dominates the running time, giving $O(n \log n)$ overall.
\end{proof}

The algorithm relies on the fact that, with equal demands, the recharge
durations depend on the group sizes alone. This property fails for the
other three job classes.

\subsubsection{Equal Execution Time ($e_i = e$) and Other Cases}

When energy demands differ, the recharge duration of a group depends on
which jobs it contains and not only on how many. We use this selection
problem to prove NP-hardness: in our construction, choosing the groups is
equivalent to finding an equal-cardinality partition, even when all jobs
have the same execution time.

\begin{theorem}[Hardness of \SCT{}]
\label{thm:sct-complete-nph}
Minimizing the sum of completion times under complete recharging is NP-hard
when jobs have equal execution time ($e_i = e$), when energy demand equals
execution time ($e_i = b_i$), or when the relationship is arbitrary.
\end{theorem}

\begin{proof}
We reduce from EquiPartition by constructing jobs that can form exactly
two full groups if and only if the partition exists. We then show that
every other grouping exceeds the decision threshold. We choose separate
thresholds for equal execution times and for $e_i=b_i$
(\cref{apx:thm:sct-complete-nph}).
\end{proof}

\subsection{Minimizing Maximum Tardiness ($T_{\max}$)}

\subsubsection{Equal Energy Consumption ($b_i = b$)}

To minimize the maximum tardiness, we first determine whether we can
choose the job order independently of the recharge positions. With equal
energy demands, exchanging two jobs changes no recharge duration, and an
exchange argument then shows that earliest-deadline-first (EDF) order is
optimal.

\begin{lemma}[EDF with equal energy]
\label{lem:edf-complete}
When all jobs consume the same energy $b$, scheduling in EDF order
minimizes maximum tardiness under complete recharging.
\end{lemma}

\begin{proof}
Let $\Pi$ contain an adjacent inversion: $J_j$ immediately precedes $J_i$
although $d_i<d_j$. Swap these two jobs while keeping the recharging
durations attached to their positions. Equal energy demand preserves
feasibility and every battery level at a slot boundary.
The swap preserves the completion of the second position; write it as $t$.
After the swap, $J_i$ completes by $t$ and $J_j$ completes at $t$.
Both new tardiness values are at most $(t-d_i)^+$, the old tardiness of
$J_i$. The swap preserves every other completion, so maximum tardiness does
not increase.
Iterating over all inversions yields an EDF schedule that is at least as good as
$\Pi$.
\end{proof}

EDF order determines the sequence of the jobs but not the positions at
which the device recharges. We use a dynamic program to choose these
positions, which are the boundaries between groups.

\begin{theorem}[Maximum tardiness, equal energy]
\label{thm:maxt-complete-eq}
Minimizing maximum tardiness under complete recharging with equal energy
consumption is solvable in $O(n^2)$ time.
\end{theorem}

\begin{proof}
For a fixed EDF order, consider a group that begins with a recharge. We
can determine the cumulative recharging time from the position of its first
job, independently of how we grouped the preceding jobs, by energy
conservation. We also allow the first group to use the initial charge
without recharging. We can therefore evaluate the maximum tardiness of
each possible group and use a dynamic program to choose where one group
ends and the next begins
(\cref{apx:thm:maxt-complete-eq}).
\end{proof}

\subsubsection{Other Cases}

When energy demands differ, we obtain NP-hardness from the corresponding
makespan problem, which we show to be NP-hard in
\cref{sec:complete-makespan}. If every deadline is zero, then the
tardiness of each job equals its completion time, so the maximum tardiness
equals the makespan.

\begin{theorem}[Hardness of maximum tardiness]
\label{thm:maxt-complete-nph}
Minimizing maximum tardiness under complete recharging is NP-hard when jobs have
equal execution time ($e_i = e$), when $e_i = b_i$, or when the relationship is
arbitrary.
\end{theorem}

\begin{proof}
We assign deadline zero to each job of a makespan decision instance.
Maximum tardiness then equals makespan for every schedule
(\cref{apx:thm:maxt-complete-nph}).
\end{proof}

\subsection{Minimizing Number of Tardy Jobs ($U$)}

\subsubsection{Equal Energy Consumption ($b_i = b$)}

To minimize the number of tardy jobs, we select a largest set of jobs that
can all meet their deadlines, execute those jobs in EDF order, and place
the remaining jobs afterward. We use a dynamic program to make the
selection (\cref{alg:numtardy-complete}). Its state records the number of
jobs selected and the battery level after the last of them. For each
state, the program retains the earliest completion time, because an
earlier completion leaves at least as much time for every later job as a
later completion would.

\begin{algorithm}[tbp]
\caption{Minimize $U$: complete recharging, $b_i = b$}
\label{alg:numtardy-complete}
\begin{algorithmic}[1]
\Require Jobs $J_1,\ldots,J_n$ in EDF order; equal energy $b$
\Ensure Maximum number of jobs that can meet their deadlines
\State $D[0][\Bz] \gets 0$ and $D[t][\beta] \gets \infty$ otherwise
       \Comment{$D[t][\beta]$: least completion of the $t$-th timely job, ending at battery $\beta$}
\For{$i = 1$ to $n$}
  \State $D' \gets D$                                       \Comment{$J_i$ tardy}
  \For{each $(t,\beta)$ with $L \gets D[t][\beta] < \infty$}
    \If{$\beta \ge b$ and $L + e_i \le d_i$}                \Comment{$J_i$ extends the group}
      \State $D'[t{+}1][\beta{-}b] \gets \min\!\bigl(D'[t{+}1][\beta{-}b],\ L + e_i\bigr)$
    \EndIf
    \If{$L + (\BM{-}\beta) + e_i \le d_i$}                  \Comment{$J_i$ starts a new group}
      \State $D'[t{+}1][\BM{-}b] \gets \min\!\bigl(D'[t{+}1][\BM{-}b],\ L + (\BM{-}\beta) + e_i\bigr)$
    \EndIf
  \EndFor
  \State $D \gets D'$
\EndFor
\State \Return $\max\{\,t : D[t][\beta] < \infty \text{ for some } \beta\,\}$
\end{algorithmic}
\end{algorithm}

\begin{theorem}[Tardy jobs, equal energy]
\label{thm:numtardy-complete-eq}
The dynamic program minimizes the number of tardy jobs under complete
recharging when all jobs consume the same energy (\cref{alg:numtardy-complete}).
It runs in $O(n^3)$ time.
\end{theorem}

\begin{proof}
Suppose two partial schedules have the same timely-job count and battery
level but different completion times. Any continuation that meets its
deadlines after the later completion can also meet them after the earlier
one. We can therefore retain only the earliest completion for each state
without losing an optimal solution
(\cref{apx:thm:numtardy-complete-eq}).
\end{proof}

\subsubsection{Other Cases}

When jobs have equal execution times, $e_i = b_i$, or an arbitrary
relationship, minimizing $U$ under complete recharging is NP-hard.
To prove this, we take a makespan decision instance of the corresponding
variant from \cref{thm:makespan-nph}, with threshold $D$, and give every
job the common deadline $d_i=D$. A schedule then has no tardy jobs if and only
if its makespan is at most $D$. This polynomial transformation preserves
the relation between execution time and energy demand, so deciding
whether the minimum number of tardy jobs is zero is NP-hard in all three
variants.

\subsection{Minimizing Makespan ($C_{\max}$)}
\label{sec:complete-makespan}

The total execution time is fixed, so we minimize
$C_{\max} = \sum_i e_i + \sum_i r_i$ by minimizing the time spent
recharging. We show that this time depends only on which jobs execute
after the last recharge. This observation gives a linear-time algorithm
for equal energy demands, an NP-hardness proof for the other job classes,
and, through two reductions to subset-sum problems, the approximation
algorithm $\mathcal{A}_2$.

\begin{definition}[Final group and head]
\label{def:head-final} The \emph{final group} $\mathcal{F}$ of a schedule $\sigma$ is the set of jobs
executed after the last recharging interval.
If $\sigma$ contains no recharging interval, then $\mathcal{F} = \mathcal{J}$.
The \emph{head} $\mathcal{H} = \mathcal{J} \setminus \mathcal{F}$ consists of
all remaining jobs.
\end{definition}

Whenever $\sum_j b_j > \Bz$, the device must recharge at least once.
The jobs after its last recharge must fit within one full battery:
\begin{equation}
  \sum_{j \in \mathcal{F}} b_j \;\le \; \BM.
  \label{eq:final-group-feasible}
\end{equation}
Equivalently, the head must consume at least the energy that does not fit
in the final group:
\begin{equation}
  \sum_{j \in \mathcal{H}} b_j \;\ge \;
    \sum_{j \in \mathcal{J}} b_j - \BM.
  \label{eq:head-constraint}
\end{equation}
We still need to determine whether the recharge positions within the head
affect the total recharging time. They do not: energy conservation
determines that time from the head alone.

\begin{lemma}[Recharging time formula]
\label{lem:recharge-formula}
For any feasible schedule $\sigma$ that contains at least one positive recharging interval,
the total recharging time equals
\begin{equation}
  \sum_{i=1}^n r_i \;=\; \BM - \Bz + \sum_{j \in \mathcal{H}} b_j.
  \label{eq:recharge-formula}
\end{equation}
\end{lemma}

\begin{proof}
Let $B_{\mathrm{end}}$ denote the battery level after the last job completes.
By energy conservation:
\[
  \Bz + \sum_{i=1}^n r_i \;=\; \sum_{j \in \mathcal{J}} b_j + B_{\mathrm{end}}.
\]
Under complete recharging the last recharging interval fills the battery to
$\BM$, after which the jobs in $\mathcal{F}$ execute.
Hence $B_{\mathrm{end}} = \BM - \sum_{j \in \mathcal{F}} b_j$. Substituting and rearranging:
\[
  \sum_{i=1}^n r_i
  = \sum_j b_j + \BM - \sum_{j \in \mathcal{F}} b_j - \Bz
  = \BM - \Bz + \sum_{j \in \mathcal{H}} b_j. \qedhere
\]
\end{proof}

Since $\BM$ and $\Bz$ are constants of the instance, whenever a recharge
is necessary we minimize recharging time by choosing a head $\mathcal{H}$ of
minimum total energy (\cref{lem:recharge-formula}). The head must satisfy
(\ref{eq:head-constraint}); equivalently, we maximize the final group's
energy subject to (\ref{eq:final-group-feasible}). Any such bipartition
$\mathcal{J}=\mathcal{H}\cup\mathcal{F}$ yields a feasible schedule
(\cref{alg:shtf}).

\begin{algorithm}[tbp]
\caption{$\textsc{ScheduleByBipartition}: \mathcal{J} = \mathcal{H} \cup \mathcal{F}$}
\label{alg:shtf}
\begin{algorithmic}[1]
\Require $\mathcal{F} \subseteq \mathcal{J}$ satisfies (\ref{eq:final-group-feasible}), $\mathcal{H} = \mathcal{J} \setminus \mathcal{F}$
\Ensure A feasible schedule $\sigma$
\If{$\sum_j b_j\le\Bz$}
  \State Execute $\mathcal{H}$ followed by $\mathcal{F}$ without recharging; \Return
\EndIf
\State Recharge $\BM - \Bz$ before scheduling any jobs
\For{each job in $\mathcal{H}$, in any order}
  \State Execute the job
  \If{at least one job in $\mathcal{H}\cup\mathcal{F}$ remains unexecuted}
    \State Recharge completely
  \EndIf
\EndFor
\State Append all jobs in $\mathcal{F}$ in any order, without adding any recharging intervals
\end{algorithmic}
\end{algorithm}

\begin{lemma}[Schedule from a bipartition]\label{lem:shtf-feasible}
If $\mathcal{H}$ and $\mathcal{F}$ form a bipartition of $\mathcal{J}$
that satisfies the equivalent constraints (\ref{eq:final-group-feasible})
and (\ref{eq:head-constraint}), then \cref{alg:shtf} constructs a feasible
schedule in $O(n)$ time.
\end{lemma}

\begin{proof}
If $\sum_j b_j\le\Bz$, then we can execute every job without recharging.
Otherwise, the construction fills the battery before each head job, which
fits by the standing assumption $b_j\le\BM$. If $\mathcal{F}$ is nonempty,
then we also fill the battery before executing its jobs; their total
consumption fits by (\ref{eq:final-group-feasible}). If $\mathcal{F}$ is
empty, then we stop after the last head job and add no terminal recharge.
These operations take $O(n)$ time.
\end{proof}

When recharging is necessary, the construction uses at most
$\BM-\Bz+\sum_{j\in\mathcal{H}}b_j$ units of recharging time. If
$\mathcal{F}$ is nonempty, then this bound is an equality: the initial
recharge contributes $\BM-\Bz$, and the recharge after each head job
replenishes its consumption. If $\mathcal{F}$ is empty, then the last
head job has no following recharge, so the actual cost can be smaller.
We retain this distinction when mapping an arbitrary subset-sum answer
back to a schedule.

\subsubsection{Equal Energy ($b_i = b$)}

If $b=0$, then all jobs run without recharging. Otherwise the final group
contains at most $\min(n, \lfloor {\BM}/b \rfloor)$ jobs
(\ref{eq:final-group-feasible}), and any selection of that many jobs for
$\mathcal{F}$ minimizes $C_{\max}$.

\begin{theorem}[Makespan, equal energy]
\label{thm:makespan-equal}
Minimizing $C_{\max}$ under complete recharging with $b_i = b$ is solvable in
$O(n)$ time.
\end{theorem}

\begin{proof}
If $\sum_j b_j \le \Bz$, then all jobs execute without recharging, so we
set $\mathcal{F}=\mathcal{J}$. Otherwise, we minimize recharging time and
hence makespan by maximizing the energy in $\mathcal{F}$ subject to its
capacity constraint (\cref{lem:recharge-formula}). Since all jobs have equal
$b$, we select any $\min(n,\lfloor\BM/b\rfloor)$ jobs for $\mathcal{F}$.
Selection and schedule construction each take $O(n)$ time
(\cref{alg:shtf,lem:shtf-feasible}).
\end{proof}

\subsubsection{Other Cases: NP-hardness}
When jobs consume different amounts of energy, choosing a final group that
uses as much of the battery as possible is a subset-selection problem. It
encodes Subset-Sum, which gives NP-hardness for each of the remaining job
classes.

\begin{theorem}[Hardness of makespan]
\label{thm:makespan-nph}
Minimizing $C_{\max}$ under complete recharging is NP-hard when jobs have equal
execution time ($e_i = e$), when $e_i = b_i$, or when the relationship is
arbitrary.
\end{theorem}

\begin{proof}
We reduce from Subset-Sum (SS): given a multiset $\{a_1, \ldots, a_n\}$ of
positive integers and a positive target $C$, does any subset sum to exactly $C$?
Discard values larger than $C$, which cannot participate in a solution.
If no values remain, then the answer is NO and we return a fixed NO instance;
otherwise every remaining value is individually feasible at capacity $C$.

\emph{Construction.}
Set $\BM = C$, $\Bz = 0$, and create $n$ jobs with $b_j = a_j$ for each $j$.
For the $e_i = b_i$ case set $e_j = b_j = a_j$;
for the $e_i = e$ case set $e_j = 1$ for all $j$;
for the arbitrary case use either assignment.

\emph{Correctness.}
Since $\Bz = 0$ and $\sum_j b_j > 0$, every feasible schedule requires at
least one recharging interval.
By \cref{lem:recharge-formula} and (\ref{eq:final-group-feasible}), an optimal schedule has a final group $\mathcal{F}^*$ that maximizes
$\sum_{j \in \mathcal{F}} b_j$ subject to $\sum_{j \in \mathcal{F}} b_j \le C$.
Since each $b_j=a_j$, the SS instance has a solution if and only if the
optimal final-group energy is $C$. Equivalently, the optimal makespan is
at most the decision threshold $D=\sum_j e_j+\sum_j a_j$:
its exact value is $D+C-\sum_{j\in\mathcal{F}^*}b_j$.
This provides a polynomial reduction for each execution-time assignment.
A fixed NO instance for the preprocessing case is one job with
$e_1=b_1=\BM=1$, $\Bz=0$, and threshold $D=1$; its makespan is $2$.
\end{proof}

\subsubsection{Other Cases: Reductions to Subset-Sum Minimization and Maximization}

We have reduced makespan minimization to choosing a head with as little
energy as possible while leaving a final group that fits in the battery.
To approximate this choice, we use \emph{Subset-Sum Minimization}
(SS-MIN): given a multiset $A = \{a_1, \ldots, a_n\}$ of positive integers
and a target $C$, find a set $A' \subset A$ that minimizes
$\sum_{a_i \in A'} a_i$ subject to $\sum_{a_i \in A'} a_i \geq C$.
We prove a \emph{strict reduction}: converting a subset-sum solution back
to a schedule does not increase its approximation ratio.

\begin{theorem}[Reduction to SS-MIN]\label{thm:ssmin-reduct}
Minimizing makespan under complete recharging strictly reduces to SS-MIN in $O(n)$ time.
\end{theorem}

\begin{proof}
Recall that $S=\sum_j b_j$.
If $S\le\BM$, then we can construct an optimal schedule directly: when $S\le\Bz$, execute
all jobs without recharge; otherwise recharge to $\BM$ initially and then
execute all jobs. In the latter case, recharging time is at least $\BM-\Bz$
(\cref{lem:recharge-formula}), attained by this schedule with empty head. For these cases the forward map can return the fixed SS-MIN
instance $A=\{1\}$, $C=1$, and the backward map returns the known optimum
regardless of the oracle's answer. Thus no ratio with zero optimal head
energy is needed.

Now suppose $S>\BM$. The forward map creates one item $a_j=b_j$ for
each positive-energy job, retaining its job index, and sets $C=S-\BM>0$.
Zero-energy jobs need no item. Given a feasible SS-MIN subset $A'$, the
backward map selects the corresponding jobs as the head $\mathcal{H}$
and all other jobs, including zero-energy jobs, as $\mathcal{F}$.
The constraint on $A'$ gives $\sum_{j\in\mathcal{F}}b_j\le\BM$,
so we can construct a feasible schedule using \cref{alg:shtf}. Both maps take $O(n)$.

Write $x$ and $x^*>0$ for the selected and optimal head energies, and
$\gamma=\sum_j e_j+\BM-\Bz\ge0$. An optimal SS-MIN solution leaves at
least one positive-energy job outside its head: excluding any such job
from the full set still leaves a feasible head because $b_j\le\BM$, and
strictly reduces its weight. The construction therefore attains
$\gamma+x^*$ for this optimum. Conversely, the head of any feasible
schedule meets the SS-MIN target, so \cref{lem:recharge-formula} gives
$C_{\max}^*\ge\gamma+x^*$. Thus $C_{\max}^*=\gamma+x^*$.

An arbitrary oracle answer may select every item and leave
$\mathcal{F}$ empty. Our construction still gives
$C_{\max}\le\gamma+x$, which suffices to compare the ratios:
\[
 \frac{C_{\max}}{C_{\max}^*}
 \le\frac{\gamma+x}{\gamma+x^*}
 \le\frac{x}{x^*}
 =\frac{\sum_{a_j\in A'}a_j}{\sum_{a_j\in A^*}a_j}.
\]
Together with the directly solved cases, this proves the strict reduction.
\end{proof}

We can use the reduction to carry an approximation guarantee from SS-MIN
to makespan minimization. To show that the resulting makespan ratio is
also tight, we need a family of subset-sum instances that approaches the
bound and also produces valid scheduling instances; in particular, every
item must fit within the resulting battery capacity. We also need a
nonempty final group, so that the schedule incurs the full recharge cost
of the head that the oracle selects. We ensure this by adding a
zero-energy job that remains after the last recharge.

\begin{theorem}[Conditional tightness transfer]\label{thm:ssmin-tight}
Let $\mathcal{P}$ be a scale-invariant SS-MIN approximation algorithm, meaning
$R_{\mathcal{P}}(\alpha I)=R_{\mathcal{P}}(I)$ for every positive integer
$\alpha$, where scaling multiplies every item and the target by $\alpha$.
Suppose a family $\mathcal{I}$ witnesses its tight ratio
$\rho=\sup_{I\in\mathcal{I}}R_{\mathcal{P}}(I)<\infty$ and every
$I=(A_I,C_I)\in\mathcal{I}$ satisfies
\[
  0<C_I<\sum_{a\in A_I}a,
  \qquad \max_{a\in A_I}a\le \sum_{a\in A_I}a-C_I.
\]
Then the reduction in \cref{thm:ssmin-reduct} attains the same worst-case
ratio $\rho$ on valid scheduling instances obtained from this family,
with equal execution times and one additional zero-energy job.
\end{theorem}

\begin{proof}
Fix $I\in\mathcal{I}$ and a positive integer $\alpha$.
For each item $a_j\in A_I$, create a job with $e_j=1$ and $b_j=\alpha a_j$;
add one job with execution time one and energy demand zero, and set
$\BM=\Bz=\alpha(\sum_j a_j-C_I)$.
The hypotheses ensure $0<b_j\le\BM$ and
$\sum_j b_j>\Bz$, so this is a valid scheduling instance requiring recharge.
The forward map of \cref{thm:ssmin-reduct} omits the zero-energy job and
therefore gives exactly $\alpha I$. The backward map places that job in
$\mathcal{F}$, so the designated final group is nonempty even if
$\mathcal{P}$ selects every item.
Let $y>0$ be the optimal SS-MIN value on $I$, and let $x_\alpha$ be the sum
of the original, unscaled items selected by $\mathcal{P}$ on $\alpha I$.
Then $x_\alpha/y=R_{\mathcal{P}}(\alpha I)=R_{\mathcal{P}}(I)$.
The selected head consumes $\alpha x_\alpha$. Because a job remains
after it, every head job's consumption contributes to recharging time, and
the constructed makespan is $N+1+\alpha x_\alpha$, where $N=|A_I|$.
The optimal makespan is $N+1+\alpha y$: the strict-reduction argument
establishes this equality, and the extra zero-energy job changes only
the fixed execution sum. Hence
\[
 \frac{C_{\max}}{C_{\max}^*}
 =\frac{N+1+\alpha x_\alpha}{N+1+\alpha y}
 =\frac{N+1+\alpha yR_{\mathcal{P}}(I)}{N+1+\alpha y}
 \longrightarrow R_{\mathcal{P}}(I).
\]
Taking the supremum over the witnessing family gives $\rho$; the matching
upper bound follows from \cref{thm:ssmin-reduct}.
\end{proof}

Any SS-MIN approximation algorithm therefore yields a makespan
approximation with no larger relative error (\cref{thm:ssmin-reduct}).
To transfer a tight ratio, we also need a witnessing family that meets the capacity condition
(\cref{thm:ssmin-tight}). Scale invariance alone does not provide such a
family, and we have not established tightness transfer for the cited
algorithm.

The SS-MIN target $C=\sum_j b_j-\BM$ grows with total energy demand.
When $\sum_j b_j>2\BM$, we instead select a final group through
\emph{Subset-Sum Maximization} (SS-MAX): given positive items
$A=\{a_1,\ldots,a_n\}$ and a target $C$, select a subset $A'$ maximizing
$\sum_{a_i\in A'}a_i$ subject to $\sum_{a_i\in A'}a_i\le C$.
This reduction uses the fixed target $C=\BM$, so the absolute error can be
bounded in terms of the battery capacity.

\begin{theorem}[Reduction to SS-MAX]\label{thm:ssmax-reduct}
Whenever $\sum_j b_j > 2\BM$, minimizing makespan under complete
recharging strictly reduces to SS-MAX in $O(n)$ time.
\end{theorem}

\begin{proof}
\emph{Reduction.}
The forward map $f$ creates one item $a_j = b_j$ for each positive-energy job $J_j$, and then sets $C= \BM$. Every item fits by itself because $b_j\le\BM$, so we may assume that the returned subset $A'$ is nonempty. The backward map $g$ chooses $\mathcal{F}$ as the jobs whose items are in $A'$ together with all zero-energy jobs, and then constructs a schedule by \cref{alg:shtf}. Constructing $A$ takes $O(n)$ time, and so does $g$ (\cref{lem:shtf-feasible}).

\emph{Correctness.} Since $\sum_j b_j > \Bz$, every feasible schedule has at least one recharging interval. An optimal schedule's final group $\mathcal{F}^*$ maximizes $\sum_{j \in \mathcal{F}}b_j$ subject to $\sum_{j \in \mathcal{F}}b_j \leq \BM$ (\cref{lem:recharge-formula}), which is the SS-MAX instance created by $f$. These maps preserve the selected energy and therefore give
\[
\frac{\sum_{j \in \mathcal{F}^*} b_j}{\sum_{j \in \mathcal{F}} b_j} = \frac{\sum_{a_j \in A^*} a_j}{\sum_{a_j \in A'} a_j}.
\]
Define $\delta := \sum_{j \in \mathcal{F}^*} b_j - \sum_{j \in \mathcal{F}} b_j \geq 0$. Then $\sum_{j \in \mathcal{F}^*} b_j = \sum_{j \in \mathcal{F}} b_j + \delta$, and by \cref{lem:recharge-formula}, $C_{\max} = C^*_{\max} + \delta$. It follows that
\[
\frac{\sum_{j \in \mathcal{F}^*} b_j}{\sum_{j \in \mathcal{F}} b_j} - \frac{C_{\max}}{C^*_{\max}} = \delta \left( \frac{1}{\sum_{j \in \mathcal{F}} b_j} - \frac{1}{C^*_{\max}}\right).
\]
Since $\delta \geq 0$, it suffices to show $C^*_{\max} \geq \sum_{j \in \mathcal{F}} b_j$. By \cref{lem:recharge-formula}, the total recharging time and hence $C^*_{\max}$ is at least $\BM - \Bz + \sum_j b_j - \BM = \sum_j b_j - \Bz$. Since $\sum_j b_j > 2 \BM$, $C^*_{\max} > \BM$. By (\ref{eq:final-group-feasible}), $\BM \geq \sum_{j \in \mathcal{F}} b_j$. Thus $C^*_{\max} \geq \sum_{j \in \mathcal{F}} b_j$ as required.
\end{proof}

We combine the two reductions
(\cref{thm:ssmin-reduct,thm:ssmax-reduct}) in $\mathcal{A}_2$
(\cref{alg:mmcr}), using the constant-factor subroutines specified below.
Replacing them with $(1+\varepsilon)$-approximation algorithms for SS-MIN
and SS-MAX gives a $(1+\varepsilon)$ guarantee for complete-recharging
makespan. For the maximization problem, this convention means a returned
value of at least $\operatorname{OPT}/(1+\varepsilon)$.
The resulting absolute error $C_{\max}-C^*_{\max}$ depends on $\varepsilon$
and the battery capacity, and not on the number of jobs. In the SS-MAX regime it equals the shortfall in final-group energy, which is at most $\varepsilon \BM$. In the SS-MIN regime it is at most $\varepsilon x^*$, where $x^*$ is the optimal head energy. This energy can exceed $\BM$: for $b=(7,7,6)$ and $\BM=10$ it is $13$. It always satisfies $x^*<\sum_j b_j\le 2\BM$, so the error is below $2\varepsilon\BM$ in both regimes.

To obtain the constant-factor guarantee for $\mathcal{A}_2$, we can use a
weaker SS-MAX guarantee than the strict reduction alone would suggest.
Whenever $\mathcal{A}_2$ calls SS-MAX, the optimal makespan exceeds $\BM$.
We can therefore bound the loss in final-group energy relative to a
makespan that is larger than the battery capacity.

\begin{theorem}[Guarantee of $\mathcal{A}_2$]\label{thm:a2-guarantee}
Suppose \textup{\textsc{SelectMinHead}} returns a feasible subset whose sum is at most $\frac{5}{4}$ times the SS-MIN optimum, and \textup{\textsc{SelectMaxFinalGroup}} returns a feasible subset whose sum is at least $\frac{3}{4}$ times the SS-MAX optimum. Then \cref{alg:mmcr} returns a feasible schedule with $C_{\max}\le\frac{5}{4}\,C^*_{\max}$.
\end{theorem}

\begin{proof}
If $\sum_j b_j\le\BM$, then the schedule with an empty head is optimal; if $\BM<\sum_j b_j\le2\BM$, then the strict reduction to SS-MIN preserves the ratio $\frac{5}{4}$ (\cref{thm:ssmin-reduct} covers both cases).
Otherwise $\sum_j b_j>2\BM$. Let $\phi$ and $\phi^*$ be the returned and optimal final-group energies, so that $\frac{3}{4}\phi^*\le\phi\le\phi^*\le\BM$. By \cref{lem:recharge-formula},
\begin{align*}
  C_{\max}-C^*_{\max}&=\phi^*-\phi\le\tfrac{1}{4}\phi^*\le\tfrac{1}{4}\BM,\\
  C^*_{\max}&\ge\textstyle\sum_j b_j-\phi^*>\BM .
\end{align*}
Hence $C_{\max}<\frac{5}{4}\,C^*_{\max}$.
\end{proof}

Our implementation uses size-class subroutines with these guarantees and
$O(n\log n)$ running time
(\cref{lem:ssmin-implementation,lem:ssmax-implementation}). We specify
both subroutines and prove their ratios in the appendix
(\cref{apx:a2-subroutines}). Together with linear-time schedule
construction, they give an $O(n\log n)$ implementation of $\mathcal{A}_2$.

We can also substitute other subroutines. Kellerer et
al.~\cite{KELLERER2000289} give SS-MAX algorithms with worst-case ratios
$\frac{3}{4}$ and $\frac{4}{5}$; we did not evaluate these
algorithms. For a tunable ratio, we can combine a fully polynomial-time
approximation scheme (FPTAS) for
SS-MAX~\cite{KellererPferschySperanza1997SubsetSumFPTAS} with one for
SS-MIN~\cite{GUNTZER200055}. The SS-MAX scheme takes
$O\!\left(\min\{n/\varepsilon,\;n+(1/\varepsilon^2)\log(1/\varepsilon)\}\right)$
time, compared with $O(n^2/\varepsilon)$ for the SS-MIN scheme. This
difference in running time is a second reason to use SS-MAX when
$\sum_j b_j>2\BM$.

\begin{algorithm}[tbp]
\caption{Algorithm $\mathcal A_2$: $5/4$-approximation for complete-recharging makespan}
\label{alg:mmcr}

\begin{algorithmic}[1]
\Require Jobs $\mathcal J$; $0\le\Bz\le\BM$ and $b_i\le\BM$ for all $i$
\Ensure A feasible schedule $\sigma$ with $C_{\max}\le(5/4)\,\mathrm{OPT}$

\If{$\sum_j b_j \leq \BM$}
    \State $\mathcal{F} \gets \mathcal{J}$
\ElsIf{$\sum_j b_j \leq 2 \BM$}
    \State $\mathcal{F} \gets \mathcal{J} \setminus \textsc{SelectMinHead}(\mathcal{J}, C = \sum_j b_j - \BM)$
\Else
    \State $\mathcal{F} \gets \textsc{SelectMaxFinalGroup}(\mathcal{J}, C = \BM)$
\EndIf
\State $\mathcal{H} \gets \mathcal{J} \setminus \mathcal{F}$
\State $\sigma \gets \textsc{ScheduleByBipartition}(\mathcal{H}, \mathcal{F})$
\end{algorithmic}
\end{algorithm}

We have now classified all 16 complete-recharging variants without setup
time. Equal energy demands admit polynomial algorithms for the four
objectives, and the other three job classes are NP-hard for each objective.
For makespan, $\mathcal{A}_2$ guarantees a factor of $5/4$. Approximation
of the other objectives in the NP-hard classes remains open
(\cref{sec:conclusions}). We next consider partial recharging, where the
scheduler also chooses the recharge amounts.

\section{Partial Recharging without Setup Time}
\label{sec:partial}

Having classified complete recharging, we now allow the device to choose
how much energy it acquires before each job. We first determine whether
choosing these amounts requires a separate optimization once the job order
is known. It does not: without setup time, acquiring only the energy that
each job lacks minimizes every completion time for that order. Each
objective then reduces to the choice of a job order.

\begin{lemma}[Minimum recharging]
\label{lem:min-recharge}
For any monotonically non-decreasing objective function $\varphi$, there exists
an optimal schedule in which the total recharging before job $J_i$ (in position
$i$ of the schedule) equals
\begin{equation}
  \max\!\Bigl\{0,\; \sum_{j=1}^{i} b_j - \Bz\Bigr\}.
  \label{eq:min-recharge}
\end{equation}
\end{lemma}

\begin{proof}
Fix a job ordering and write $P_i=\sum_{j=1}^i b_j$.
Every feasible schedule in this order must replenish at least
$R_i=\max(0,P_i-\Bz)$ before completing job $i$, by energy conservation.
Set $R_0=0$ and recharge for $r_i=R_i-R_{i-1}$ immediately before job $i$.
If the remaining initial charge covers $b_i$, then $r_i=0$; otherwise,
$r_i$ replenishes exactly the deficit. Since $b_i\le\BM$, this construction
respects capacity and never exhausts the battery during execution.
It attains each cumulative lower bound $R_i$ and hence minimizes every
completion time for the fixed order. Applying it to the order of an optimal
schedule cannot increase any coordinatewise non-decreasing objective.
\end{proof}

We can now determine the recharge durations from any given job order
(\cref{lem:min-recharge}). To construct an optimal schedule, we therefore
need to find an optimal permutation and insert the corresponding recharges.
To analyze how a permutation affects the completion times, we distinguish
the jobs that execute on the initial charge from those that require
recharging.

We call $J_{\pi_k}$ the \emph{separation job} of schedule $\pi$ if it is the
first job that requires a positive recharging interval, that is, $r_k > 0$
and $r_j = 0$ for all $j < k$.
The jobs before the separation job form the \emph{prefix} and the jobs from the
separation job onward form the \emph{suffix}.
If $\beta$ units of initial charge remain before the separation job, then
that job needs $b_{\pi_k}-\beta$ units of recharging. Every subsequent job
$J_{\pi_j}$ ($j>k$) needs exactly $b_{\pi_j}$, so its execution and
recharging together take $e_{\pi_j}+b_{\pi_j}$ time units.

\subsection{Minimizing Sum of Completion Times (\SCT{})}
\label{sec:partial-sct}

\subsubsection{Structured Execution Times and Energy Costs}
\label{sec:partial-sct-poly}

We first consider whether SJF order still minimizes the sum of completion
times. It does when ordering the jobs by execution time also orders them
by execution plus recharging time. To prove this, we compare adjacent jobs
and show that removing an inversion cannot increase the objective, even
when the exchange changes which job first requires a recharge.

\begin{lemma}[SJF under partial recharging]
\label{lem:sjf-partial}
Suppose that $e_j < e_i \Rightarrow e_j + b_j < e_i + b_i$ for every pair
$(i,j)$. Then, under partial recharging, SJF with ties broken by smaller
$b_i$ minimizes the sum of completion times.
\end{lemma}

\begin{proof}
Consider adjacent jobs $i,j$ in a minimum-recharge schedule and let $\gamma$
be the remaining battery immediately before the pair. The first completion
increment is
$e_i+\max(0,b_i-\gamma)=\max(e_i,e_i+b_i-\gamma)$.
If $e_j\le e_i$ and $e_j+b_j\le e_i+b_i$, then placing $j$ first cannot increase
this increment. The swap preserves cumulative execution and energy demand at the
second position and every later position, so their completions do not change. The assumed condition, with ties in $e_i$ broken
by smaller $b_i$, supplies both inequalities at every adjacent inversion.
Removing the inversions proves the claim, including exchanges that move the
separation position.
\end{proof}

\begin{theorem}[\SCT{}, structured job classes]
\label{thm:sct-partial-poly}
Minimizing the sum of completion times under partial recharging is solvable in
$O(n \log n)$ time when $e_i < e_j \Rightarrow e_i + b_i < e_j + b_j$.
This condition holds for each of the following cases: $b_i = b$, $e_i = e$,
and $e_i = b_i$.
\end{theorem}

\begin{proof}
We choose minimum recharging durations for each ordering
(\cref{lem:min-recharge}) and obtain an optimal schedule by sorting in SJF
order (\cref{lem:sjf-partial}).
Sorting in SJF order takes $O(n \log n)$.
Each of the three stated cases satisfies the precondition: for $b_i = b$ or
$e_i = e$ or $e_i = b_i$, a direct calculation confirms
$e_j < e_i \Rightarrow e_j + b_j < e_i + b_i$.
\end{proof}

The three structured job classes are therefore no harder under partial
recharging than they are without a battery. The arbitrary class does not
satisfy the condition, and we show next that it is NP-hard.

\subsubsection{Arbitrary $(e_i, b_i)$: NP-hardness}
\label{sec:partial-sct-arb}

When $\Bz=0$, the device must replenish the full energy demand of every
job, so each job has an effective duration of $e_i+b_i$. Sorting the jobs
by this duration, which we call \emph{shortest-sum-first} (SSF) order,
then minimizes the sum of completion times in $O(n\log n)$ time. With
$\Bz>0$, the scheduler must also decide which jobs execute on the initial
charge without recharging. We call this initial recharge-free interval the
\emph{free window}, and we use the choice of jobs within it to construct
the hardness reduction.

\paragraph{Structure of optimal schedules.}

Before we construct the reduction, we establish how an optimal schedule
orders the jobs on either side of the first recharge. SJF order applies
before that recharge, and SSF order applies after the job that requires
it. We must treat that job separately, because it can still use some of
the initial charge.

\begin{lemma}[Optimal schedule structure]
\label{lem:structure}
For the \SCT{} problem under partial recharging with arbitrary $(e_i,b_i)$,
there exists an optimal minimum-recharge schedule with the following properties.
If there is no separation job, then all jobs are in SJF order. Otherwise, writing
$J_{\pi_k}$ for the separation job:
\begin{enumerate}[label=(\alph*)]
  \item \emph{SJF prefix.} The jobs strictly before position $k$ are in
        non-decreasing execution time.
  \item \emph{SSF tail.} The jobs strictly after position $k$ are in
        non-decreasing $e_i+b_i$. The separation job need not satisfy this order.
  \item \emph{Boundary conditions.} Let
        $\beta=\Bz-\sum_{j<k}b_{\pi_j}$, so $0\le\beta<b_{\pi_k}$.
        For each neighbor that exists, optimality requires
        \begin{align}
          &e_{\pi_{k-1}}+\min\{b_{\pi_{k-1}},b_{\pi_k}-\beta\}\notag\\
          &\quad\le e_{\pi_k}+b_{\pi_k}-\beta \quad (k>1),\notag\\
          &e_{\pi_k}+b_{\pi_k}-\beta\notag\\
          &\quad\le e_{\pi_{k+1}}+\max\{0,b_{\pi_{k+1}}-\beta\} \quad (k<n).
          \label{eq:boundary}
        \end{align}
\end{enumerate}
\end{lemma}

\begin{proof}
We sort the recharge-free prefix by execution time and the tail strictly
after the separation job by execution plus recharging time. Adjacent exchanges
at the separation job give the boundary inequalities (\cref{apx:lem:structure}).
\end{proof}

The separation job need not follow SSF order, because its recharge can be
shorter than its full energy demand. For
$\Bz=10$, $\BM=11$, and jobs $(e,b)=(1,11),(5,0)$, scheduling the job $(1,11)$ first gives total completion
time $2+7=9$, whereas SSF gives $5+7=12$. Its first recharge is only
$11-10=1$, not~11.

\paragraph{The NP-hardness reduction.}

We use Partition to show that choosing the jobs that benefit from the
initial charge is NP-hard. The reduction makes a schedule meet a cost
threshold exactly when the input integers admit an equal partition.

\begin{definition}[Partition]
Given positive integers $a_1, \ldots, a_n$ with $\sum_{i=1}^n a_i = 2T$,
decide whether there is a subset $X_1 \subseteq \{1, \ldots, n\}$ such that
$\sum_{i \in X_1} a_i = T$.
\end{definition}

\begin{theorem}[Hardness of \SCT{}, arbitrary parameters]
\label{thm:sct-arb-nph}
Minimizing the sum of completion times under partial recharging with arbitrary
$(e_i, b_i)$ and $\Bz > 0$ is binary NP-hard.
\end{theorem}

\begin{proof}
We reduce from Partition using two copies of each of $n+1$ job types and
one anchor. To compare different choices of jobs before the anchor, we
first define an auxiliary cost that allows a signed recharge term. This
expression is not the cost of a feasible schedule, so we derive an exact
correction before we use it in the reduction. We need the correction
because an exchange that increases the auxiliary cost can decrease the
feasible cost.

Let $a_1,\ldots,a_n$ be positive integers with $\sum_i a_i=2T$, and choose
\begin{equation}
  \Lambda=(4n+8)T+1.
  \label{eq:Lambdabound}
\end{equation}
Define $S_i=\sum_{j=1}^i a_j$ for $0\le i\le n$, with $S_0=0$ and $S_n=2T$.
For each type $i$, create two identical jobs $J_i,J_{n+1+i}$ with
\begin{equation}
  e_i=e_{n+1+i}=S_i,\qquad
  b_i=b_{n+1+i}=\Lambda-2S_i.
  \label{eq:job-params}
\end{equation}
Create an anchor $J_{2n+2}$ with
\begin{equation}
  e_{2n+2}=\Lambda-2T=:A,\qquad b_{2n+2}=0,
  \label{eq:anchor}
\end{equation}
and set
\begin{equation}
  \BM=\sum_{i=0}^n(\Lambda-2S_i),\qquad \Bz=\BM-2T.
  \label{eq:battery}
\end{equation}
All energy demands are nonnegative and at most $\BM$.
The anchor has strictly largest $e_i$ and smallest $e_i+b_i$, tied in the
latter quantity with the two jobs of type $n$.
The choice of $\Lambda$ also gives
\begin{equation}
  n\Lambda<\Bz< (n+2)(\Lambda-4T).
  \label{eq:partial-count-bounds}
\end{equation}
Thus any $n$ non-anchor jobs fit within the initial charge, whereas any
$n+2$ exceed it.

\noindent\emph{Reference ordering and threshold.}
The reference ordering is
\begin{equation}
  \pi^0=(J_0,J_1,\ldots,J_n,\ J_{2n+2},\ J_{2n+1},J_{2n},\ldots,J_{n+1}).
  \label{eq:pi0}
\end{equation}
Its first $n+1$ jobs consume $\BM=\Bz+2T$.
The following baseline charges this deficit to the anchor and later jobs,
but omits its contribution to the completion of $J_n$:
\begin{align}
  C_{\mathrm{ref}}
  &=\sum_{j=0}^n(n-j+1)S_j\notag\\
  &\quad+(n+2)\!\left(\sum_{j=0}^n S_j+2T+(\Lambda-2T)\right)\notag\\
  &\quad+\sum_{j=0}^n(j+1)(\Lambda-S_j).
  \label{eq:Cref}
\end{align}
The feasible minimum-recharge cost of $\pi^0$ is $C_{\mathrm{ref}}+2T$.
We use decision threshold $C_{\mathrm{ref}}+T$.

\noindent\emph{Auxiliary cost.}
Write $m=n+1$. A \emph{central ordering} has $m$ non-anchor jobs on each side
of the anchor. Let $\ell_i\in\{0,1,2\}$ be the number of type-$i$ jobs on
its left, with $\sum_i\ell_i=m$. Sort the left block by increasing $S_i$
and the right block by decreasing $S_i$. Define
\begin{align}
  H_k&=\sum_{i=0}^{k-1}(\ell_i-1),\quad 0\le k\le m,\qquad H_0=H_m=0,\notag\\
  d&=\sum_{i=0}^n(\ell_i-1)S_i
     =-\sum_{k=1}^n a_k H_k,\qquad
  \delta=2(T-d).
  \label{eq:partial-imbalance}
\end{align}
The left block has energy demand $\Bz+\delta$. Let its ordered jobs
be $x_1,\ldots,x_m$, let the ordered right block be $y_1,\ldots,y_m$, and
write $E_\ell=\sum_i\ell_iS_i$. Define the auxiliary cost
\begin{align}
  F(\ell)&=\sum_{r=1}^m(m-r+1)e_{x_r}
      +(m+1)(E_\ell+\delta+A)\notag\\
    &\quad+\sum_{r=1}^m(m-r+1)(e_{y_r}+b_{y_r}).
  \label{eq:partial-auxiliary}
\end{align}
This expression assigns no recharge to left-block completions and assigns
the signed amount $\delta$ to the anchor and subsequent completions.
In particular, it is well defined even when $\delta<0$.

\begin{lemma}[Auxiliary and feasible costs]
\label{lem:partial-auxiliary}
For every central ordering sorted as above,
\begin{align}
  F(\ell)&=C_{\mathrm{ref}}+\sum_{k=1}^n a_kH_k^2,\notag\\
  C(\ell)&=F(\ell)+|\delta|
          =C_{\mathrm{ref}}+\sum_{k=1}^n a_kH_k^2+2|T-d|,
  \label{eq:partial-cost-identity}
\end{align}
where $C(\ell)$ is its feasible minimum-recharge cost.
\end{lemma}
\begin{proof}
We first expand the auxiliary cost and telescope its coefficients to
obtain the first identity. To recover the feasible cost, we then correct
the completions of the last left-block job and the anchor. These corrections
sum to $|\delta|$, giving the second identity
(\cref{apx:lem:partial-auxiliary}).
\end{proof}

\noindent\emph{Regular and generalized exchanges.}
An exchange replacing one left-block job of type $v$ with one right-block
job of type $u>v$, followed by sorting both blocks, reduces left-block energy
by $2(S_u-S_v)$ and changes the auxiliary cost by
\begin{equation}
  F(\ell')-F(\ell)=\sum_{k=v+1}^{u}a_k(1-2H_k).
  \label{eq:partial-general-swap}
\end{equation}
Indeed, exactly those $H_k$ decrease by one. The reverse exchange instead
changes the auxiliary cost by $\sum_{k=v+1}^{u}a_k(1+2H_k)$, evaluated at
its own starting state. These are exact identities, with no assumption on
earlier exchanges. A regular forward exchange has $u=v+1=p$, so its
auxiliary increase is $a_p$ when $H_p=0$. The feasible-cost change is always
\begin{equation}
  C(\ell')-C(\ell)=F(\ell')-F(\ell)
       +2\bigl(|T-d'|-|T-d|\bigr).
  \label{eq:partial-feasible-swap}
\end{equation}
For $a=(2,2)$ and $\Lambda=33$, this gives $F:346\to348$ and
$C:350\to348$ for the regular exchange at $p=1$.
Thus a regular exchange need not increase the feasible objective.

\noindent\emph{The ``if'' direction.}

\begin{lemma}[A partition yields a schedule]
\label{lem:if-dir}
If Partition has a solution, then the constructed instance has a feasible schedule
with cost exactly $C_{\mathrm{ref}}+T$.
\end{lemma}
\begin{proof}
For a solution $X\subseteq\{1,\ldots,n\}$, set
$H_k=-1$ for $k\in X$ and $H_k=0$ otherwise, together with $H_0=H_m=0$.
Then $\ell_i=1+H_{i+1}-H_i\in\{0,1,2\}$ and $\sum_i\ell_i=m$, so these
counts define a central ordering. Here
$d=-\sum_k a_kH_k=T$ and $\sum_k a_kH_k^2=T$.
The resulting feasible cost is $C_{\mathrm{ref}}+T$ (\cref{lem:partial-auxiliary}). Equivalently,
apply regular exchanges for the selected indices in increasing order;
the count vector above shows that every exchange has an available outgoing
and incoming job. Each changes $F$ by $a_p$, and the final correction
$2|T-d|$ is zero.
\end{proof}

\noindent\emph{Anchor position and the ``only if'' direction.}

We have constructed a schedule from a partition. For the converse, we need
to show that every schedule that meets the cost threshold yields a
partition. We first restrict the positions that the anchor can occupy in
an optimal schedule, and then account for each position.

\begin{lemma}[Position of the anchor job]
\label{lem:anchor-pos}
There is an optimal schedule satisfying \cref{lem:structure} in which the
anchor is either at position $n+2$, or at position $n+3$ immediately after
the separation job. In the latter case, $n+1$ non-anchor jobs precede the
separation job.
\end{lemma}
\begin{proof}
We order only jobs strictly after the separation job and resolve the
anchor's SSF ties within that tail (\cref{apx:lem:anchor-pos}).
\end{proof}

\begin{lemma}[A schedule yields a partition]
\label{lem:only-if}
If the constructed instance has a feasible schedule with cost at most
$C_{\mathrm{ref}}+T$, then Partition has a solution.
\end{lemma}
\begin{proof}
For a central ordering, we obtain the lower bound from the integer inequality $H_k(H_k+1)\ge0$:
\begin{equation}
  C(\ell)-(C_{\mathrm{ref}}+T)
   =\sum_{k=1}^n a_kH_k(H_k+1)
       +d-T+2|T-d|\ \ge\ 0.
  \label{eq:partial-threshold-gap}
\end{equation}
Equality requires $d=T$ and every $H_k\in\{0,-1\}$; the indices with
$H_k=-1$ form the required subset. We establish the corresponding lower bound for an anchor at position
$n+3$ in the appendix (\cref{apx:lem:only-if}).
\end{proof}

The constructed instance meets the threshold if and only if Partition has a solution (\cref{lem:if-dir,lem:only-if}).
There are $2n+3$ jobs, and all parameters, including the threshold defined
through (\ref{eq:Cref}), have
bit lengths polynomial in the input length. The reduction is therefore
polynomial in the binary encoding of the Partition instance.
\end{proof}

\subsubsection{Arbitrary $(e_i, b_i)$: 2-Approximation}
\label{sec:partial-sct-approx}

The hardness result motivates an approximation algorithm for arbitrary
execution times and energy demands. \algAone{} uses SSF order to select a
prefix that fits within the initial charge, and then reorders the jobs in
that prefix by execution time (\cref{alg:A1}). The resulting schedule
follows SJF order where no recharge is needed and SSF order for the
remaining jobs. We prove the factor of two by bounding the completion
times of the prefix and of the suffix separately.

\begin{algorithm}[tbp]
\caption{Algorithm $\mathcal{A}_1$: 2-approximation for \SCT{} under partial recharging}
\label{alg:A1}
\begin{algorithmic}[1]
\Require Jobs $\mathcal{J} = \{J_1,\ldots,J_n\}$ with arbitrary $(e_i,b_i)$;
         initial battery $\Bz$
\Ensure A feasible schedule with $\SCT \le 2 \cdot \mathrm{OPT}$
\State Sort jobs so $e_1+b_1 \le e_2+b_2 \le \cdots \le e_n+b_n$
       \Comment{SSF order: $O(n\log n)$}
\State $k\gets\min\bigl(\{i\in\{1,\ldots,n\}:\sum_{j=1}^i b_j>\Bz\}\cup\{n+1\}\bigr)$
       \Comment{$k=n+1$ if no recharge is needed}
\State Re-sort $\{J_1,\ldots,J_{k-1}\}$ in non-decreasing order of $e_i$
       \Comment{SJF order: $O(n\log n)$}
\State Output this prefix followed by $J_k,\ldots,J_n$ (an empty tail if $k=n+1$)
\State Before each job, recharge only its energy deficit, as in \cref{lem:min-recharge}
\end{algorithmic}
\end{algorithm}

\begin{theorem}[Guarantee of $\mathcal{A}_1$]
\label{thm:2approx}
\algAone{} runs in $O(n \log n)$ time and produces a feasible
schedule with $\SCT \le 2 \cdot \mathrm{OPT}$.
\end{theorem}

\begin{proof}
Let $\pi$ be the schedule that $\mathcal{A}_1$ produces and let $\pi^*$ be an
optimal minimum-recharge schedule, which exists by \cref{lem:min-recharge}.
If all jobs fit within $\Bz$, then the algorithm returns SJF with no recharging,
which is optimal. Otherwise $1\le k\le n$ below.
Write $w_{(1)} \le w_{(2)} \le \cdots \le w_{(n)}$ for the values $e_i+b_i$
sorted in non-decreasing order, and let $k$ be the separation index of $\pi$,
so the prefix $\{J_{\pi_1},\ldots,J_{\pi_{k-1}}\}$ consists of the $k-1$ jobs of
smallest $e_i+b_i$ and satisfies
$\sum_{j<k} b_{\pi_j} \le \Bz < \sum_{j\le k} b_{\pi_j}$.

\emph{A universal lower bound on completion times.}
For any feasible schedule $\sigma$ and any position $i$, total recharging
through the first $i$ jobs is at least
$\max(0,\sum_{j\le i} b_{\sigma_j}-\Bz)$ (\cref{lem:min-recharge}), so
\begin{align}
  C_i(\sigma)
   &\ge \sum_{j\le i}\bigl(e_{\sigma_j}+b_{\sigma_j}\bigr)
     - \min\!\Bigl(\textstyle\sum_{j\le i} b_{\sigma_j},\,\Bz\Bigr)\notag\\
   &\ge \sum_{l=1}^{i} w_{(l)} - \Bz,
  \label{eq:lb-completion}
\end{align}
because the first $i$ jobs of $\sigma$ have total $e+b$ at least that of the $i$
smallest such values, and the subtracted term is at most $\Bz$.

\emph{Suffix bound.}
In $\pi$ the prefix holds the $k-1$ jobs of smallest $e_i+b_i$ and the suffix
lists the remaining jobs in non-decreasing $e_i+b_i$.
Hence for a suffix position $i\ge k$ the first $i$ scheduled jobs are precisely
the $i$ jobs of smallest $e_i+b_i$, so
$\sum_{j\le i}(e_{\pi_j}+b_{\pi_j})=\sum_{l=1}^i w_{(l)}$; moreover their
cumulative battery already exceeds $\Bz$ (since $\sum_{j\le k} b_{\pi_j}>\Bz$),
so the subtracted term in (\ref{eq:lb-completion}) equals $\Bz$ and
\[
  C_i(\pi) = \sum_{l=1}^{i} w_{(l)} - \Bz .
\]
Comparing with (\ref{eq:lb-completion}) applied to $\pi^*$ gives
$C_i(\pi)\le C_i(\pi^*)$ for every $i\ge k$, hence
\begin{equation}
  \sum_{i=k}^{n} C_i(\pi) \;\le\; \sum_{i=k}^{n} C_i(\pi^*) \;\le\; \mathrm{OPT}.
  \label{eq:suffix-bound}
\end{equation}

\emph{Prefix bound.}
The prefix jobs $P=\{J_{\pi_1},\ldots,J_{\pi_{k-1}}\}$ run first, in SJF order
and without any recharging, so $\sum_{i<k} C_i(\pi)$ is the sum of completion
times of $P$ scheduled alone by SJF.
In $\pi^*$, deleting every job outside $P$ and every recharging interval only
decreases the completion time of each job of $P$; the result schedules the jobs
of $P$ in some order without recharging, whose total completion time is at least
the minimum over such schedules, namely the SJF value $\sum_{i<k} C_i(\pi)$.
Consequently
\begin{equation}
  \sum_{i=1}^{k-1} C_i(\pi) \;\le\; \sum_{J_j\in P} C_j(\pi^*) \;\le\; \mathrm{OPT}.
  \label{eq:prefix-bound}
\end{equation}

\emph{Combining.}
Adding (\ref{eq:suffix-bound}) and (\ref{eq:prefix-bound}) gives
$\sum_{i=1}^n C_i(\pi) \le 2 \cdot \mathrm{OPT}$.
The two sorting steps dominate the running time, each taking $O(n \log n)$.
\end{proof}

\begin{remark}[Tightness]
The ratio 2 is tight.
Consider two jobs $J_1 = (1, \lambda)$ and $J_2 = (\lambda-1, 1)$ with
$\Bz = \BM = \lambda$, for a large parameter $\lambda$.
Since $e_2 + b_2 = \lambda < 1 + \lambda = e_1 + b_1$, \algAone{}
places $J_2$ in the prefix and $J_1$ in the suffix, producing
$\pi = (J_2,\, [\text{recharge}],\, J_1)$.
The completion times are $\lambda - 1$ for $J_2$ and
$(\lambda - 1) + 1 + 1 = \lambda + 1$ for $J_1$, giving $\SCT = 2\lambda$.
The schedule $(J_1,\, [\text{recharge}],\, J_2)$ achieves completion times
$1$ for $J_1$ and $1 + 1 + (\lambda - 1) = \lambda + 1$ for $J_2$,
giving $\SCT = \lambda + 2$.
The ratio $2\lambda / (\lambda + 2)$ approaches 2 as $\lambda \to \infty$.
\end{remark}

\subsubsection{Arbitrary $(e_i, b_i)$: Enumerating the Final Job}
\label{sec:final-job}

We can strengthen the bound for each fixed $n$ by enumerating the final
job. For each job $j$, we run $\mathcal A_1$ on the other jobs with the
same $\Bz$ and $\BM$, append $j$, and recharge only the deficits; we
return the best of the $n$ schedules. Two sorts and a Boolean membership
array make each candidate cost $O(n)$, so the enumeration takes $O(n^2)$
time.

\begin{theorem}[Final-job enumeration]
\label{thm:final-job}
Under partial recharging without setup time, let $M=E+\max\{0,S-\Bz\}$ be
the makespan of every minimum-recharge schedule. The enumeration returns a
schedule whose sum of completion times $F$ satisfies
\[
 F\leq2\,\mathrm{OPT}-M\leq\Bigl(2-\frac{2}{n+1}\Bigr)\mathrm{OPT},
\]
and it is never worse than $\mathcal A_1$.
\end{theorem}
\begin{proof}
Let $j^*$ be the final job of an optimal minimum-recharge schedule. It
completes at $M$ whatever precedes it, and its predecessors form an optimal
schedule of the instance without $j^*$: a better order of those jobs,
followed by $j^*$, would improve the whole objective. Hence
$\mathrm{OPT}(\mathcal J\setminus\{j^*\})=\mathrm{OPT}-M$, and the
candidate that omits $j^*$ costs at most
$2\,\mathrm{OPT}(\mathcal J\setminus\{j^*\})+M=2\,\mathrm{OPT}-M$.

For the second inequality, write $R=\max\{0,S-\Bz\}$ and suppose $S>0$.
Convexity gives $\max\{0,x-\Bz\}\leq xR/S$ for $0\leq x\leq S$. In a
uniformly random order, the first $i$ jobs have expected execution time
$iE/n$ and expected demand $iS/n$, so the expected completion at position
$i$ is at most $iM/n$. Summing over positions, some order costs at most
$(n+1)M/2$, and therefore $\mathrm{OPT}\leq(n+1)M/2$. If $S=0$, then no
charging occurs and the execution times alone give the same bound.

Finally, $\mathcal A_1$'s own schedule is one of the candidates: omitting
its last job leaves its prefix selection and the order of the other jobs
unchanged.
\end{proof}

The enumeration is necessary in the following sense: taking the better of
$\mathcal A_1$ and SJF does not lower the factor.
For $m\geq2$, take $m$ jobs $(e,b)=(2,0)$, one job $(1,m^3+m^2)$, one job
$(m^3,0)$, $\Bz=m^3$, and $\BM=m^3+m^2$. The two rules cost
$2m^3+2m^2+5m+1$ and $2m^3+3m^2+4m+2$ against the optimum
$m^3+3m^2+5m+2$, so both ratios tend to two, whereas the enumeration is
exact on this family. The enumeration is not exact in general: jobs
$(1,20)$, $(11,14)$, $(16,0)$ with $\Bz=17$ and $\BM=20$ give 78 against
the optimum 69. A uniform factor below two would follow from the bound
$\mathcal A_1\leq\mathrm{OPT}+M$, but that bound does not hold: ten jobs
$(100,100)$ and four jobs $(0,201)$ with $\Bz=\BM=1000$ have
$\mathcal A_1-\mathrm{OPT}=2{,}374>1{,}804=M$. Whether some polynomial
algorithm achieves a constant factor below two, independent of $n$, is
open.

\subsection{Minimizing Maximum Tardiness ($T_{\max}$)}

To minimize the maximum tardiness, we need to determine whether EDF order
remains optimal when jobs require different amounts of recharging. We show
that exchanging two adjacent jobs that violate deadline order changes
neither the completion time of the second of them nor that of any later
job, and that the exchange cannot increase the maximum tardiness of the
pair. EDF order is therefore optimal even when execution times and energy
demands vary independently.

\begin{theorem}[EDF under partial recharging]
\label{thm:maxt-partial}
Minimizing maximum tardiness under partial recharging is solvable in
$O(n \log n)$ time for all four energy-execution relationships.
\end{theorem}

\begin{proof}
By \cref{lem:min-recharge}, the total recharging up to and including any
position depends only on the \emph{set} of jobs scheduled there, not on their
internal order.
An inversion is a pair $J_i$, $J_j$ with $d_i < d_j$ but $J_j$ scheduled
immediately before $J_i$, at positions $p$ and $p{+}1$.
Swapping them leaves the set of the first $p{+}1$ jobs unchanged, so the
swap preserves the completion time $C$ of position $p{+}1$ and every later
completion. Only the two swapped jobs can change their completion times.
Let $C_1, C_1' \le C$ be the completion time of position $p$ before and after
the swap.
Before the swap the pair's maximum tardiness is
$\max(C_1 - d_j,\, C - d_i)^+ = (C - d_i)^+$, since $C_1 \le C$ and $d_i < d_j$.
After the swap it is $\max(C_1' - d_i,\, C - d_j)^+$, and both terms are at
most $C - d_i$ (because $C_1' \le C$ and $d_j > d_i$).
The swap therefore does not increase the maximum tardiness.
Iterating over all inversions yields an EDF schedule, which is thus optimal.
Sorting by deadline runs in $O(n \log n)$.
\end{proof}

\subsection{Minimizing Number of Tardy Jobs ($U$)}

\subsubsection{Cases $b_i = b$, $e_i = e$, and $e_i = b_i$}

To minimize the number of tardy jobs, we must choose which jobs meet their
deadlines as well as the order in which to execute them. In the three
structured job classes, we can build a largest timely subset by scanning
the jobs in EDF order: whenever adding a job causes a deadline miss, we
remove a selected job of largest weight, where the weight depends on the
job class. This rule adapts the Moore--Hodgson
algorithm~\cite{moore1968}; we justify the choice of weight in the proof.

\begin{theorem}[Tardy jobs, structured job classes]
\label{thm:numtardy-partial-poly}
Minimizing the number of tardy jobs under partial recharging is solvable in
$O(n \log n)$ time when $b_i = b$, when $e_i = e$, or when $e_i = b_i$.
\end{theorem}

\begin{proof}
By \cref{lem:min-recharge}, the total recharging before the $j$-th job of a
schedule equals $\max(0,\sum_{k\le j} b_k - \Bz)$ and depends only on the
\emph{set} of jobs in that prefix. Hence the completion time of the $j$-th
scheduled job is
\[
  C_j = E_j + \max(0,\, P_j - \Bz),
\]
where $E_j=\sum_{k\le j} e_k$ and $P_j=\sum_{k\le j} b_k$.
Since $C_j$ is nondecreasing along the schedule and the prefix recharge depends
only on the set, the adjacent-swap argument of \cref{thm:maxt-partial} shows
that a set of jobs can all meet their deadlines in some order if and only if it
does so in EDF order; we may therefore restrict attention to EDF order and seek
a maximum-cardinality \emph{timely set}.

In each of the three cases the completion time depends on one cumulative
quantity and on the position. Write $w_i$ for the quantity that varies
($e_i$ when $e_i=b_i$ or $b_i=b$, and $b_i$ when $e_i=e$) and
$W_j=\sum_{k\le j}w_k$. Then $C_j=\Phi_j(W_j)$, where
\[
  \Phi_j(W)=
  \begin{cases}
    W+\max(0,\,W-\Bz), & e_i=b_i,\\
    W+\max(0,\,jb-\Bz), & b_i=b,\\
    je+\max(0,\,W-\Bz), & e_i=e.
  \end{cases}
\]
Each $\Phi_j(W)$ is non-decreasing in both $j$ and $W$.
For $e_i=b_i$ it does not depend on $j$, and the problem is the classical
single-machine tardy-job problem $1\,\|\,\sum U_j$, where $U_j$ indicates
whether job $j$ is tardy. In that formulation, job $i$ has deadline
$\max\{W:\Phi_j(W)\le d_i\}$.
In the other two cases the bound on $W_j$ depends on the position of the job
within the timely set, so the instance is not a classical one.
The Moore--Hodgson rule remains optimal because its exchange argument uses
only the two monotonicity properties.

The rule scans the jobs in EDF order and, whenever the current job's
completion exceeds its deadline, discards the selected job of largest weight
(\cref{alg:moore-hodgson}). In all three cases the other component of
$e_k+b_k$ is constant or equal to $w_k$, so the job of largest $e_k+b_k$ has
the largest weight; the test uses the true completion $E+\max(0,P-\Bz)$.
Let $q$ be the first job that misses its deadline when we select all jobs
$1,\ldots,q$, and let $k\le q$ have the largest weight among them.
\begin{enumerate}[label=(\roman*)]
  \item Deleting a job from a timely set leaves it timely. Each later job
        moves one position earlier and its cumulative weight decreases, so
        its completion does not increase.
  \item The set $\{1,\ldots,q\}\setminus\{k\}$ is timely. Deleting $k$ preserves the completions before it, and jobs between $k$ and $q$ finish no later.
        If $k\ne q$, then job $q$ now completes at
        $\Phi_{q-1}(W_q-w_k)\le\Phi_{q-1}(W_{q-1})=C_{q-1}\le d_{q-1}\le d_q$,
        since $w_k\ge w_q$.
  \item Some maximum timely set omits $k$. By (i), no timely set contains
        all of $1,\ldots,q$, so a maximum timely set $O$ omits some
        $j\le q$. If $k\in O$, then replace $k$ by $j$. A completion depends only
        on the selected jobs that precede it. The jobs of the new set with
        index at most $q$ form a subset of $\{1,\ldots,q\}\setminus\{k\}$
        and are timely by (i) and (ii). Every later job keeps its position,
        and its cumulative weight changes by $w_j-w_k\le0$.
\end{enumerate}
After discarding $k$, the rule proceeds exactly as it would on the instance
without job $k$: by (i) and (ii) it discards nothing there before job $q+1$.
Induction on $n$ proves optimality. Maintaining $\mathcal{T}$ in a
max-heap keyed by $e_i+b_i$ gives $O(n\log n)$ time.

For arbitrary $(e_i,b_i)$ the position fixes neither $E_j$ nor $P_j$, so
$C_j$ depends on two independent cumulative quantities. Minimizing the number
of tardy jobs is NP-hard in this regime (\cref{thm:numtardy-arb-nph}).
\end{proof}

\begin{algorithm}[tbp]
\caption{Moore--Hodgson rule for partial recharging}
\label{alg:moore-hodgson}
\begin{algorithmic}[1]
\Require Jobs $\mathcal{J}$ in EDF order
\Ensure Maximum-cardinality set $\mathcal{T}$ of timely jobs
\State $\mathcal{T} \gets \emptyset$;\; $E \gets 0$;\; $P \gets 0$
\For{$i = 1$ to $n$}
  \State $\mathcal{T} \gets \mathcal{T} \cup \{J_i\}$;\; $E \gets E + e_i$;\; $P \gets P + b_i$
  \If{$E + \max(0,\, P - \Bz) > d_i$} \Comment{true completion under partial recharging}
    \State Let $J_k \in \mathcal{T}$ have the largest $e_k + b_k$
    \State $\mathcal{T} \gets \mathcal{T} \setminus \{J_k\}$;\; $E \gets E - e_k$;\; $P \gets P - b_k$
  \EndIf
\EndFor
\end{algorithmic}
\end{algorithm}

\subsubsection{Arbitrary $(e_i, b_i)$: NP-hardness and Dynamic Program}

\begin{figure}[tbp]
\small
\begin{align*}
  \eta(m,f,B) &= 0
    && \text{if } m=0,\; f=0,\; B=\Bz, \\
  \eta(m,f,B) &= \infty
    && \text{if } m=0 \text{ and } (f \ne 0 \text{ or } B \ne \Bz), \\
  \eta(m,f,B) &= \infty
    && \text{if } f < 0 \text{ or } B > \BM, \\
  \eta(m,f,B) &= \min\bigl\{\eta(m{-}1,f,B){+}1,\;
                            \eta(m{-}1,f{-}e_m,B{+}b_m)\bigr\}
    && \text{if } B>0,\; f \le d_m, \\
  \eta(m,f,B) &= \min\Bigl\{\eta(m{-}1,f,0){+}1,\;
               \min_{0\le b\le b_m}\eta(m{-}1,f{-}e_m{-}b_m{+}b,b)\Bigr\}
    && \text{if } B=0,\; f \le d_m, \\
  \eta(m,f,B) &= \eta(m{-}1,f,B){+}1
    && \text{if } f > d_m.
\end{align*}
\caption{The Bellman equation for minimizing the number of tardy jobs under partial recharging with arbitrary
$(e_i, b_i)$ (\cref{thm:numtardy-dp}).}
\label{fig:bellman}
\end{figure}

For arbitrary $(e_i,b_i)$, the Moore--Hodgson rule no longer applies,
because the completion time of a job depends on two independent cumulative
quantities. We show that the problem is NP-hard and then give a dynamic
program for it.

\begin{theorem}[Hardness of the tardy count, arbitrary parameters]
\label{thm:numtardy-arb-nph}
Minimizing the number of tardy jobs under partial recharging with arbitrary
$(e_i, b_i)$ is NP-hard.
\end{theorem}

\begin{proof}
We reduce from EquiPartition: given positive integers
$x_1, \ldots, x_{2n}$ with $X = \sum_i x_i$, decide whether $n$ of them sum
to $X/2$. We may assume that $n\ge2$ and $X$ is even, since we can decide the remaining instances directly.
Set $\BM = \Bz = (2n-1)X$ and construct $2n$ jobs with $e_i = x_i$,
$b_i = 2(X - x_i)$, and $d_i = X/2$. Every job is individually feasible,
because $b_i<2X\le\BM$.

Let $Y$ be the set of jobs scheduled first, in any order, and write
$x(Y)=\sum_{i\in Y}x_i$. These jobs consume $2|Y|X-2x(Y)$, so the last of
them completes at
\[
  x(Y)+\max\bigl(0,\;(2|Y|-2n+1)X-2x(Y)\bigr)
\]
(\cref{lem:min-recharge}). For $|Y|=n+1$ this value is
$3X-x(Y)\ge2X$, so at most $n$ jobs meet the deadline.
For $|Y|=n$ it is $\max\{x(Y),\,X-x(Y)\}\ge X/2$, with equality exactly
when $x(Y)=X/2$. Hence a schedule with at most $n$ tardy jobs exists if and
only if EquiPartition has a solution.
\end{proof}

Although the arbitrary case is NP-hard, we can solve it with a dynamic
program whose running time depends on the numerical job parameters. To
decide whether another timely job can be added, the program needs the
finish time and the remaining battery of the last timely job. Its state
therefore holds both quantities, together with the number of jobs
considered (\cref{fig:bellman}).

\begin{theorem}[Pseudo-polynomial dynamic program]
\label{thm:numtardy-dp}
A dynamic program minimizes the number of tardy jobs under partial
recharging with arbitrary $(e_i,b_i)$ (\cref{fig:bellman}).
Its running time is
\[
  O\!\left(n\BM \Bigl(\sum_{i=1}^n(e_i+b_i)\Bigr)\max_{1\le i\le n} b_i\right).
\]
\end{theorem}

\begin{proof}
We process jobs in EDF order. For each job, we consider whether to append
it to the timely set or defer it until after the timely jobs. We retain
the smallest tardy count for each finish-time and battery state, since
these quantities determine whether later jobs can meet their deadlines
(\cref{apx:thm:numtardy-dp}).
\end{proof}

\subsection{Minimizing Makespan ($C_{\max}$)}
\label{sec:partial-makespan}

The total execution time is fixed, so we minimize makespan by minimizing
the total recharging time. Under partial recharging, the device can
replenish exactly the energy that the initial charge does not cover,
whatever the job order. We can therefore compute the optimum from the
total execution time and the total energy demand, and any order attains
it.

\begin{theorem}[Makespan under partial recharging]
\label{thm:makespan-partial}
Under partial recharging, every schedule that recharges only as needed has the
same makespan, independent of the job ordering:
\begin{equation}
  C_{\max} = \sum_{i=1}^n e_i
           + \max\!\Bigl(0,\;\sum_{i=1}^n b_i - \Bz\Bigr).
\end{equation}
In particular, minimizing $C_{\max}$ imposes no constraint on the job ordering,
and the optimum value is computable in $O(n)$ time for all four
energy-execution relationships, or in $O(1)$ when $\sum_i e_i$ and
$\sum_i b_i$ are already available. Constructing a schedule takes $O(n)$ time.
\end{theorem}

\begin{proof}
By \cref{lem:min-recharge}, the minimum total recharging equals
$\max(0, \sum_i b_i - \Bz)$, regardless of job ordering.
The makespan equals $\sum_i e_i$ plus the total recharging; for any schedule
that recharges only as needed it is therefore the same constant, independent of
the ordering.
A schedule may recharge more than necessary, strictly increasing $C_{\max}$;
the value above is the minimum.
\end{proof}

This argument depends on the ability to replenish only the energy needed.
Under complete recharging, the choice of jobs between recharges determines
the total recharging time, and we showed that this choice is NP-hard when
energy demands differ (\cref{thm:makespan-nph}).

Partial recharging without setup time is therefore solvable in polynomial
time in 14 of its 16 variants. The two NP-hard variants both have
arbitrary $(e_i,b_i)$: for the sum of completion times we have a
factor-two guarantee, and for the number of tardy jobs we have a
pseudo-polynomial exact algorithm. The whole analysis rests on
\cref{lem:min-recharge}, which fails as soon as a recharge has a fixed
cost (\cref{sec:setup-partial}). Before we turn to that setting, we
evaluate $\mathcal{A}_1$ and $\mathcal{A}_2$ empirically.

\section{Evaluation without Setup Time}
\label{sec:experiments}

The approximation guarantees hold for every input, but they do not
indicate how close \algAone{} and \algAtwo{} come to the optimum on
particular job sets, or how long the algorithms take to execute. To
examine these questions, we compared their objective values with exact
optima and measured their running times on a desktop computer. We also
varied the battery capacity and the charging rate, to determine how these
parameters affect the approximation quality and the cost of requiring
complete recharging. This sensitivity experiment used 30 synthetic job
sets and six job sets constructed from public current traces of wireless
sensor nodes.

All experiments use the offline model with $h=0$. The traces supply job
demands only; we did not validate a physical battery or charger. The
proofs establish the guarantees. The experiments are only indicative: the
ratios that we report describe the tested job sets and parameter grids,
and they are not bounds for other inputs. (The figures and tables label
the two algorithms A1 and A2.)

\subsection{Empirical Approximation Quality}
\label{sec:exp-approx}

For each of the two approximation algorithms, we computed the exact
optimum of every instance by dynamic programming, and we report the ratio
of the algorithm's objective value to that optimum. For \SCT{} under
partial recharging, the exact solver is a bitmask dynamic program that
takes $O(n \cdot 2^n)$ time; for makespan under complete recharging, it is
a subset-sum dynamic program that takes $O(n \cdot \sum b_i)$ time.

\subsubsection{\SCT{} under Partial Recharging}

We evaluated \algAone{} (2-approximation) on two distributions of
instances with arbitrary $(e_i, b_i)$:
\emph{adversarial} instances that alternate between low-execution/high-battery
and high-execution/low-battery jobs, and
\emph{uniform random} instances with $e_i, b_i \in [1, 50]$.
For both distributions, $\BM=\sum_i b_i$. The alternating instances used
$\Bz=\lfloor\BM/3\rfloor$, and we sampled the low and high values uniformly
from $[1,10]$ and $[50,100]$, respectively. Uniform instances used
$\Bz=\lfloor0.3\BM\rfloor$.
For each distribution, we generated 200 instances at each of
$n\in\{6,8,10,12\}$, 100 at $n=14$, and 30 at $n=16$.
Thus we tested 930 instances from each distribution.

Across the 1,860 instances, the trial-weighted mean ratios were 1.0186
for the alternating distribution and 1.0084 for the uniform distribution,
and the largest observed ratio was 1.1814, at $n=6$
(\cref{fig:approx-ratios,tab:sct-approx}). On these distributions,
$\mathcal{A}_1$ is therefore much closer to the optimum than its
worst-case factor of two. The factor is nevertheless attained in the
limit: the two-job construction that follows \cref{thm:2approx} has ratio
$2\lambda/(\lambda+2)$, which approaches 2 as $\lambda$ grows
(\cref{fig:tightness}). Repeating the alternating structure with more
jobs reduced the observed ratio over the tested sizes. This observation
does not exclude other worst-case families with more jobs.

\subsubsection{Makespan under Complete Recharging}

We evaluated \algAtwo{} ($\nicefrac{5}{4}$-approximation) on
instances with $e_i\in[1,50]$, $b_i\in[1,100]$,
$\BM=\max(\max_i b_i,\lfloor\sum b_i/2\rfloor)$, and $\Bz=0$.
Our implementation uses size-class enumeration subroutines that meet the
requirements for the makespan guarantee (\cref{thm:a2-guarantee}): a
$\nicefrac{5}{4}$-approximation for SS-MIN and a $\nicefrac{3}{4}$-approximation
for SS-MAX (\cref{lem:ssmin-implementation,lem:ssmax-implementation}).
We used 200 instances at each of $n\in\{6,8,10,12,14\}$ and 100 at each of
$n\in\{16,18,20\}$, for 1,300 instances in total.

The trial-weighted mean ratio was 1.0119 and the largest observed ratio
was 1.0635 (\cref{fig:approx-ratios}(b) and \cref{tab:makespan-approx}).
Both values are well below the bound of~$\nicefrac{5}{4}$. One reason is
that the total execution time contributes equally to the makespan of
$\mathcal{A}_2$ and to the optimal makespan, which reduces the makespan
ratio relative to the ratio of the subset-sum subroutine
(\cref{thm:ssmin-reduct}).

\begin{figure}[tbp]
  \centering
  \includegraphics[width=0.95\linewidth]{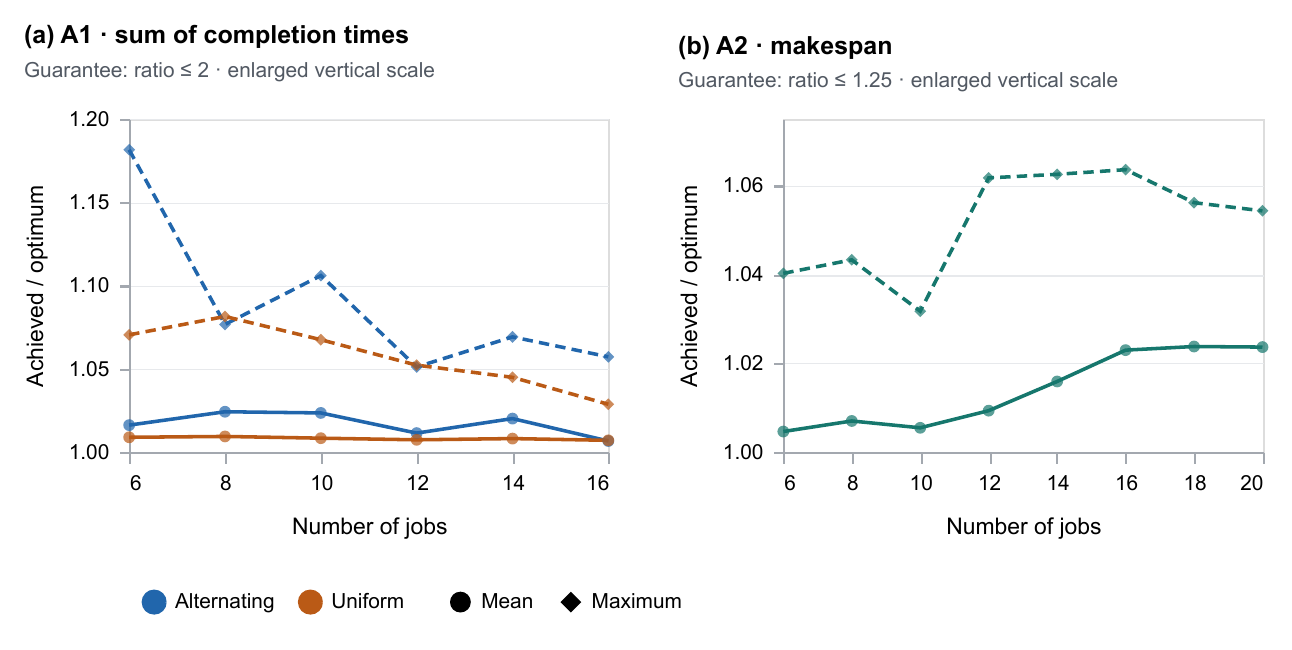}
  \caption{Objective value achieved by the algorithm divided by the exact optimum versus job count:
  (a)~\SCT{} of $\mathcal{A}_1$ under partial recharging, and
  (b)~makespan of $\mathcal{A}_2$ under complete recharging.
  Solid lines show means and dashed lines show maxima at each size.
  Across the sampled instances, the largest ratios were 1.1814 and 1.0635,
  respectively. We enlarged the vertical scales to show these differences;
  the theoretical bounds of 2 and $5/4$ lie above the plotted ranges.}
  \label{fig:approx-ratios}
\end{figure}

\begin{figure}[tbp]
  \centering
  \includegraphics[width=0.95\linewidth]{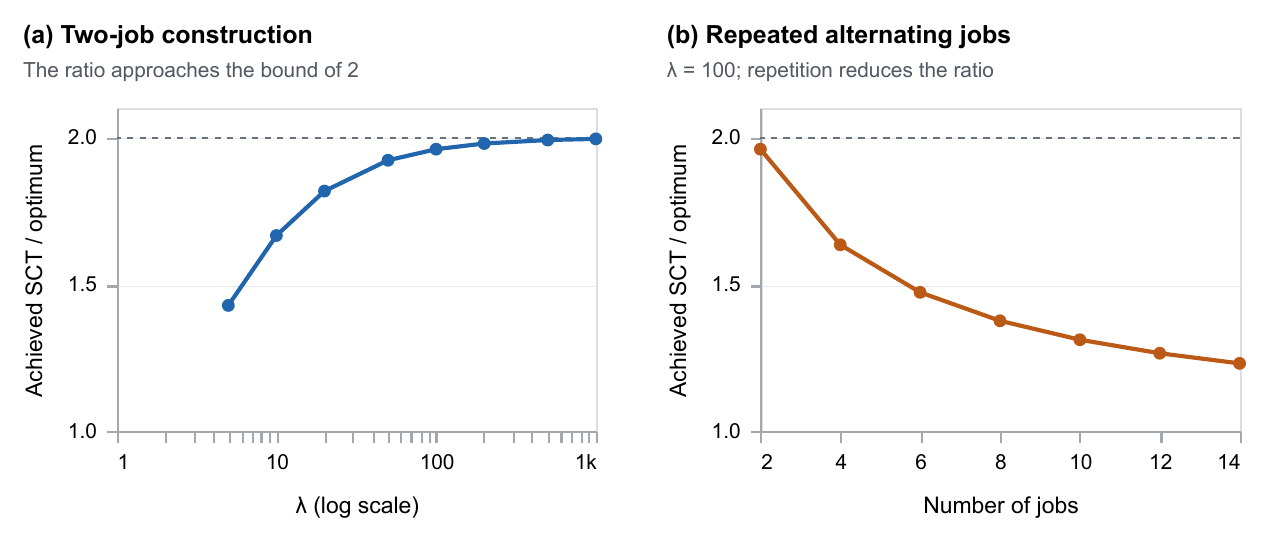}
  \caption{SCT approximation ratio of $\mathcal{A}_1$ on
  (a)~the two-job construction
  $J_1 = (1, \lambda)$, $J_2 = (\lambda - 1, 1)$, $\Bz = \BM = \lambda$
  versus $\lambda$, and (b)~its repeated alternating structure at
  $\lambda=100$ versus job count. For two jobs, the ratio
  $2\lambda/(\lambda+2)$ approaches 2, establishing tightness of the
  guarantee (dashed line). Repeating the structure reduced the observed ratios toward 1 and did not
  preserve the two-job worst case over the tested sizes.}
  \label{fig:tightness}
\end{figure}

\begin{figure}[tbp]
  \centering
  \includegraphics[width=0.95\linewidth]{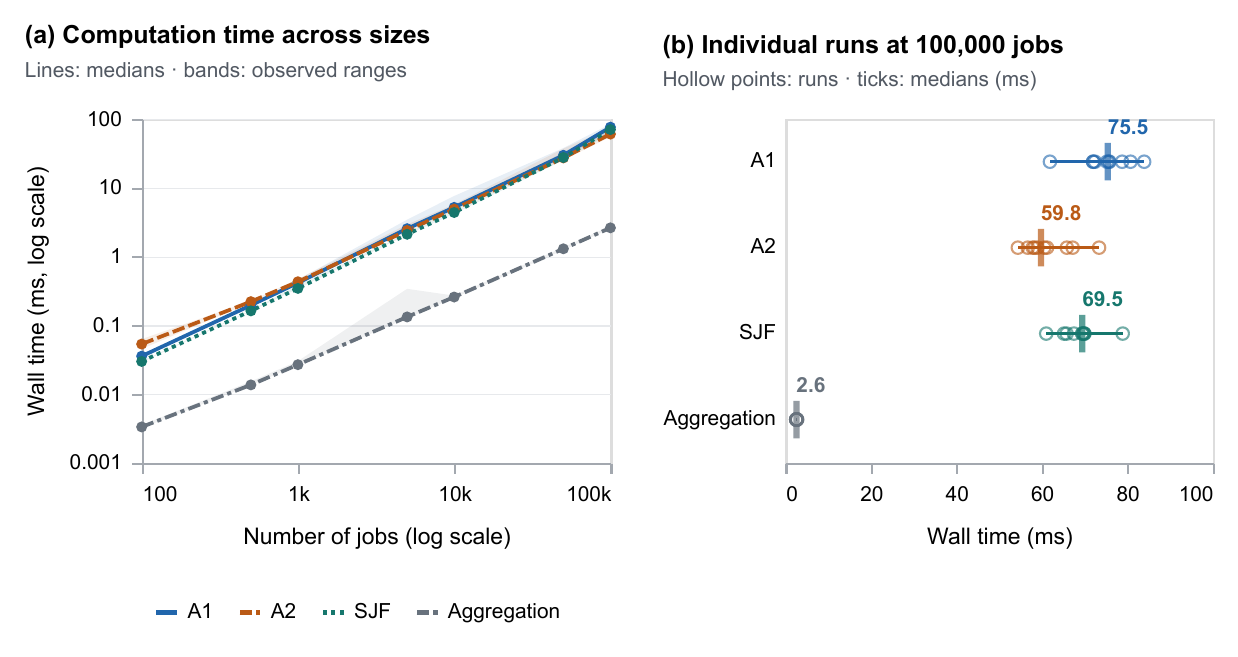}
  \caption{Desktop wall-clock computation time for $\mathcal{A}_1$,
  $\mathcal{A}_2$, SJF scheduling, and partial-makespan evaluation by summing
  job parameters. Timing included schedule construction and excluded input
  generation and disposal of inputs and results. (a)~Median time versus job count, with bands spanning the ten
  seeded runs. (b)~Individual runs at 100,000 jobs, with median ticks and
  values in milliseconds. At this size, median times were 75.5\,ms for
  $\mathcal{A}_1$ and 59.8\,ms for $\mathcal{A}_2$ on the tested desktop.}
  \label{fig:scalability}
\end{figure}

\begin{table}[tbp]
  \small
\centering
\caption{SCT of $\mathcal{A}_1$ divided by the exact partial-recharging optimum on 930 alternating low-execution/high-energy and high-execution/low-energy job sets. The trial-weighted mean ratio was 1.0186 and the maximum was 1.1814; the last column gives the number of trials at each size.}
\label{tab:sct-approx}
\begin{tabular}{rcccc}
\toprule
$n$ & Avg.\ ratio & Max.\ ratio & \% Optimal & Trials \\
\midrule
  6 & 1.0162 & 1.1814 & 39.5\% & 200 \\
  8 & 1.0242 & 1.0766 & 13.0\% & 200 \\
  10 & 1.0235 & 1.1059 & 34.5\% & 200 \\
  12 & 1.0114 & 1.0511 & 5.5\% & 200 \\
  14 & 1.0201 & 1.0691 & 20.0\% & 100 \\
  16 & 1.0066 & 1.0572 & 16.7\% & 30 \\
\bottomrule
\end{tabular}
\end{table}
\begin{table}[tbp]
  \small
\centering
\caption{Makespan of $\mathcal{A}_2$ divided by the exact complete-recharging optimum on 1,300 randomly generated job sets. The trial-weighted mean ratio was 1.0119 and the maximum was 1.0635, corresponding to at most 6.35\% excess makespan on these inputs.}
\label{tab:makespan-approx}
\begin{tabular}{rcccc}
\toprule
$n$ & Avg.\ ratio & Max.\ ratio & \% Optimal & Trials \\
\midrule
  6 & 1.0046 & 1.0402 & 61.5\% & 200 \\
  8 & 1.0070 & 1.0433 & 24.0\% & 200 \\
  10 & 1.0055 & 1.0317 & 19.5\% & 200 \\
  12 & 1.0093 & 1.0617 & 10.0\% & 200 \\
  14 & 1.0159 & 1.0625 & 6.5\% & 200 \\
  16 & 1.0229 & 1.0635 & 3.0\% & 100 \\
  18 & 1.0238 & 1.0561 & 1.0\% & 100 \\
  20 & 1.0236 & 1.0543 & 1.0\% & 100 \\
\bottomrule
\end{tabular}
\end{table}

\subsection{Running-Time Scalability}
\label{sec:exp-scalability}

We measured running time under Python 3.14.7 on an Apple M4 Pro host with
48\,GiB of memory running macOS 27.
At each of seven sizes from 100 to 100,000 jobs, we used one untimed
warmup followed by ten separately seeded runs. We included schedule construction in the timing and excluded instance
generation, schedule validation, and disposal of inputs and results. We
released each result after stopping its timer and before starting the
next call, and released each workload after all its calls. We left Python's
cyclic garbage collector enabled, shuffled the order of the algorithms, and
did not isolate the runs from other host activity.

At 100,000 jobs, median wall times were 75.5\,ms for \algAone{},
59.8\,ms for \algAtwo{}, and 69.5\,ms for SJF
(\cref{tab:timing-samples,fig:scalability}). The corresponding sample
standard deviations were 6.0, 5.8, and 4.6\,ms. The partial-makespan
measurement had a median of 2.6\,ms and included summing the input arrays, an $O(n)$ operation;
only evaluation from two already-computed sums is $O(1)$.
The three scheduling methods had similar running times, which is
consistent with sorting dominating each of them. These measurements
describe one desktop computer; they do not establish execution times on
embedded hardware.

\begin{table}[tbp]
  \centering
\small
\caption{Desktop computation time in milliseconds for scheduling with $\mathcal{A}_1$, $\mathcal{A}_2$, and SJF, and for partial-makespan evaluation by summing job parameters. Entries are medians (sample standard deviations) of ten seeded runs, including schedule construction and excluding input generation and disposal of inputs and results. All three scheduling methods had median times below 80\,ms at 100,000 jobs on this host; sum aggregation took 2.6\,ms.}
\label{tab:timing-samples}
\begin{tabular}{rrrrr}
\toprule
$n$ & $\mathcal{A}_1$ & $\mathcal{A}_2$ & SJF & Aggregation \\
\midrule
100 & 0.035 (0.001) & 0.053 (0.005) & 0.030 (0.000) & 0.003 (0.000) \\
500 & 0.194 (0.002) & 0.219 (0.004) & 0.161 (0.001) & 0.013 (0.001) \\
1,000 & 0.410 (0.016) & 0.426 (0.020) & 0.341 (0.018) & 0.026 (0.001) \\
5,000 & 2.518 (0.337) & 2.369 (0.221) & 2.087 (0.035) & 0.131 (0.065) \\
10,000 & 5.141 (0.789) & 4.851 (0.415) & 4.329 (0.250) & 0.255 (0.006) \\
50,000 & 29.411 (3.947) & 26.949 (4.657) & 27.290 (2.730) & 1.284 (0.035) \\
100,000 & 75.480 (6.012) & 59.821 (5.754) & 69.518 (4.610) & 2.591 (0.080) \\
\bottomrule
\end{tabular}

\end{table}

\subsection{Workloads and Capacity/Charge-Rate Sensitivity}
\label{sec:exp-casestudy}

We needed to distinguish the extra time imposed by the complete-recharging
policy from the error introduced by an approximation algorithm. To make
this distinction, we compared the optimal makespans under the two
policies, and we separately compared each algorithm with the optimum for
its own objective. We repeated these comparisons while varying the
capacity and the charging rate. We used the same design for 30 synthetic
job sets and for six job sets constructed from measured sensor-current
traces (\cref{sec:exp-traces}).

\begin{figure}[tbp]
  \centering
  \includegraphics[width=0.9\linewidth]{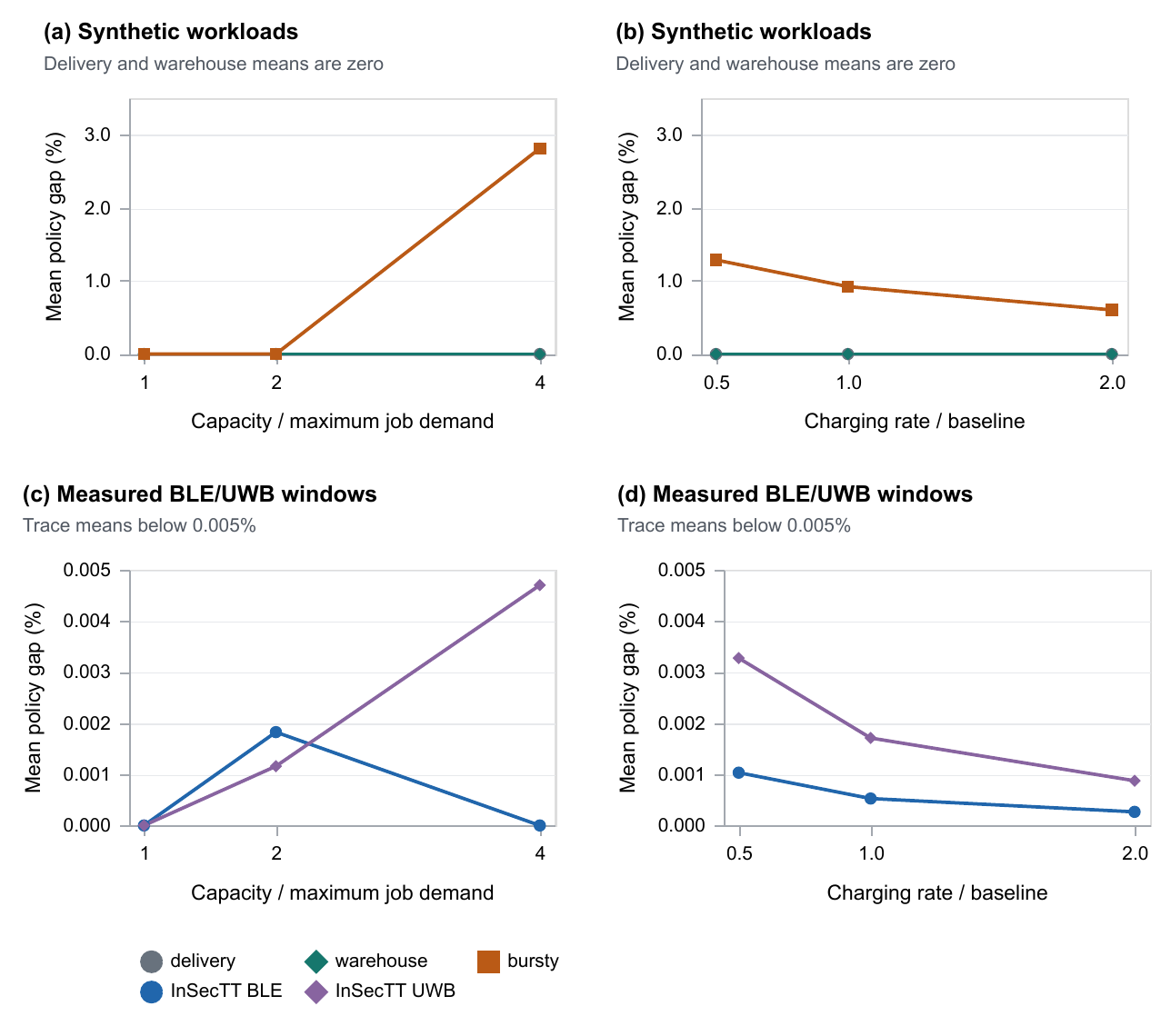}
  \caption{Extra optimal makespan caused by complete recharging, expressed
  as a percentage of the optimum under the partial-recharging policy, versus capacity (left)
  and charging rate (right). Points average over job sets, initial charges,
  and the other swept parameter within each family. The synthetic bursty
  jobs had the largest mean penalty; the BLE/UWB activity windows had
  mean penalties below 0.005\% at the tested settings. Faster charging
  reduced the mean time penalty over the tested settings; larger capacity
  did not always reduce it.
  Synthetic (top) and measured-window (bottom) panels use different vertical
  scales. We omitted capacity factor 0.75 because it made every set infeasible.}
  \label{fig:sensitivity}
\end{figure}
\begin{table}[tbp]
  \centering
\small
\caption{Policy cost and approximation quality across the 972 feasible capacity/charge-rate configurations. Policy gap is the percentage by which the optimal makespan under the complete-recharging policy exceeds that under the partial-recharging policy. We measured A2 excess against the optimal makespan under the complete-recharging policy and A1 ratios against the optimal SCT under the partial-recharging policy. The bursty sets had the largest policy gap (18.18\%); on the measured BLE/UWB windows, A1 attained the optimum and the policy gap stayed below 0.022\%.}
\label{tab:sensitivity}
\begin{tabular}{lrrrrrr}
\toprule
 & & \multicolumn{2}{c}{Policy gap (\%)} & A2 excess (\%) & \multicolumn{2}{c}{A1 ratio} \\
Workload & Sets & Mean & Max & Max & Mean & Max \\
\midrule
delivery & 10 & 0.00 & 0.00 & 1.58 & 1.0005 & 1.0092 \\
warehouse & 10 & 0.00 & 0.00 & 0.00 & 1.0002 & 1.0098 \\
bursty & 10 & 0.94 & 18.18 & 1.89 & 1.0166 & 1.2077 \\
InSecTT BLE & 3 & 0.0006 & 0.0062 & 0.1291 & 1.0000 & 1.0000 \\
InSecTT UWB & 3 & 0.0020 & 0.0215 & 0.1779 & 1.0000 & 1.0000 \\
\bottomrule
\end{tabular}

\end{table}

For each of three synthetic families, we generated ten sets with ten jobs
each, using seeds 0 through 9. Delivery jobs used $e=60D/V$ and charge demand
$Q=D(1.2+0.3W)$ with $D\sim U[4,20]$, $V\sim U[25,45]$, and
$W\sim U[0.5,5]$. Warehouse jobs used $e=D/(60V)+H$ and
$Q=0.01D+L$, with $D\sim U[20,200]$, $V\sim U[1,2]$,
$H\sim U[0.5,3]$, and $L\sim U[0.5,1.5]$. Bursty jobs alternated
between $e\sim U[1,5],Q\sim U[20,60]$ and
$e\sim U[10,30],Q\sim U[1,5]$. Draws were independent within these
rules. We rounded generated execution durations upward to integer scenario time units. Synthetic charge units are arbitrary and
$b_i=\lceil Q_i\rceil$;
these assumed rules have no physical calibration. We generated demands before choosing capacity and never clipped them to fit it.

For each job set, let $b_{\max}=\max_i b_i$. We tested
$\BM=\lfloor\kappa b_{\max}\rfloor$ for
$\kappa\in\{0.75,1,2,4\}$ and
$\Bz=\lfloor f\BM\rfloor$ for $f\in\{0,0.5,1\}$.
If any job exceeded capacity, then we rejected the configuration without
dropping jobs or reducing their demands. Charge rates were 0.5, 1, and 2 times a stated baseline,
which is one charge unit per scenario time unit for synthetic jobs.
For rate $r$, we passed $e'_i=2re_i$, $b'_i=2b_i$, and battery levels
scaled by two to the unit-rate algorithms, then divided time objectives by
$2r$. This common scaling made all inputs integral on the stated grid. It is
an exact change of units within the linear charging model.

For every feasible configuration, we obtained exact makespan and \SCT{}
optima under the partial-recharging policy, the exact makespan optimum under
the complete-recharging policy, and the schedules of \algAone{} and \algAtwo{}.
We independently replayed each schedule to check its charge bounds, recharge
policy, job coverage, and completion times. To quantify the difference
between the optimal makespans, we define the \emph{policy gap} as
$(C^*_{\max,\mathrm{complete}}/C^*_{\max,\mathrm{partial}}-1)100\%$;
to quantify the approximation error of $\mathcal{A}_2$, we define the
\emph{A2 excess} as
$(C_{\max,\mathcal A_2}/C^*_{\max,\mathrm{complete}}-1)100\%$.

To interpret the policy gap, we must also specify when the batch ends.
Our makespan objective stops at the completion of the last job and assigns
no value to the charge that remains in the battery. Let $S=\sum_i b_i$ and $E=\sum_i e_i$
in the scenario's charge and time units, and let $r$ be its charging rate.
If $S>\Bz$, then minimum partial recharging ends with zero charge.
If an optimal complete-recharging schedule ends with charge
$B_{\mathrm{end}}^*$, then energy conservation gives
\[
 C^*_{\max,\mathrm{complete}}-C^*_{\max,\mathrm{partial}}
 =\frac{B_{\mathrm{end}}^*}{r}.
\]
The extra time therefore replenishes charge that remains available after
the batch, and the gap is zero exactly when an optimal complete-recharging
schedule ends with an empty battery. If we instead require the device to
finish by restoring its battery to $\BM$, then both policies attain total
duration $E+(S+\BM-\Bz)/r$. The policy gap below measures completion of
one batch without this terminal-charge requirement.

\paragraph{Results on synthetic workloads.}

Across the 30 synthetic job sets, 810 of the 1,080 configurations were
feasible. We rejected the 270 configurations with $\kappa=0.75$ because
at least one job exceeded capacity. We gave each feasible configuration
equal weight within its workload family (\cref{tab:sensitivity,fig:sensitivity}).
The shared summaries also include the 162 feasible measured-window
configurations, whose construction and results we discuss separately
(\cref{sec:exp-traces}). Thus the full experiment had 972 feasible cases
among 1,296 configurations.

At $\kappa=1$, a job with demand $b_{\max}=\BM$ can fill the final
group exactly. If recharging is necessary, then this choice leaves zero
terminal charge and makes the two policy optima equal; otherwise both
policies execute the batch without recharging. This property accounts for
the zero policy gap in all 324 feasible configurations at $\kappa=1$.
Those zeros follow from the capacity construction, independently of the
workload family.

On the synthetic bursty jobs, the optimal makespan under complete
recharging exceeded the optimum under the partial-recharging policy by up to 18.18\%. The two
policy optima coincided on the sampled delivery and warehouse jobs.
The largest observed ratio of $\mathcal{A}_1$ was 1.2077, and the largest
A2 excess was 1.89\% (\cref{tab:sensitivity}).

Faster charging reduced the mean policy gap over the tested rate settings,
while increasing capacity did not always reduce it (\cref{fig:sensitivity}).
For fixed demands, capacity, and initial charge, let $R_p$ and $R_c$ be
the optimal amounts of replenished charge under the partial- and
complete-recharging policies. The percentage gap is
\[
 100\,\frac{R_c-R_p}{rE+R_p},
\]
when the partial-recharging optimum is positive. The optimal recharge
amounts do not depend on $r$ in our linear model, so the gap cannot
increase as charging becomes faster. This explains the rate trend in
the experiment. Changes in capacity also change which subsets fit in
the final group, so the same monotonicity need not hold for capacity.
The bursty workloads had the largest mean gap. We averaged over job sets,
initial charges, and the other swept parameter within each family, keeping
synthetic and measured families separate.

\subsection{Experiments with Measured Sensor-Current Traces}
\label{sec:exp-traces}

We also wanted to evaluate the algorithms on charge demands derived from
measurements. For this purpose, we constructed jobs from public
sensor-current recordings. The recordings supply the demands; the job
boundaries and the battery and charger parameters are our choices.

\paragraph{Data recorded by InSecTT.}
Version 1.0 of the InSecTT WSN Power Consumption
Dataset~\cite{karoliny2023insectt} contains current measurements from sensor
nodes using Bluetooth Low Energy (BLE), Thread, EPhESOS, and ultra-wideband
(UWB) communication. Its motivating application collects
two bytes of sensor data every 100\,ms and varies the communication period
from 100 to 1600\,ms. The corresponding application-data batches are
2, 4, 8, 16, and 32 bytes; longer periods batch more readings per transmission.
We used the BLE and UWB recordings. The BLE nodes used an nRF52840
development kit; the UWB nodes used an nRF52832 with a DW1000 transceiver.

Each protocol file contains a timestamp column and five current columns,
one for each communication period (100, 200, 400, 800, or 1600\,ms).
Each column is a separate 60-second recording sampled every
10\,$\mu$s with a Nordic Power Profiler Kit II, giving six million current
samples per configuration. The dataset reports current in microamperes and includes
radio activity and sleep. The recordings have arbitrary starting phases and
exclude acknowledgments and retransmissions. They provide current histories;
we constructed job boundaries and battery capacities for the scheduling
experiment, and our algorithms chose the recharge decisions.

\paragraph{Converting recordings into jobs.}
To use the recordings as scheduling inputs, we need an execution duration
and a charge demand for each job. We divided each
current column into nonoverlapping windows whose length equals
that column's communication period, starting at timestamp zero. For a window
$W_i$ with measured current samples $I_t$ in microamperes, we set $e_i$
to the window length in seconds and computed
\[
  Q_i=10^{-5}\sum_{t\in W_i} I_t\quad\text{microcoulombs}.
\]
We then set $b_i=\lceil Q_i/q\rceil$ with charge quantum
$q=0.1\,\mu$C, retaining the raw integral and rounding excess.
For example, the first 400-ms BLE window contains 40,000 samples and
consumed approximately $13.505\,\mu$C. We represented it as a job with
$e_i=0.4$\,s and $b_i=136$ charge units, representing $13.6\,\mu$C.
The window includes sleep as well as communication. We used its full
duration because we treated the activity window as a job in this model;
we did not interpret it as a measured atomic-task execution time.
The source supplies no voltage, so the measured quantity remains charge.

We constructed three ten-job sets per protocol. The first set took windows
0 and 1 from each of the five current columns, the second took windows
2 and 3, and the third took windows 4 and 5. Thus each set mixed two
windows at each of the five communication periods and had total execution
duration 6.2\,s. We selected the 60 distinct windows by fixed indices
before evaluating schedules. The two protocol files contain 2,324 full
windows in total; only these 60 entered the scheduling experiment. We
excluded incomplete final windows.

\paragraph{Scheduling experiment.}
We treated each ten-job set as an offline batch on one device, with all jobs
available at time zero. The scheduler chooses their order and inserts
recharges between jobs. For each set we varied capacity and initial charge
using the factors in \cref{sec:exp-casestudy}, and assumed constant charging
currents of 500, 1000, or 2000\,$\mu$A. With the charge quantum above,
the corresponding model rates are $r=I/q$. These battery and charger
settings were experimental inputs. We ran $\mathcal{A}_1$ for
\SCT{} under the partial-recharging policy and $\mathcal{A}_2$ for makespan
under the complete-recharging policy, computed their exact optima, and compared the optimal makespans of the two recharging policies.
The trace portion therefore had $6\times4\times3\times3=216$
configurations: 162 feasible cases and 54 rejected because at least one
job exceeded capacity.

\paragraph{Results on measured windows.}

For the selected sensor windows, $\mathcal{A}_1$ attained the exact \SCT{}
optimum; the maximum policy gap was 0.0215\% and the maximum A2 excess was
0.1779\% (\cref{tab:sensitivity,fig:sensitivity}). The first of these
observations has a structural explanation.
We checked each pair of jobs in all six batches at each of the three
charging currents. All 18 batch/current combinations satisfied
$e_i<e_j\Rightarrow e_i+b_i/r<e_j+b_j/r$. SJF with increasing-energy
tie-breaking therefore attains the \SCT{} optimum
(\cref{thm:sct-partial-poly}), independently of the capacity and the
initial charge within the feasible grid. These measured inputs lie in a
polynomially solvable case, and they do not distinguish $\mathcal{A}_1$
from that exact rule.

The batches each contained 6.2\,s of execution, which also helps explain
the small makespan percentages. The largest policy penalty was
0.4\,ms for BLE and 1.4\,ms for UWB; the largest A2 excess was
8.2\,ms and 11.6\,ms, respectively. For the UWB configuration with the
largest A2 excess, the optimal complete-recharging makespan was
6.5208\,s and the makespan of $\mathcal{A}_2$ was 6.5324\,s. The 11.6\,ms difference
was 0.1779\% of makespan but 3.616\% of the optimal 0.3208\,s spent
recharging.
To examine the effect of charge quantization, we repeated the experiment
with the coarser quantum $q=1\,\mu$C. This rerun gave a maximum trace
policy gap of zero and A2 excess of 0.2147\%. Changing the quantum also
changed the induced capacity and initial-charge grids, so we cannot use
this comparison as an estimate of measurement uncertainty.

We used the measured duration--charge pairs as fixed inputs to our offline
model. We combined separate recordings of communication configurations and
allowed the scheduler to reorder their windows. This construction does not
reproduce the original periodic traffic or test its deadlines and protocol
dependencies. The three selections per protocol are disjoint
parts of the recordings, not independent hardware trials. The dataset also
contains no measured charger or battery state trajectory. The results
therefore support approximation-quality and policy-cost comparisons for
these constructed batches; deployment would additionally require validating
job independence and measuring battery and charger behavior.

\subsection{Comparison with Simple Scheduling Heuristics}
\label{sec:exp-baselines}

A worst-case guarantee does not establish that an algorithm improves on
simpler rules for a particular workload. We therefore compared
$\mathcal A_1$ and $\mathcal A_2$ with simple heuristics on the same 972
feasible configurations. For \SCT{} minimization under the partial-recharging policy, SJF orders jobs by
$e_i$, breaking ties by $b_i$; SSF orders them by $e_i+b_i$, breaking ties
by $e_i$. Both recharge only enough to execute the next job. The weights
use the rate-normalized time units of the model, so SSF compares execution
time plus charging time. Remaining ties use input order.

For makespan minimization under the complete-recharging policy, the greedy heuristic scans jobs in decreasing
$b_i$, breaking ties by input order, and places each job that fits into the
final group. It schedules the remaining jobs first in input order,
recharging when needed, then fills the battery before the final group.
If all jobs fit in the initial charge, then it executes them without recharging.
$\mathcal{A}_2$ and the greedy heuristic both follow the
complete-recharging policy, so they differ only in how they select the
final group.

\begin{table}[tbp]
  \centering
\small
\caption{Mean objective divided by its exact optimum for five methods on the same 972 feasible configurations. SCT methods use partial recharging; makespan methods use complete recharging. Each configuration has equal weight within its workload family.}
\label{tab:baselines}
\begin{tabular}{lrccccc}
\toprule
 & & \multicolumn{3}{c}{\shortstack{SCT under the\\partial-recharging policy}} & \multicolumn{2}{c}{\shortstack{Makespan under the\\complete-recharging policy}} \\
Workload & Cases & A1 & SJF & SSF & A2 & Greedy \\
\midrule
delivery & 270 & 1.00053 & 1.01376 & 1.00122 & 1.00172 & 1.00531 \\
warehouse & 270 & 1.00020 & 1.00201 & 1.00054 & 1.00000 & 1.00383 \\
bursty & 270 & 1.01662 & 1.37473 & 1.14524 & 1.00060 & 1.00047 \\
InSecTT BLE & 81 & 1.00000 & 1.00000 & 1.00000 & 1.00010 & 1.00029 \\
InSecTT UWB & 81 & 1.00000 & 1.00000 & 1.00000 & 1.00032 & 1.00014 \\
\bottomrule
\end{tabular}

\end{table}
\begin{table}[tbp]
  \centering
\small
\caption{Paired wins/ties/losses on identical configurations. A win means that A1 or A2 has a strictly smaller integer objective than the named heuristic. Counts describe the tested grid, not independent deployment trials.}
\label{tab:baseline-pairs}
\begin{tabular}{lccc}
\toprule
Workload & A1 vs. SJF & A1 vs. SSF & A2 vs. Greedy \\
\midrule
delivery & 230/36/4 & 51/219/0 & 99/171/0 \\
warehouse & 69/196/5 & 9/261/0 & 54/216/0 \\
bursty & 243/15/12 & 155/115/0 & 27/225/18 \\
InSecTT BLE & 0/81/0 & 0/81/0 & 27/54/0 \\
InSecTT UWB & 0/81/0 & 0/81/0 & 0/54/27 \\
\bottomrule
\end{tabular}

\end{table}

The largest separation between the \SCT{} methods occurred on the
synthetic bursty workloads (\cref{tab:baselines}). The mean ratio of
$\mathcal{A}_1$ was 1.01662, compared with 1.37473 for SJF and 1.14524 for
SSF. The bursty jobs alternate between short jobs with large demands and
long jobs with small demands, which is the situation in which neither
execution time alone nor $e_i+b_i$ alone determines a good order.
$\mathcal{A}_1$ improved on SSF in 215 configurations and tied with it in
the remaining 757. Against SJF, $\mathcal{A}_1$ won 542 times, tied 409
times, and lost 21 times (\cref{tab:baseline-pairs}). All three methods
attained the \SCT{} optimum on the selected BLE and UWB windows, as the
ordering condition above predicts, so those batches show no advantage for
$\mathcal{A}_1$.

$\mathcal{A}_2$ improved on greedy final-group selection in 207
configurations, tied with it in 720, and lost in 45. Its mean makespan
ratio was smaller for delivery, warehouse, and BLE, while the greedy
heuristic had a slightly smaller mean for bursty and UWB. Both methods
were close to the optimum on these workloads. The comparison therefore
shows a benefit for $\mathcal{A}_1$ on the bursty jobs, but it does not
show that $\mathcal{A}_2$ consistently improves on greedy selection.
$\mathcal{A}_2$ has a proved worst-case bound; we have not established
such a bound for the greedy heuristic. The paired counts describe one
fixed parameter grid, and repeated settings of the same job sets are not
independent trials.

In these experiments, both algorithms were much closer to the optimum
than their worst-case factors: the largest observed ratios were 1.2077 for
$\mathcal{A}_1$ and 1.0635 for $\mathcal{A}_2$. The complete-recharging
policy itself cost up to 18.18\% in optimal makespan on the synthetic
bursty jobs and at most 0.0215\% on the other families. None of these
experiments includes setup time, which we now add to the model.

\section{Partial Recharging with Setup Time}
\label{sec:setup-partial}

Having classified the variants without setup time, we now let every
positive recharge of $q$ units take $h+q$ time units. The number of
charging sessions now affects every completion time. In the two-trip
example of
\cref{sec:intro} ($e_i=b_i=1$, $\BM=2$, $\Bz=0$, $h=3$), recharging one
unit before each job gives completion times $(5,10)$ and $\SCT=15$, whereas
a single recharge of two units gives $(6,7)$ and $\SCT=13$. The
minimum-recharging rule (\cref{lem:min-recharge}) therefore no longer
applies: even for a fixed job order, we must optimize the recharge amounts
together with their positions.

We first optimize the recharges for a fixed job order. We then use equal
energy demands to obtain polynomial algorithms that also choose the order.
This section treats partial recharging, and \cref{sec:setup-complete}
treats complete recharging.

\subsection{Recharging for a Fixed Order}

We first remove recharging decisions that cannot improve any of the four
objectives. Fix an order and index its jobs by position. We write
$P_j=\sum_{\ell=1}^j b_\ell$ and $E_j=\sum_{\ell=1}^j e_\ell$ for the
prefix sums of energy demand and execution time, with $P_0=E_0=0$. Let $k$
be the largest index satisfying $P_k\leq\Bz$. The first $k$ jobs form the
\emph{free prefix}: they can execute using the initial energy alone.

The following lemma shows that we may restrict attention to schedules in
which each session supplies exactly the jobs up to the next session.

\begin{lemma}[Partial-recharging normal form]
\label{lem:setup-partial-normal}
For any fixed order, an optimal schedule for each objective executes the
free prefix without recharging. If $k<n$, then it first recharges before job
$k+1$, has zero battery before each later positive recharge, and finishes
with zero battery.
\end{lemma}

\begin{proof}
We transform an arbitrary feasible schedule without increasing any
completion time. First, we postpone all recharges performed before job
$k+1$ to the boundary immediately before that job, and combine them. The
free prefix remains feasible. The battery after the combined recharge
equals its previous level at that boundary, so the capacity constraint
still holds. No completion time increases, because the transformation
postpones charging work and removes any redundant setups.

Next, we consider two consecutive positive recharges, and let $r$ be the
residual battery immediately before the second. We reduce the first
recharge by $r$ and increase the second by $r$. In the first charged
group, the first job requires more energy than the residual initial battery
$a=\Bz-P_k$; hence $r$ is less than the first recharge amount.
For each later group, the previously normalized starting battery is zero,
so the same reduction is feasible. Battery levels decrease monotonically
between recharges, and the original final residual was $r$.
Consequently, subtracting $r$ preserves every intermediate energy constraint.

The second recharge starts $r$ time units earlier and acquires $r$
additional units. Its completion and final battery are unchanged unless a
setup disappears, in which case subsequent jobs finish earlier.
We apply this transformation from left to right, and then remove unused
energy from the last recharge.
\end{proof}

In this normal form, each session supplies one consecutive group of jobs.
The only residual that a group may inherit is $a$, after the free prefix.
All recharge amounts are integers, because the demands and the initial
energy are integers. The lemma therefore reduces the choice of recharge
amounts to the choice of group boundaries, and we choose these boundaries
differently for each objective.

\paragraph{Sum of completion times.}
A recharge before position $i+1$ delays all $n-i$ remaining completions.
Let $D[j]$ be the least accumulated weighted charging time at a group
boundary after job $j$. We initialize $D[k]=0$ and all other values to
infinity. For $i\geq k$, we define $a_i=a$ when $i=k$, and $a_i=0$
otherwise. A group $i+1,\ldots,j$ is feasible when $P_j-P_i\leq\BM$;
its required recharge is $q_{ij}=\max(0,P_j-P_i-a_i)$. We use the
recurrence
\begin{equation}
 D[j]=\min_{\substack{k\leq i<j\\P_j-P_i\leq\BM}}
       \{D[i]+(n-i)L(q_{ij})\}.
 \label{eq:setup-partial-sct-dp}
\end{equation}
Adding the execution contribution
$\sum_{j=1}^n(n-j+1)e_j$ gives the optimum.
The normal form supplies a path through this recurrence for some optimal
schedule, and every path reconstructs a feasible schedule.
Prefix sums make each transition constant work, which gives $O(n^2)$ time
and $O(n)$ storage for values and predecessors.

\paragraph{Makespan.}
Makespan permits a faster construction, because every unit of charging
time has the same weight. After the free prefix, we form the longest
possible consecutive group with demand at most $\BM$, recharge exactly its
deficit, and repeat. This construction acquires exactly $S-\Bz$ energy when
recharging is necessary. Its first group ends at least as late as the
first group in any other feasible partition, and repeating this comparison
at successive boundaries shows that greedy grouping uses the fewest
sessions. Fixed-order makespan therefore takes $O(n)$ time.

\paragraph{Deadline objectives.}
For the deadline objectives, two prefixes with the same boundary may have
different elapsed times, because they used different numbers of setups.
We therefore retain the session count $K$ in the state.
For a charged group that ends at position $j$, each member $\ell$
completes at
\begin{equation}
 C_\ell=E_\ell+P_j-\Bz+hK.
 \label{eq:setup-partial-group-completion}
\end{equation}
The cumulative acquired energy at that group is $P_j-\Bz$, which explains
why earlier recharging decisions do not appear in this expression.

A state $(i,K)$ stores the least objective value through boundary $i$.
For $T_{\max}$, a transition over $i+1,\ldots,j$ takes the maximum of
the predecessor value, zero, and the group maximum of
$E_\ell-d_\ell+P_j-\Bz+hK$.
Maintaining a running maximum while extending a group gives $O(n^3)$
time over all boundaries and session counts.
For $U$, we fix a destination $(j,K)$ and precompute
\[
 V_{j,K}(r)=\sum_{\ell=1}^{r}
   \mathbf{1}\{E_\ell+P_j-\Bz+hK>d_\ell\},
 \qquad 0\leq r\leq j .
\]
The group contributes $V_{j,K}(j)-V_{j,K}(i)$.
Its predecessor uses $K-1$ sessions if $q_{ij}>0$, and $K$ otherwise.
Processing the destinations by increasing $j$ permits us to evaluate all
starting boundaries in $O(n)$ additional work per destination.
Both deadline objectives therefore take $O(n^3)$ time with $O(n^2)$
dynamic-programming values.

These fixed-order algorithms accept arbitrary energy demands. They do not
choose the order, which is the hard part of the problem for unequal demands
(\cref{sec:hardness}).

\subsection{Equal Energy Demands}

When $b_i=b$, a fixed recharging pattern is feasible independently of which
job occupies each position. We can therefore exchange the identities of
jobs without changing the battery trajectory or the recharge durations,
and we use such exchanges to fix the job order for two of the objectives.

For $\SCT$, the execution contribution is
$\sum_j(n-j+1)e_{\pi(j)}$, which SJF order minimizes.
For $T_{\max}$, consider adjacent jobs $u,v$ with $d_u>d_v$.
We exchange their identities while keeping the recharges at the same
positions. The second completion time is unchanged. Job $v$ now finishes
no later than its original completion, and the new tardiness of $u$ is no
greater than the original tardiness of $v$. Repeated exchanges give EDF
order. These arguments apply to both recharging policies.

\begin{lemma}[Equal-energy orders]
\label{lem:setup-equal-orders}
Under either recharging policy, some optimal equal-energy schedule uses
SJF for $\SCT$, and some optimal equal-energy schedule uses EDF for
$T_{\max}$.
\end{lemma}

The lemma separates sequencing from recharging for these two objectives.
For $T_{\max}$ under partial recharging, EDF order followed by the
fixed-order dynamic program gives an $O(n^3)$ algorithm.
For $\SCT$, equal demands also permit a faster optimization of the
recharges, which we describe next.

\subsubsection{Sum of Completion Times}

If $b=0$ or $nb\leq\Bz$, then SJF needs no recharging.
Otherwise, let $k=\lfloor\Bz/b\rfloor$, $N=n-k$,
$a=\Bz-kb$, and $m=\lfloor\BM/b\rfloor$.
The free prefix is fixed, and the remaining decision is to divide the
$N$ jobs into charged groups of sizes $g_1,\ldots,g_K$, each at most $m$.
The first recharge is $bg_1-a$ and the later recharges are $bg_r$.
Their weighted contribution is
\begin{equation}
 \frac b2\left(N^2+\sum_{r=1}^{K}g_r^2\right)
       +h\sum_{r=1}^{K}r g_r-Na .
 \label{eq:setup-partial-equal-convex}
\end{equation}
The quadratic term counts the charging energy that precedes the jobs of
each group; the linear term counts the number of setups for which each job
waits.

We minimize the variable part of
(\ref{eq:setup-partial-equal-convex}) without enumerating all partitions.
Equivalently, we minimize
$\sum_r\bigl(bg_r(g_r+1)/2+hrg_r\bigr)$ subject to
$\sum_r g_r=N$ and $0\leq g_r\leq m$, using $N$ potential groups.
Increasing group $r$ from size $t-1$ to $t$ has marginal cost
\begin{equation}
 \Delta_{r,t}=bt+hr.
 \label{eq:setup-equal-marginal}
\end{equation}
The marginal costs of each group increase with $t$, so a greedy selection
of the cheapest marginal costs is valid. A minimum heap initially contains
$\Delta_{r,1}$ for all $1\leq r\leq N$. We repeatedly select the least
available marginal cost, increase the size of that group, and insert its
next marginal cost if its capacity permits. After $N$ selections, we omit
the empty groups.

\begin{theorem}[Equal-energy partial \SCT{}]
\label{thm:setup-partial-equal-sct}
SJF followed by the marginal-cost allocation minimizes $\SCT$ under
partial recharging in $O(n\log n)$ time and $O(n)$ additional space.
\end{theorem}

\begin{proof}
\Cref{lem:setup-equal-orders} fixes an optimal job order.
Selecting the $N$ cheapest marginal costs minimizes the separable
expression, because the earlier marginal costs of a group never exceed its
later ones. We resolve ties by the smaller group index.
Since $h\geq0$, the occupied groups then form a prefix and their sizes are
nonincreasing, so the selected allocation is a valid group sequence.
Its charging cost is optimal by (\ref{eq:setup-partial-equal-convex}).
Sorting and $N$ heap operations each take $O(n\log n)$
time, and the heap and the reconstructed schedule use $O(n)$ space.
\end{proof}

For the two-trip example above, the first marginal costs of the two
potential groups are $4$ and $7$. Increasing the first group to size two
costs $5$, so the heap places both jobs in one group. If $h=0$, then the
first marginal cost of every group is $b$, and the algorithm produces
one-job groups, in agreement with \cref{lem:min-recharge}.

\subsubsection{Makespan}

When $b>0$ and $nb>\Bz$, the maximum size of the free prefix is
$\lfloor\Bz/b\rfloor$, and each charged group contains at most $m$ jobs.
Consequently, every schedule needs at least
\begin{equation}
 K_{\min}=
 \left\lceil\frac{n-\lfloor\Bz/b\rfloor}{m}\right\rceil
 \label{eq:setup-equal-min-sessions}
\end{equation}
sessions. The fixed-order greedy construction attains this count
and acquires exactly $nb-\Bz$ energy.
Any job order therefore yields an optimal equal-energy makespan schedule
in $O(n)$ time. This statement includes the case in which no recharge is
needed, $b=0$ or $nb\leq\Bz$.

\subsubsection{Number of Tardy Jobs}

Minimizing $U$ requires us to select which jobs finish on time. We first
show that the timely jobs can be moved to the front. Delete all tardy jobs
from a feasible schedule, and reduce the recharges as needed.
The retained jobs finish no later: deleting execution frees time, and
deleting energy demand cannot require additional charging.
Under complete recharging, each retained recharge can still fill the
battery; its amount is no larger than before.
We then append the deleted jobs after the retained ones.
The EDF exchange argument orders the timely prefix without making
any of its jobs late.

We therefore scan the jobs in EDF order and retain the earliest completion
time $A(t,\beta)$ for each count $t$ of selected timely jobs and each
battery level $\beta$. Skipping the next job preserves the state.
Selecting job $i$ requires a feasible recharge $q$ and updates
\begin{equation}
 (t,\beta)\longrightarrow(t+1,\beta+q-b),
 \qquad A(t,\beta)+L(q)+e_i\leq d_i .
 \label{eq:setup-timely-transition}
\end{equation}
Among histories with the same $(t,\beta)$, an earlier completion permits
every continuation of a later one. This dominance rule justifies storing
one time per state.

For partial recharging, \cref{lem:setup-partial-normal} applies to
the selected subset. We need only transitions without a recharge and
positive recharges to target levels $\ell b$, where
$1\leq\ell\leq\min(n-t,\lfloor\BM/b\rfloor)$.
Reachable batteries have the forms $\Bz-rb$ or $rb$, with at most $n$
executed or reserved jobs, so only $O(n)$ battery values occur.
A direct enumeration of every source and target would take quartic time.
For target $z=\ell b$, the completion time after a recharge is instead
\[
 e_i+h+z+\min_{\beta<z}\{A(t,\beta)-\beta\}.
\]
For each $t$, we sweep the sorted batteries and the increasing targets
while retaining the prefix minimum in this expression. This sweep
evaluates all transitions with a recharge in $O(n)$ time, in addition to
the transitions without one.

\begin{theorem}[Equal-energy partial tardy count]
\label{thm:setup-partial-equal-tardy}
The EDF subset dynamic program with prefix minima minimizes $U$ under
partial recharging in $O(n^3)$ time.
\end{theorem}

The maximum feasible count $t$ determines the optimal value $U=n-t$.
Predecessor choices recover the timely prefix; any feasible execution
of the remaining jobs completes the schedule.
There are $O(n^2)$ values per input layer. Retaining all reconstruction
layers uses $O(n^3)$ storage, whereas computing the optimal value alone
requires $O(n^2)$ storage.
For $b=0$, only transitions without a recharge remain.

Equal energy demands therefore admit exact polynomial algorithms for all
four objectives under partial recharging with setup time. These algorithms
rely on equal demands; their exchange arguments do not extend to the
unequal-demand classes (\cref{sec:hardness}).

\section{Complete Recharging with Setup Time}
\label{sec:setup-complete}

Having treated partial recharging, we now consider complete recharging
with setup time. This policy fixes each positive recharge amount at the
battery deficit $\BM-\beta$, where $\beta$ is the current energy, so the
scheduler decides only when to incur this amount and its setup.
Unlike partial recharging, an optimal schedule can carry residual energy
into its next recharge or finish with unused energy, and the normal form
of \cref{lem:setup-partial-normal} is not available.
As in the previous section, we first optimize the recharges for a given
order, and then use exchanges under equal energy demands to obtain
polynomial algorithms that also choose the order.

\subsection{Recharging for a Fixed Order}

For a fixed order, let $P_i$ and $E_i$ denote the prefix sums of energy
demand and execution time, as in \cref{sec:setup-partial}.
Immediately after position $i$, the battery is either $\Bz-P_i$, if no
recharge has occurred, or $\BM-(P_i-P_r)$, if the last recharge was before
position $r+1$.
Only $O(i)$ battery levels are therefore reachable, even when $\BM$ is
large. This observation replaces the normal form: it keeps the state space
of each dynamic program polynomial in $n$.

\paragraph{Sum of completion times.}
We store the least weighted elapsed-time contribution at each
battery level after each position.
An action before job $i+1$ chooses $q\in\{0,\BM-\beta\}$, subject to
$\beta+q\geq b_{i+1}$.
Its duration is $L(q)+e_{i+1}$, which delays $n-i$ unfinished jobs.
The action therefore contributes
$(n-i)(L(q)+e_{i+1})$.
There are two actions per battery state, which gives $O(n^2)$ time.
This weighted-duration representation already includes every downstream
effect of elapsed time, so it requires no separate session count in the
state.

\paragraph{Deadline objectives.}
For the deadline objectives, the session count affects the time at which
a state is reached. We therefore retain $(\beta,K)$ after position $i$.
Energy conservation determines the elapsed time of this state:
\begin{equation}
 t_i=E_i+P_i-\Bz+\beta+hK .
 \label{eq:setup-complete-state-time}
\end{equation}
A feasible action determines the next completion time and updates $K$ only
when $q>0$. For $T_{\max}$, we retain the smaller accumulated maximum
tardiness at each state; for $U$, we retain the smaller accumulated count.
States with the same $(i,\beta,K)$ have the same elapsed time and the same
continuations, so these dominance rules are valid.
The $O(n)$ battery levels and $O(n)$ session counts at each of $n$
positions give $O(n^3)$ time for either deadline objective.

\paragraph{Makespan.}
Makespan admits a linear construction that avoids dynamic programming.
If the total energy demand is at most $\Bz$, then no recharge is needed.
Otherwise, we choose the earliest position $s+1$ such that the final suffix
$s+1,\ldots,n$ has energy demand at most $\BM$.
We execute the head $1,\ldots,s$ with a full recharge only when the next
job would not fit, and then recharge to full immediately before the suffix.
The suffix needs no further recharge. The following lemma shows that this
construction minimizes both the acquired energy and the number of
sessions.

\begin{lemma}[Fixed-order complete makespan]
\label{lem:setup-complete-fixed-makespan}
The earliest-feasible-suffix construction minimizes complete-recharging
makespan for its fixed job order in $O(n)$ time.
\end{lemma}

\begin{proof}
For any schedule with a positive recharge, let $r$ be the final recharge
boundary. Its total acquired energy is
\[
 Q=\BM-\Bz+P_r.
\]
Indeed, the battery after the last job is
$\BM-(P_n-P_r)$, and conservation gives the displayed identity.
Feasibility requires $r\geq s$, so the proposed construction minimizes $Q$.

Lazy full recharging minimizes the number of sessions for every fixed
prefix: its first group ends as late as possible, and the same comparison
applies after each successive recharge.
The minimum number of sessions needed to execute a prefix is nondecreasing
with the prefix length. Consequently, the minimum for the head ending at
$s$, plus its final recharge, is no larger than the count of a schedule
whose last recharge follows a longer head.
The proposed final recharge is positive when recharging is necessary.
The construction therefore minimizes both $Q$ and $K$, and hence
$E+Q+hK$.
\end{proof}

This result concerns a specified order. Finding an order that minimizes
both terms remains strongly NP-hard for unequal demands
(\cref{sec:hardness}). The lemma nevertheless provides exact recharging
for any ordering heuristic, and we use it in the approximation algorithms
(\cref{sec:selection}).

\subsection{Equal Energy Demands}

For $b_i=b$, \cref{lem:setup-equal-orders} gives SJF order for $\SCT$
and EDF order for $T_{\max}$.
EDF order followed by the preceding deadline dynamic program takes $O(n^3)$
time. The other objectives permit more direct constructions, which we
describe in turn.

\subsubsection{Sum of Completion Times}

Assume that $b>0$ and $nb>\Bz$, since otherwise SJF needs no recharge.
We write $m=\lfloor\BM/b\rfloor$.
Fix an initial prefix of length $p$ that executes without a recharge, where
$0\leq p\leq\lfloor\Bz/b\rfloor$.
Unlike partial recharging, complete recharging may benefit from a recharge
before the device has consumed all of its initial energy, so $p$ is a
decision. The first recharge amount is $x=\BM-\Bz+pb$.
After that recharge, every group can contain at most $m$ jobs.

We show that we can fill each nonfinal charged group to $m$ jobs.
Consider adjacent charged groups of positive sizes
$u,v$, and let $t$ jobs follow both groups.
Their first recharge is unchanged when the internal boundary moves.
The next recharge and the recharge after the pair, if present, contribute
\[
 (t+v)(h+bu)+t(h+bv).
\]
For fixed $w=u+v$, the variable part is
$h(w-u)+bu(w-u)$, a concave function of $u$.
If $w>m$, then its feasible endpoints are $w-m$ and $m$.
Their energy terms agree, and the setup term favors the larger first
group. Replacing the pair by $(m,w-m)$ therefore does not increase the
cost. If $w\leq m$, then merging the pair removes
$h(t+v)+buv\geq0$ from the charging contribution.
Repeated applications fill every nonfinal charged group, which proves the
following lemma.

\begin{lemma}[Full interior groups]
\label{lem:setup-complete-full-groups}
For each initial prefix size $p$, an optimal equal-energy complete
schedule in SJF order has $m$ jobs in each nonfinal charged group.
\end{lemma}

The lemma leaves only $p$ to optimize.
Let $N=n-p$ and $\ell=\lfloor(N-1)/m\rfloor$.
The weighted recharge contribution is
\begin{equation}
 \begin{split}
 R_h(p)={}&N L(\BM-\Bz+pb)\\
       &+(mb+h)\left(\ell N-\frac{m\ell(\ell+1)}2\right).
 \end{split}
 \label{eq:setup-complete-equal-sct}
\end{equation}
The first term accounts for the first fill, and the second sums the
$\ell$ later fills, each after $m$ more completions.
The convention $L(0)=0$ is necessary when $\Bz=\BM$ and $p=0$:
this candidate has no initial session.

We evaluate every feasible $p$ in constant time and add the common SJF
execution contribution. There are at most $n$ candidates.
Reconstruction takes linear time, so sorting dominates the
$O(n\log n)$ running time for $\SCT$ under complete recharging. At $h=0$,
expression (\ref{eq:setup-complete-equal-sct}) reduces to the recharge
contribution in the proof of \cref{thm:sct-complete-eq}.

\subsubsection{Makespan}

Equal energy demands make every job order equivalent for makespan.
The earliest feasible final suffix contains $\min(n,m)$ jobs.
We execute its head with lazy full recharges and perform the final fill;
\cref{lem:setup-complete-fixed-makespan} shows that the result is optimal
and takes $O(n)$ time.

The construction also attains the minimum session count
(\ref{eq:setup-equal-min-sessions}).
If the final group contains $m$ jobs, then the head requires the remaining
sessions, after accounting for the initial free jobs.
If $n<m$ and recharging is necessary, then one fill suffices.
The cases $b=0$ and $nb\leq\Bz$ require no recharge.

\subsubsection{Number of Tardy Jobs}

We use the EDF timely-subset argument and the states $A(t,\beta)$ from
\cref{sec:setup-partial}.
A selection transition now chooses either $q=0$ or $q=\BM-\beta$,
and applies (\ref{eq:setup-timely-transition}).
Positive recharges incur $h$; filling an already full battery costs
nothing. Skipping a job leaves the state unchanged.

Reachable batteries have the forms $\Bz-rb$ and $\BM-rb$ for at most
$n$ selected jobs.
There are therefore $O(n^2)$ count/battery states per input layer and at
most two actions per state.
Keeping the earliest feasible time at each state gives an $O(n^3)$
algorithm, because an earlier arrival at the same battery level permits
every continuation of a later arrival.
The largest selected count is the size of an optimal timely prefix.
We append all unselected jobs with feasible recharging to obtain $U=n-t$.

We now have exact polynomial algorithms for all four equal-energy
objectives under both policies. They remain valid at $h=0$, where the
algorithms of \cref{sec:complete,sec:partial} are at least as fast.
For complete recharging, residual energy prevents the normal form that we
used for partial recharging, but only $O(n)$ battery values are relevant
to each fixed-order or equal-energy subproblem.
The remaining difficulty is the sequencing of unequal demands. We show
next that this difficulty is inherent.

\section{Hardness with Setup Time}
\label{sec:hardness}

Having obtained polynomial algorithms for equal energy demands, we now
show that the other three job classes are strongly NP-hard under both
policies. The proof uses the observation that setup time makes the number
of charging sessions matter even when the acquired energy is fixed:
deciding whether a schedule can use the minimum number of sessions is a
packing problem. The construction also explains why the equal-energy
algorithms do not extend to equal execution times.

The source problem is 3-Partition~\cite{garey1979computers}: given $3m$
positive integers $a_i$ summing to $mB$, with $B/4<a_i<B/2$, decide whether
they form $m$ triples of sum $B$. We set $\BM=B$, $\Bz=0$, and $b_i=a_i$.
We choose $e_i=1$ for the equal-execution class and $e_i=b_i$ for the
class with energy equal to execution time. We write $n=3m$,
$E=\sum_i e_i$, and $S=mB$.

We first relate the number of sessions to the existence of a 3-partition.
Every group contains at most three jobs, because four demands exceed the
capacity. This remains true if a schedule carries unused energy into its
next recharge. Consequently, at least $m$ sessions are necessary. Exactly
$m$ sessions would give $m$ triples, each with energy demand at most $B$.
Their total is $mB$, so each must have energy demand exactly $B$.
Conversely, a 3-partition permits one full recharge per triple. A schedule
with $m$ sessions therefore exists, under either policy, if and only if
the 3-Partition instance has a solution.

\begin{theorem}[Complexity with setup time]
\label{thm:setup-hardness}
For either recharging policy and each of the equal-execution,
energy-equal-to-execution, and arbitrary job classes, all four objectives
are strongly NP-hard when $h$ is an input. Makespan, maximum tardiness,
and the number of tardy jobs are strongly NP-hard even at $h=1$.
\end{theorem}

\begin{proof}
For makespan, we set $h=1$ and $D=E+S+m$. A YES instance has a schedule of
makespan $D$. A NO instance needs at least $m+1$ sessions and at least $S$
units of acquired energy, so every schedule exceeds $D$. Giving every job
the deadline $D$ makes $T_{\max}=0$ equivalent to a makespan of at most
$D$; the same equivalence holds for $U=0$.

For \SCT{}, we need a further argument, because hardness of makespan does
not imply hardness of the sum. Before the $r$-th charging session, at most
$3(r-1)$ jobs have completed. The first $m$ sessions therefore contribute
at least $hW$ to \SCT{}, where
\[
 W=\sum_{r=1}^{m}\bigl(3m-3(r-1)\bigr)=\frac{3m(m+1)}2.
\]
Any additional session precedes an unfinished job and adds at least
one to this weighted session count. We set $V=n(E+S)$, $h=V+1$, and the
threshold $Z=hW+V$. A YES schedule with one full recharge per triple has
weighted setup contribution $hW$. The non-setup contribution of each
completion is at most $E+S$, so the \SCT{} of this schedule is at most $Z$.
Every NO schedule has \SCT{} at least $h(W+1)=Z+1$ from setup time alone.

All constructed quantities are polynomially bounded when the source
3-Partition numbers are polynomially bounded, so the reductions establish
strong NP-hardness. The arbitrary job class includes both restricted
constructions.
\end{proof}

The decision forms of these problems are in NP. An order and at most one
integer recharge amount before each job form a certificate of polynomial
size, because consecutive recharges can be merged and idle time can be
removed. The decision problems are therefore strongly NP-complete.

The theorem has two limits. The \SCT{} proof chooses $h$ from the input, so
it does not establish hardness for every prescribed positive constant $h$.
The makespan and deadline proofs do hold at $h=1$, and their thresholds of
zero tardiness also limit what an approximation guarantee for the deadline
objectives could be.

\begin{corollary}[Deadline objectives]
\label{cor:deadline-inapprox}
Unless P$=$NP, none of the unequal-demand deadline classes in
\cref{thm:setup-hardness} admits a polynomial-time approximation with a finite
multiplicative factor on all instances.
\end{corollary}
\begin{proof}
A finite multiplicative guarantee must return objective zero whenever the
optimum is zero. The reductions would then decide 3-Partition by testing
whether the returned schedule has zero tardiness or zero tardy jobs.
\end{proof}

The corollary depends on instances whose optimum is zero. It does not
exclude additive guarantees, parameterized algorithms, or approximation
under a promise that excludes such instances. For the two objectives that
are not deadline objectives, we give approximation algorithms in the next
section.

\section{Algorithms for Unequal Demands with Setup Time}
\label{sec:selection}

Having shown that the three unequal-demand job classes are strongly
NP-hard with setup time, we now give the algorithms that we use for them.
A job order and its recharging decisions depend on one another in this
setting, and no single algorithm covers every objective. We give exact
subset algorithms for small job counts, optimal recharging for a
prescribed order, and approximation algorithms for makespan and for the
sum of completion times. We end the section with instances that show why
the ordering rules of \cref{sec:complete,sec:partial} cannot be reused.

For comparison, the eight equal-energy variants have the same time bounds
under both recharging policies (\cref{tab:equal-energy-selection}). For
partial recharging, the convex allocation of group sizes improves the
quadratic fixed-order \SCT{} algorithm to $O(n\log n)$, and prefix minima
over battery levels reduce the timely-subset algorithm from $O(n^4)$ to
$O(n^3)$. These improvements use equal demands and do not apply to the
other three job classes.

\begin{table}[tbp]
\caption{Exact algorithms for equal energy demands with setup time. The
bounds include sorting where required.}
\label{tab:equal-energy-selection}
\centering
\begin{tabular}{lcc}
\toprule
Objective & Partial & Complete\\
\midrule
$\SCT$ & $O(n\log n)$ & $O(n\log n)$\\
$C_{\max}$ & $O(n)$ & $O(n)$\\
$T_{\max}$ & $O(n^3)$ & $O(n^3)$\\
$U$ & $O(n^3)$ & $O(n^3)$\\
\bottomrule
\end{tabular}
\end{table}

The other three job classes share general exact algorithms
(\cref{tab:algorithms}). These algorithms are exponential in $n$ but do
not depend on the numerical battery capacity. All bounds count arithmetic
operations and stored numerical values; their bit cost depends on the
input lengths.

\begin{table}[tbp]
\caption{General exact algorithms, applicable separately to equal execution
times, $e_i=b_i$, and arbitrary parameters. The same algorithms also solve
equal-energy instances, where the polynomial algorithms are preferable.}
\label{tab:algorithms}
\centering
\begin{tabular}{lcccc}
\toprule
&\multicolumn{2}{c}{Partial recharging}&\multicolumn{2}{c}{Complete recharging}\\
\cmidrule(lr){2-3}\cmidrule(lr){4-5}
Objective & Time & Space & Time & Space\\
\midrule
$\SCT$ & $O(3^n)$ & $O(2^n)$ & $O(n3^n)$ & $O(3^n)$\\
$C_{\max}$ & $O(n2^n)$ & $O(2^n)$ & $O(n3^n)$ & $O(3^n)$\\
$T_{\max}$ & $O(n3^n)$ & $O(n2^n)$ & $O(n^2 3^n)$ & $O(n3^n)$\\
$U$ & $O(3^n)$ & $O(2^n)$ & $O(n3^n)$ & $O(3^n)$\\
\bottomrule
\end{tabular}
\end{table}

\subsection{Subset Algorithms for Partial Recharging}

We first reduce the choice of recharge amounts to a choice of job groups,
so that a dynamic program over subsets of jobs suffices. For a subset
$X\subseteq\mathcal J$, let $e(X)=\sum_{i\in X}e_i$ and
$b(X)=\sum_{i\in X}b_i$. We can transform a partial-recharging schedule so
that every group preceded by a positive recharge ends with an empty
battery, without increasing any completion time. If a group leaves
residual energy, then
we move as much of that energy as possible from its preceding recharge
to the next recharge.
The battery level after the next recharge remains unchanged, and
completions before that next recharge become earlier. If the earlier
recharge vanishes, then
its setup also disappears and the neighboring groups merge. We remove
any terminal surplus in the same way. Repeating this operation gives the
required form.

The schedule therefore consists of an initial free set $F$, with
$b(F)\leq\Bz$, followed by groups of energy at most $\BM$.
The first charged group $G$ acquires
$q=b(G)-(\Bz-b(F))>0$ units; every later charged group acquires its own
demand. Zero-energy groups may execute without a recharge.
We enumerate the free set and the first charged group explicitly. There
are $3^n$ assignments of jobs to these two sets and the remaining set.
If $b(F)=\Bz$, then $F$ can also initialize an empty-battery state directly.
We evaluate the case with no recharging separately.

Within a group, changing the job order does not change feasibility, the
recharge amount, or the ending time. For \SCT{}, we use SJF order.
Let $A(X)$ denote the sum of completion times when the jobs of $X$ execute
consecutively in SJF order from time zero. For deadline feasibility and
maximum tardiness, we use EDF order, and we define
\begin{equation}
 \ell(X)=\max_{j\in X}
 \left(\sum_{i\preceq_X j}e_i-d_j\right),
 \label{eq:group-lateness}
\end{equation}
where $\preceq_X$ is EDF order within $X$. A group starting execution at
time $t$ meets every deadline exactly when $t+\ell(X)\leq0$.
The adjacent-exchange argument for EDF applies because the group contains
no charging interruption. We retain negative values of $\ell(X)$, since
they measure how much delay a group can tolerate.

We precompute subset sums, cardinalities, $A(X)$, and $\ell(X)$ in
$O(n2^n)$ time and $O(2^n)$ space. For example, removing the last SJF job
from $X$ gives $A(X)=A(X\setminus\{j\})+e(X)$. Removing the last EDF job
gives $\ell(X)=\max\{\ell(X\setminus\{j\}),e(X)-d_j\}$, with the
empty-set term omitted. Each subsequent group transition therefore has
constant arithmetic cost.

\paragraph{Sum of completion times.}
Let $D(X)$ be the minimum weighted contribution of actions completing $X$
and ending with an empty battery. Each elapsed time unit before the next
job delays all $n-|X|$ unfinished completions. Appending a disjoint group
$G$ therefore adds
\begin{equation}
 (n-|X|)L(b(G))+A(G)+(n-|X|-|G|)e(G).
 \label{eq:subset-sct-transition}
\end{equation}
We minimize $D(X\cup G)$ over feasible $G$. The first term accounts for
charging, the second for completions within the group, and the third for
the execution delay imposed on jobs outside it.

To initialize this recurrence, the free group contributes
$A(F)+(n-|F|)e(F)$. We then use (\ref{eq:subset-sct-transition}) with
$X=F$ and recharge amount $b(G)-(\Bz-b(F))$ in its first term.
The remaining transitions require $b(G)\leq\BM$.
There are at most $\sum_X2^{n-|X|}=3^n$ subset/group pairs, which gives the
first partial-recharging bound (\cref{tab:algorithms}). Predecessors
recover the groups, their SJF orders, and their recharge amounts.

\paragraph{Maximum tardiness.}
The \SCT{} state omits elapsed time, because weighted action durations
already account for every delayed completion. The deadline objectives
require that time. At an empty-battery boundary after set $X$ and $k$
sessions, conservation of energy determines it as
\begin{equation}
 t(X,k)=e(X)+b(X)-\Bz+hk.
 \label{eq:subset-boundary-time}
\end{equation}
We keep the minimum maximum tardiness $D(X,k)$ and append each feasible
group $G$ using
\begin{equation}
 \max\{D(X,k),\,0,\,t(X,k)+L(b(G))+\ell(G)\}.
 \label{eq:subset-tardiness-transition}
\end{equation}
The destination session count is $k+\mathbf1\{b(G)>0\}$.
For the first charged group, execution starts at
$e(F)+L(b(G)-(\Bz-b(F)))$; the maximum tardiness of the free group is
$\max\{0,\ell(F)\}$, or zero when $F$ is empty.
There are at most $n+1$ session counts per subset, which gives $O(n3^n)$
time and $O(n2^n)$ space.

\paragraph{Number of tardy jobs.}
For this objective we need a timely subset of maximum cardinality, and not
an EDF order of all jobs. Deleting tardy jobs from a schedule cannot delay
the retained
jobs: it removes execution and energy demand, while any recharge can be
reduced to respect capacity. Under complete recharging, a retained recharge
still fills the battery; its amount can only decrease, and a zero amount
also removes its setup. We may therefore move the timely jobs to the front
under either policy and append all omitted jobs afterward.

For partial recharging, we keep the earliest empty-battery completion time
$V(X)$ of a schedule in which every job of $X$ is timely. An earlier time
dominates a later one, because the battery and the remaining jobs are
identical.
We admit group $G$ only when
\begin{equation}
 V(X)+L(b(G))+\ell(G)\leq0,
 \label{eq:subset-timely-transition}
\end{equation}
and minimize $V(X\cup G)$ with candidate
$V(X)+L(b(G))+e(G)$. Free and first charged groups use their actual
starting times and the same deadline check. Among reachable subsets,
including feasible free groups, we select one with maximum cardinality.
Appending all omitted jobs yields the minimum tardy count. There is one time
value per subset, so the time and space bounds are $O(3^n)$ and $O(2^n)$.

\paragraph{Makespan.}
Every partial schedule in the normal form acquires exactly
$\max\{0,S-\Bz\}$ energy, so we need only to minimize its session count.
We view the free group as a bin of capacity $\Bz$ and each later group as
a bin of capacity $\BM$. For each subset $X$, we retain the
lexicographically smallest pair $(k,a)$, where $k$ is the number of later
bins and $a$ is the load of the current bin. We initialize the empty set
with $(0,0)$.
Appending job $j$ replaces the pair by $(k,a+b_j)$ when it fits the
current capacity, and by $(k+1,b_j)$ otherwise.

The retained pair dominates every discarded pair under future insertions.
For equal $k$, a smaller load leaves more capacity. For smaller $k$, an
insertion can increase the count by at most one; if it reaches the other
count, then its new load is only $b_j$, which is no greater than the load
obtained by adding $j$ to an existing bin. This argument also covers the
smaller initial capacity. Every packing gives an order of its groups;
inserting that order greedily uses no more groups. The recurrence
therefore contains a packing with the minimum number of sessions.
Enumerating the final job of each subset gives $O(n2^n)$ time and $O(2^n)$
space. Recharging each recovered group exactly as above attains makespan
$E+\max\{0,S-\Bz\}+hk$.

\subsection{Subset Algorithms for Complete Recharging}

Under complete recharging the normal form is not available, but a complete
recharge determines the battery level independently of earlier recharging
decisions. After completed set $X$, the battery is either
$\Bz-b(X)$, if no positive recharge has occurred, or
$\BM-b(Y)$, where $Y\subseteq X$ comprises the jobs since the last
positive recharge. We store only reachable battery values $\beta$.
Summing at most $1+2^{|X|}$ possibilities over all $X$ gives $O(3^n)$
subset/battery states, regardless of the numerical capacity.

From each state, we try each uncompleted job $j$ without a recharge and
with a full recharge. A feasible action has recharge amount
$q\in\{0,\BM-\beta\}$, duration $a=e_j+L(q)$, and resulting battery
$\beta+q-b_j$.
For \SCT{}, we add $(n-|X|)a$ to the accumulated weighted cost and retain its
minimum per state. For makespan, we retain the earliest time and add $a$.
These dominance rules give $O(n3^n)$ time and $O(3^n)$ space.

For maximum tardiness, we also store the session count $k$. The elapsed
time is
\begin{equation}
 t(X,\beta,k)=e(X)+b(X)-\Bz+\beta+hk.
 \label{eq:complete-subset-time}
\end{equation}
An action updates the objective to
$\max\{D(X,\beta,k),0,t(X,\beta,k)+a-d_j\}$ and increments $k$
only for a positive recharge. The added dimension gives $O(n^2 3^n)$
time and $O(n3^n)$ space.

For the number of tardy jobs, we use the timely-prefix property proved
above. We keep the earliest time per $(X,\beta)$, admit an action only
when its completion meets $d_j$, and select a reachable subset of largest
cardinality. An earlier time permits every continuation of a later state
with the same set and battery, so the state needs no session count. The
bounds are again $O(n3^n)$ time and $O(3^n)$ space.

The subset algorithms give exact schedules for every variant, at a cost
that is exponential in $n$. We use them for small instances and as the
reference optima in \cref{sec:setup-experiments}.

\subsection{Recharging a Prescribed Job Order}

For larger instances, a sequencing rule can supply an order, and we can
still choose the recharges for that order exactly with the fixed-order
algorithms of \cref{sec:setup-partial,sec:setup-complete}. Under partial
recharging, group transitions over consecutive intervals take $O(n^2)$
time for \SCT{} and $O(n^3)$ for either deadline objective. Under complete
recharging, the reachable battery states give the same bounds. These
algorithms accept arbitrary energy demands. They do not imply that the
supplied order is globally optimal.

Makespan has a faster $O(n)$ algorithm under both policies.
For partial recharging, we execute the longest feasible free prefix and
then take maximal consecutive groups of demand at most $\BM$.
For complete recharging, we reserve the longest final suffix whose demand
does not exceed $\BM$, execute its head with full recharges only when
necessary, and fill the battery before the reserved suffix
(\cref{lem:setup-complete-fixed-makespan}). We use both constructions in
the approximation algorithms that follow.

\subsection{Makespan Approximation}

For large instances with unequal energy demands, we can retain a makespan
guarantee without solving the subset recurrences. We first establish a
bound that permits arbitrary initial energy, and then improve it when the
battery starts empty. The first bound combines two factor-two arguments,
one for the number of sessions and one for the acquired energy.

\begin{theorem}[Makespan with arbitrary initial energy]
\label{thm:makespan-approx}
For either recharging policy, there is an $O(n)$-time factor-two approximation
for makespan with any $\Bz\in[0,\BM]$ and $h\geq0$.
\end{theorem}
\begin{proof}
If $S\leq\Bz$, then no recharging is needed and any order is optimal.
Suppose otherwise, and let $Q^*$ and $K^*$ denote acquired energy and
sessions in one makespan-optimal schedule. The maximal-group construction
for a fixed order has
an initial group of load $a_0$ and $K$ charged groups of loads
$a_1,\ldots,a_K$. Maximality gives $a_0+a_1>\Bz$ and
$a_i+a_{i+1}>\BM$ for $i\geq1$.
Any schedule with $K^*$ sessions has $S\leq\Bz+K^*\BM$.
If $K\geq2K^*+1$, then pairing groups $(0,1),(2,3),\ldots,
(2K^*,2K^*+1)$ contradicts this inequality. Hence $K\leq2K^*$.
For partial recharging, acquired energy is already minimal, so any order
with the linear fixed-order algorithm gives the claimed factor.

For complete recharging, the linear fixed-order algorithm uses no more
sessions than that same maximal-group construction. We additionally choose
an order with acquired energy at most twice the minimum possible acquired
energy. We write $H$ for the jobs before the final recharge. Then
\begin{equation}
 Q=\BM-\Bz+b(H),\qquad b(H)\geq S-\BM.
 \label{eq:approx-head-energy}
\end{equation}
If $S\leq\BM$, then a full recharge before all jobs minimizes both
energy and sessions. If $S>2\BM$, then any final group gives
$Q\leq S+\BM-\Bz\leq2(S-\Bz)\leq2Q^*$.

The remaining case is $\BM<S\leq2\BM$. We set $T=S-\BM$.
We compare the smallest single demand at least $T$ with a prefix of positive
demands below $T$ whose sum first reaches $T$, omitting a candidate when
it does not exist. We choose the candidate with smaller total demand as $H$.
Let $H^*$ be a minimum-demand subset of demand at least $T$. If $H^*$
contains a demand at least $T$, then the single-job candidate is no worse.
Otherwise the smaller demands total at least $T$, and their prefix has
demand below $2T\leq2b(H^*)$. Hence $b(H)\leq2b(H^*)$.
Its complement fits in the final group, so (\ref{eq:approx-head-energy})
gives $Q\leq2Q^*$ because $\BM-\Bz\geq0$ and any feasible schedule
has head demand at least $b(H^*)$.
Optimizing the recharges for this order can only reduce the acquired
energy. Linear scans find all candidate sets.

For a makespan-optimal schedule, both $Q\leq2Q^*$ and $K\leq2K^*$
now hold, and consequently
$E+Q+hK\leq2(E+Q^*+hK^*)$.
\end{proof}

The energy bound and the session bound hold against the same optimal
schedule, even though that schedule may trade energy against sessions.
This is why the two factors of two combine into a single factor of two
for the makespan.

An empty initial battery permits a stronger bound from bin packing.
First-fit decreasing (FFD) sorts the positive demands in nonincreasing
order and places each demand in the first bin with enough capacity. We
execute the resulting bins as groups and place a fullest bin last.
Zero-demand jobs can precede these groups without a recharge.

\begin{theorem}[Makespan with an empty initial battery]
\label{thm:ffd}
If $\Bz=0$, then the FFD construction is a factor-$3/2$ makespan
approximation under either recharging policy. A direct implementation takes
$O(n^2)$ time.
\end{theorem}
\begin{proof}
Let $k^*$ be the minimum number of capacity-$\BM$ bins for the positive
demands. D\'osa~\cite{dosa2007ffd} proves
$K_{\rm FFD}\leq11k^*/9+2/3$.
Integrality gives $K_{\rm FFD}=1$ when $k^*=1$ and
$K_{\rm FFD}\leq3$ when $k^*=2$; for $k^*\geq3$, the displayed
bound is at most $3k^*/2$. Every schedule supplies a feasible
packing, so $K_{\rm FFD}\leq3K^*/2$ for a makespan-optimal schedule.

For partial recharging, $Q=S$ is minimal. For complete recharging, let
$f$ be the load of the fullest FFD bin. Then $Q=S+\BM-f$.
If one bin suffices, then $Q=\BM$ is optimal. With at least two bins,
the loads of any two FFD bins sum to more than $\BM$: the first item of
the later bin did not fit in the earlier bin when inserted. Consequently,
$f>\BM/2$, and $S>\BM$. It follows that
$Q<S+\BM/2<3S/2\leq3Q^*/2$.
Combining the energy and session bounds gives the ratio. The all-zero
demand case uses no recharge and is exact under either policy.
\end{proof}

The ratio is $3/2$ and not $11/9$, because the additive constant in
D\'osa's bound cannot be dropped. The guarantee requires $\Bz=0$;
\cref{thm:makespan-approx} covers arbitrary initial energy.

\subsection{Sum of Completion Times Approximation}
\label{sec:sct-approx}

Setup time enters the sum of completion times through the position of each
session, so a sorting rule must weigh energy demand against execution time
in a way that depends on $h$. We give such a rule for an initially empty
battery and positive demands. The rule also accommodates job weights
$w_i>0$ in the objective $\sum_iw_iC_i$; unit weights give $\SCT$. We
define the coefficient
\begin{equation}
 c=\begin{cases}
   \max\bigl\{1,\;(1+2h/\BM)/3\bigr\} & \text{partial recharging},\\
   (1+2h/\BM)/3 & \text{complete recharging}.
 \end{cases}
 \label{eq:sct-coefficient}
\end{equation}
The coefficient $c$ is the weight that the sorting key gives to a unit of
energy demand relative to a unit of execution time. It grows with
$h/\BM$, because every $\BM$ units of demand require at least one setup.

The construction sorts the jobs by nondecreasing $(e_i+cb_i)/w_i$, breaking
ties arbitrarily, and partitions this order into consecutive groups by next
fit: a job joins the current group when the demand of the group stays
within a limit $\theta$, and otherwise starts a new group. Partial recharging uses
$\theta=\min\{\BM,h\}$, and a job with $b_i>\theta$ forms its own group. Complete
recharging uses $\theta=\BM$. The device recharges before each group: by exactly
the demand of the group under partial recharging, and to full capacity under
complete recharging. When $h=0$ under partial recharging, every job is its
own group.
Comparisons use the integers $3\BM e_i+\max\{3\BM,\BM+2h\}b_i$ (partial)
or $3\BM e_i+(\BM+2h)b_i$ (complete), divided by $w_i$, so the algorithm
takes $O(n\log n)$ arithmetic operations.

\begin{theorem}[Weighted completion times with an empty initial battery]
\label{thm:sct-approx}
Let $\Bz=0$, $b_i>0$ for every job, and $w_i>0$. Under either charging
policy, the construction above returns a feasible schedule with
$\sum_iw_iC_i\leq3\,\mathrm{OPT}$, where $\mathrm{OPT}$ is the optimum for
the same policy.
\end{theorem}
\begin{proof}
Define
\[
 d=\begin{cases}
   2h/3 & \text{partial recharging with } h\leq\BM,\\
   (h+\BM)/3 & \text{otherwise.}
 \end{cases}
\]
The two partial cases agree at $h=\BM$, and $c=1$ in the first case.
Throughout, $E$ and $P$ denote the execution time and demand of the jobs
completed up to and including a fixed completion.

\emph{Upper bound.} Consider the constructed schedule, and let $K$ be the
number of groups started by the fixed completion. Suppose $h>0$ under
partial recharging, or any $h$ under complete recharging. Truncate the
current group after the completed job. Any two adjacent groups have
combined demand above $\theta$: the first job of the later group did not fit in
the earlier one, and a singleton with $b_i>\theta$ exceeds $\theta$ by itself.
Pairing groups $(1,2),(3,4),\ldots$ when $K$ is even and $(2,3),(4,5),\ldots$
when $K$ is odd gives $P>\lfloor K/2\rfloor \theta$, hence $K<2P/\theta+1$. The
energy acquired so far is at most $P+\theta$: under partial recharging the
excess is the demand of the later jobs of the current group, and under
complete recharging the residual battery is at most $\BM=\theta$. Each of the $K$
sessions costs one setup, so
\begin{equation}
 C\leq E+P+\theta+h\Bigl(\frac{2P}{\theta}+1\Bigr)
  =E+\Bigl(1+\frac{2h}{\theta}\Bigr)P+h+\theta.
 \label{eq:sct-approx-upper}
\end{equation}
For partial recharging with $h\leq\BM$, $\theta=h$ and the right side is
$E+3P+2h$. Otherwise $\theta=\BM$, the coefficient of $P$ equals $3c$, and the
constant equals $3d$. When $h=0$ under partial recharging, $C=E+P$. In every
case,
\begin{equation}
 C\leq3(E+cP+d).
 \label{eq:sct-approx-three}
\end{equation}

\emph{Lower bound.} Consider any feasible schedule and any of its
completions. Since demands are positive, $P>0$ and at least one session has
started. A session acquires at most $\BM$, so at least $\max\{1,P/\BM\}$
sessions have started, and the acquired energy is at least $P$. Under
complete recharging the first session acquires $\BM$, so the acquired
energy is at least $\max\{\BM,P\}$. Hence
\begin{equation}
 C\geq\begin{cases}
   E+P+h\max\{1,P/\BM\} & \text{(partial)},\\
   E+(\BM+h)\max\{1,P/\BM\} & \text{(complete)}.
 \end{cases}
 \label{eq:sct-approx-lower}
\end{equation}
Each right side is at least $E+cP+d$. For partial recharging with
$h\leq\BM$ this reads $P+h\geq P+2h/3$. In the remaining cases both sides
are piecewise linear in $P$ with a breakpoint at $P=\BM$. At $P=\BM$ the
difference is $\BM/3$; as $P\to0$ it is $(2h-\BM)/3\geq\BM/3$ for partial
recharging with $h\geq\BM$ and $2(\BM+h)/3$ for complete recharging; and
for $P\geq\BM$ the slope $1+h/\BM$ of the left side exceeds $c$. Therefore
every feasible schedule satisfies
\begin{equation}
 C\geq E+cP+d
 \label{eq:sct-approx-lb}
\end{equation}
at every completion.

\emph{Sequencing.} Put $p_i=e_i+cb_i$ and $W=\sum_iw_i$. Summing
(\ref{eq:sct-approx-lb}) over the jobs of any schedule with their weights
shows that $\mathrm{OPT}$ is at least $dW$ plus the minimum weighted sum of
completion times of ordinary jobs with processing times $p_i$. Sorting by
$p_i/w_i$ attains that minimum, because exchanging adjacent jobs $i$ and
$j$ changes the objective by $w_ip_j-w_jp_i$. The construction uses this
order, so summing (\ref{eq:sct-approx-three}) with the same weights bounds
its objective by three times the same quantity.
\end{proof}

Optimizing the recharges for the returned order with the fixed-order
algorithm can only improve the objective. The guarantee uses both
hypotheses: a zero-demand job can complete before the first setup, which
violates (\ref{eq:sct-approx-lb}), and initial energy changes the lower
bound. We do not know whether the factor three can be improved; the
largest ratios that we observed were $491/400$ and $13/10$
(\cref{sec:setup-experiments}). The zero-setup order by $e_i+b_i$ cannot
replace the setup-dependent key, even with optimal recharging
(\cref{thm:a1-order-lower-bound}).

\subsection{Limits of the Zero-Setup Rules}

A positive setup time invalidates several rules that are sufficient
without setup time (\cref{sec:complete,sec:partial}). In each instance
below, the failure lies in the job order: optimizing the recharges for
that order does not remove it.

EDF order can fail even when all execution times are one. Consider energy
demands $(2,2,3,3)$, deadlines $(33,34,35,36)$, $\BM=5$, $\Bz=0$, and
$h=10$. The best partial recharging of the EDF order has maximum tardiness
eight. The order $(1,3,2,4)$ with two recharges of five units meets every
deadline. Retaining EDF order and repairing only its recharge amounts
therefore does not restore optimality for unequal demands.

Under partial recharging, the Moore--Hodgson deletion rule can fail even
with equal energy. Let
$(e_i,b_i,d_i)$ be $(1,1,7)$, $(2,1,9)$, and $(2,1,11)$, with
$\BM=2$, $\Bz=0$, and $h=5$.
Job 1 alone is timely, but neither pair $(1,2)$ nor $(1,3)$ is feasible
in EDF order. Deleting the longest job at each infeasible prefix therefore
retains only job 1, even with an exact test of recharging feasibility.
Jobs 2 and 3 instead share a recharge of two units and finish at times nine
and eleven. This is why the equal-energy algorithm for the number of tardy
jobs selects subsets (\cref{thm:setup-partial-equal-tardy}).

The recharge rule of $\mathcal A_1$ can have an unbounded \SCT{} ratio.
For $n$ identical jobs with $e_i=b_i=1$, $\BM=n$, and $\Bz=0$,
recharging only the demand of the next job gives \SCT{}
$(h+2)n(n+1)/2$. A single initial recharge gives
$nh+n^2+n(n+1)/2$. As $h$ increases, their ratio tends to $(n+1)/2$.
Optimizing the recharges for the job order of $\mathcal A_1$ removes this
particular failure, but the next theorem shows that it cannot restore a
constant factor.

\begin{theorem}[The zero-setup order with optimal recharging]
\label{thm:a1-order-lower-bound}
For every integer $m\geq2$ and every $\varepsilon>0$ there is an instance
with $n=m^2+m$ jobs, $\Bz=0$, and a sufficiently large setup time on which
the order by nondecreasing $e_i+b_i$, with any recharging decisions, has
$\SCT$ at least $((2m+1)/3-\varepsilon)$ times the optimum under either
policy.
\end{theorem}
\begin{proof}
Let $\BM=m^2$ and take $m$ jobs with $(e,b)=(1,\BM)$ and $m^2$ jobs with
$(e,b)=(\BM+1,1)$. Since $\BM+1<\BM+2$, the order places every
capacity-$\BM$ job first. Each of them needs its own session, so the $r$-th
completes after at least $r$ setups, and every later job waits for at least
$m+1$ setups because the battery is empty after the $m$-th. The setup
contribution is therefore at least $h[m(m+1)/2+m^2(m+1)]$ for any
recharging of this order. The schedule that executes the $m^2$ unit-demand
jobs as one
group and then each large job in its own session is feasible under both
policies and has setup contribution $h[m^2+m(m+3)/2]$. Its remaining
contribution is at most $n(E+S)=O(m^6)$. The two setup contributions have
ratio exactly $(2m+1)/3$, so the ratio of the two schedules tends to
$(2m+1)/3$ as $h\to\infty$, and the optimum is no larger than the second
schedule.
\end{proof}
The zero-setup order therefore has ratio $\Omega(\sqrt n)$ even with
optimal recharging; we do not have a matching upper bound. The setup time
must enter the sorting rule itself, as it does in \cref{thm:sct-approx}.

The order of $\mathcal A_2$ can also exceed its $5/4$ makespan bound after
we optimize its recharges. Consider six unit-execution jobs with demands
$(2,2,2,3,3,3)$, $\BM=5$, $\Bz=0$, and $h=1{,}000$.
The order of $\mathcal A_2$ has best complete-recharging makespan $4{,}021$.
Pairing each demand-two job with a demand-three job gives $3{,}021$,
which attains the lower bound of three sessions, 15 acquired energy units,
and six execution units. The ratio exceeds $5/4$.

With setup time, we therefore use the polynomial exact algorithms for
equal energy demands, the subset algorithms for exact solutions of small
general instances, \cref{thm:makespan-approx,thm:ffd} for makespan
approximation, and \cref{thm:sct-approx} for the sum of completion times
when the battery starts empty and every job needs energy. Other job orders
can still serve as heuristics when they are followed by optimal
recharging. We have no approximation ratio for such heuristics, and none
for the sum of completion times when the initial energy is positive or a
job has zero demand.

\section{Evaluation with Setup Time}
\label{sec:setup-experiments}

The results of the preceding four sections are proved, and the experiments
in this section have a narrower purpose than those of
\cref{sec:experiments}. We first checked our implementations of the
setup-time algorithms against independent enumeration. We then measured
how often the polynomial algorithms attained an exact optimum on small
synthetic inputs. Finite comparisons of this kind can expose implementation
errors, and they indicate typical behavior on the tested inputs; they do
not establish correctness or approximation ratios.

\subsection{Correctness Checks}

We used two independent reference implementations. The first, which we
call the \emph{integer-state oracle}, enumerates both the next job and
every feasible integer recharge amount; its state records the completed
subset, the battery level, and the number of sessions. The second, which we
call the \emph{literal enumerator}, advances the elapsed time directly
through each permutation and each sequence of recharges, without using a
grouping normal form or the energy-conservation identity for elapsed time.

We compared the oracle with the literal enumerator on 3,418 instances with
zero to two jobs, capacities one to three, execution times zero or two,
every feasible energy demand and initial energy, deadlines zero or four,
and setup time zero or three. These instances produced 27,344
policy/objective comparisons. We also checked the fixed-order and
equal-energy algorithms on these inputs and on 240 seeded instances with
three to six jobs. All 29,264 fixed-order comparisons and all 10,424
equal-energy comparisons agreed.

We then checked the remaining algorithms. The subset algorithms matched
the integer-state oracle on 27,344 exhaustive and 4,000 random
policy/objective comparisons; the random cases used three to seven jobs.
Sixteen scaling checks multiplied the time and energy parameters by
$2^{70}$: the number of tardy jobs remained unchanged, and the time-valued
objectives scaled by the same factor, as expected for algorithms that do
not depend on the numerical capacity. A 16-job packing instance with a
large integer capacity attained its lower bound of four sessions.

For the polynomial algorithms, 7,620 fixed-order makespan comparisons,
10,224 equal-energy partial-\SCT{} comparisons, and 15,202 equal-energy
partial-tardy-count comparisons had no mismatches. The approximation
checks included 180 instances for each general factor-two construction,
180 for the order of $\mathcal A_2$ with optimized recharges, and 360
policy/instance checks of the FFD construction with an empty battery. We
also replayed the counterexamples of \cref{sec:selection} and evaluated 32
YES/NO instances of the hardness constructions.

For the factor-three \SCT{} construction (\cref{thm:sct-approx}), we
compared the returned value, the intermediate lower bound
$\sum_iw_i(E_i+cP_i+d)$ of its proof, and the exact optimum on 7,532
policy/instance pairs. These pairs comprise every multiset of one to three
jobs with execution time zero or two, positive demand, capacity one to
three, and setup time in $\{0,1,\BM,2\BM,10\BM\}$, together with 2,000
seeded random instances with up to seven jobs and capacities up to
$10^6$. A further 2,000 weighted comparisons used a literal integer-battery
oracle with weights from one to ten. The largest observed ratios were
$491/400$ and $13/10$ without weights and $2{,}601/1{,}858$ with weights,
all well below the proved factor of three.

On the family of \cref{thm:a1-order-lower-bound} with $m=2,4,8,16$, the
zero-setup order with optimal recharging cost up to $10.99995$ times the
witness schedule at $272$ jobs, whereas the factor-three construction
matched that witness. This agrees with the theorem, which predicts a ratio
that approaches $(2m+1)/3=11$ for $m=16$.

For the final-job enumeration without setup time (\cref{thm:final-job}),
6,000 seeded instances with up to nine jobs, including zero execution
times, zero demands, and empty or full initial batteries, satisfied every
stated inequality against the exact subset dynamic program. The
enumeration improved on $\mathcal A_1$ in 69 of these instances. We
replayed the tight family, the instance on which the enumeration is not
exact, and the instance that violates
$\mathcal A_1\leq\mathrm{OPT}+M$, with $m$ up to $1{,}000$ for the family.

These counts overlap, and they are not counts of independent statistical
trials.

\subsection{Comparison on Paired Setup Times}

To observe how close the polynomial algorithms come to the optimum, we
generated ten six-job sets for each of the four job classes, which gives
40 sets. The capacity was an integer from three to ten. Energy demands
were positive integers within the capacity, execution times were integers
from one to eight except when set equal to the energy demand, and
deadlines were integers from five to 80. Each set retained its jobs,
capacity, initial energy, and deadlines while the setup time varied over
$0,1,5,20,100$, so that the comparisons are paired across setup times.
Half of the sets in each class started with an empty battery; the others
used an initial level sampled from zero through the capacity. The
resulting 200 instances gave 1,600 comparisons across the two policies and
the four objectives.

For equal energy demands, we executed the exact polynomial algorithm.
Otherwise, we tried several job orders (SJF, increasing $e_i+b_i$,
increasing energy, decreasing energy, EDF, and the order of
$\mathcal A_1$), optimized the recharges for each, and kept the best
schedule. For makespan, the candidates also included the factor-two
construction, FFD when $\Bz=0$, and the order of $\mathcal A_2$ under
complete recharging. For \SCT{}, they included the factor-three
construction when $\Bz=0$; every generated demand is positive. We compared
the best candidate with the exact optimum from the subset algorithms. The
result is the best of several orders, so we cannot attribute its quality
to any single sorting rule.

The equal-energy algorithms attained the exact optimum in all 400
policy/objective comparisons for that class, as the proofs require. The
other three classes contributed 150 comparisons per policy/objective pair
(\cref{tab:setup-results}). These classes permit unequal demands; an
individual generated instance may also satisfy a more restrictive class.

\begin{table}[tbp]
\caption{Best polynomial candidate versus the exact optimum on the three
job classes that permit unequal energy demands. Each row contains 150
comparisons.
A gap is the candidate objective minus the optimum. Time-valued gaps use
the normalized time unit; $U$ gaps count jobs. Deadline ratios are omitted
because some exact optima are zero.}
\label{tab:setup-results}
\centering
\begin{tabular}{llrrr}
\toprule
Policy & Objective & Optimal / 150 & Largest gap & Largest ratio \\
\midrule
Partial & $\SCT$ & 108 & 294 & 1.329 \\
Partial & $C_{\max}$ & 134 & 100 & 1.459 \\
Partial & $T_{\max}$ & 111 & 99 & -- \\
Partial & $U$ & 108 & 2 & -- \\
Complete & $\SCT$ & 94 & 298 & 1.331 \\
Complete & $C_{\max}$ & 130 & 100 & 1.459 \\
Complete & $T_{\max}$ & 106 & 100 & -- \\
Complete & $U$ & 104 & 2 & -- \\
\bottomrule
\end{tabular}

\end{table}

The best candidate was optimal in between 94 and 134 of the 150
comparisons, depending on the policy and the objective. The largest
observed \SCT{} ratios were $198/149$ under partial recharging and
$1197/899$ under complete recharging, approximately 1.329 and 1.331.
These values are observations about this sample and not approximation
guarantees; in particular, we have no \SCT{} guarantee for the job sets
with positive initial energy. For makespan, the largest observed ratio was
$159/109$, approximately 1.459, under either policy, and every tested
makespan schedule satisfied the proved bound that applies to it.

For the deadline objectives, a ratio would conceal the cases in which the
optimum is zero: 45 instances under partial recharging and 39 under
complete recharging, for each deadline objective. We therefore report
absolute gaps. The best candidate missed at most two more deadlines than
the optimum. Its largest excess maximum tardiness was 99 time units under
partial recharging and 100 under complete recharging. These gaps are
large relative to deadlines of at most 80, and they are consistent with
\cref{cor:deadline-inapprox}: when the question is whether every deadline
can be met, the exact algorithms remain necessary.

\subsection{Limits of the Evaluation}

We did not repeat the sensor-trace or scalability experiments of
\cref{sec:experiments} under the setup-time model. The comparison above
has six jobs per set and a small synthetic parameter range, so it does not
indicate how the polynomial algorithms perform on larger or more realistic
workloads. We have not measured setup times or charging rates on a
physical device.

\section{Related Work}
\label{sec:related}

We relate our model to work on battery-constrained routing and task
scheduling for particular applications, on recharge durations and setup
time, on batching and packing, and on energy-aware scheduling more
broadly.

\subsection{UAV and Drone Scheduling}

Most work on scheduling for unmanned aerial vehicles (UAVs) under battery
constraints treats it as a routing problem. Pasha et al.~\cite{pasha2022drone} survey drone
scheduling in logistics, aerial monitoring, and healthcare delivery.
Alyassi et al.~\cite{alyassi2022tase} develop a machine-learning energy
model and a constant-factor approximation algorithm for joint flight tour
planning and recharging, and validate it on commercial UAVs. Their energy
model is richer than ours: consumption is a continuous function of
trajectory and payload, where we assume a fixed quantity per job. Their
problem is a fuel-constrained tour (which sites to visit, in what order,
with charging stops on the way), without per-job deadlines or a
classification across objectives. We take the execution time and the
energy demand of each job as inputs, as they are once a route is fixed,
and determine the complexity of choosing the job order and the recharges.

Jana and Mandal~\cite{jana2023drone} give a fully polynomial-time
approximation scheme for single-drone delivery scheduling and
constant-factor approximations for several drones, maximizing delivery
profit over time-windowed orders subject to battery budgets. As in the
routing formulations, the battery is a per-tour budget and not a
recharging process interleaved with execution.

\subsection{Warehouse Robots and Electric Ground Vehicles}

Warehouse robots execute tasks in sequence and can recharge between them,
which matches our fixed job costs and operator-controlled recharging.
Zhen et al.~\cite{zhen2024warehouse} develop a column-and-row generation
algorithm for order picking with autonomous mobile robots (AMRs): their
model assigns batched orders to the robots and routes the robots and the
human pickers to minimize the makespan of a batch.
Tomy et al.~\cite{tomy2020battery} formulate battery charge
scheduling for long-horizon autonomous robots as a multi-objective Markov
decision process. Fragapane et al.~\cite{fragapane2021amr} found that
fewer than 4\% of AMR planning papers address battery constraints
explicitly. For electric forklifts, Lee et al.~\cite{lee2022forklift}
jointly optimize routing and charging schedules as an integer program.
Mandal et al.~\cite{mandal2025evcharging} give a dynamic-programming
strategy for fixed-route electric-vehicle charging under a discrete
partial-charging model, in which the operator also chooses the amount of
charge.

\subsection{Recharge Durations and Setup Time}

Yu and Oron~\cite{yu2025recharging} study single-machine scheduling with
rechargeable energy and recharges of fixed duration. They minimize the
number of late jobs and maximize the weighted number of jobs completed
exactly at their due dates, and analyze parameterizations by the numbers
of distinct due dates, processing times, and energy demands; both problems
are fixed-parameter tractable in the combined numbers of energy demands
and due dates. A fixed-duration recharge differs from ours, whose length
$h+q$ depends on the amount acquired as well as on the decision to start.
Our results concern the amount-dependent model and do not cover their
just-in-time objective.

Li et al.~\cite{li2025smart} study automated guided vehicle scheduling
with flexible charging and a fixed setup time per charging operation.
They minimize makespan with a mixed-integer linear formulation and an
exact logic-based Benders decomposition, whose master problem assigns
tasks to vehicles and whose subproblems optimize each vehicle's charging.
We analyze a single device under both recharging policies, with four
objectives and four job classes, and give complexity results and
approximation guarantees. We have not compared our algorithms with the
single-vehicle case of their formulation.

Setup times also arise when a processor switches between types of work.
Gopalakrishnan et al.~\cite{gopalakrishnan2025coldstart} study a
uniprocessor that executes jobs from several job families, where a setup
time -- or cold start delay -- is incurred whenever a job follows a job
of a different family. This situation occurs in serverless computing,
where a request may need a suitable container to be provisioned. They
minimize the number of tardy jobs, show that the offline problem is
NP-hard even when all jobs have the same execution time, and give an
online policy that buffers jobs and is nearly 1-competitive when every job
has sufficiently large slack. In that model the job sequence triggers a
setup. In our model the scheduler triggers a setup by choosing to
recharge, and the duration of the recharge also depends on the energy
acquired. The two models share one consequence: jobs should be grouped so
that they share a setup, and EDF order can fail once setup takes time
(\cref{sec:selection}).

\subsection{Batching and Packing}

A recharge can supply several jobs, so charging groups resemble batches.
The completion convention matters here. Shabtay and
Steiner~\cite{shabtay2007batch} study resource-controlled setup and
processing times under batch availability, where all jobs in a batch
complete together. Our jobs complete individually, each after its own
execution, so their results do not establish the recurrences that we use.

Bin packing corresponds more closely to makespan under partial recharging
with $\Bz=0$. Each recharge supplies one bin of total demand at most $\BM$,
and the optimal
makespan is $E+S+hK^*$, where $K^*$ is the minimum number of bins.
D{\'o}sa~\cite{dosa2007ffd} proves that first-fit decreasing uses at most
$11K^*/9+2/3$ bins, which with the scheduling identity gives a makespan
ratio of $3/2$ (\cref{thm:ffd}). Complete recharging needs the separate
final-group energy argument in that theorem, and arbitrary initial energy
needs the factor-two construction of \cref{thm:makespan-approx}.

\subsection{CubeSat and Nanosatellite Scheduling}

Nanosatellite scheduling also couples task execution with battery
availability, but orbital charging windows add constraints that the
\SCP{} lacks. Seman et al.~\cite{seman2023nanosatellite} formulate offline
nanosatellite scheduling as an energy-constrained mixed-integer linear
program and solve it by branch-cut-and-price. Rouzot et
al.~\cite{rouzot2025nanosatellite} apply integer and constraint
programming to the partition scheduling variant. Our model has an offline
batch with recharging available between any two jobs, so our hardness
results do not transfer directly to either formulation.

\subsection{Energy Harvesting and Intermittent Computing}

In energy-harvesting systems the environment determines when energy
arrives; in our model the operator chooses when to recharge. Islam and
Nirjon~\cite{islam2020scheduling} report that deadline-aware scheduling on
intermittent systems improves on EDF baselines. Caruso et
al.~\cite{caruso2018dynamic} establish NP-hardness and develop dynamic
programs for a related energy-harvesting scheduling problem. Schieber et
al.~\cite{schieber2023harvesting} draw the connection to
inventory-constrained machine scheduling.

Intermittent computing addresses a different problem: preserving the
correctness of a computation across uncontrolled power interruptions. Lucia et
al.~\cite{lucia2017intermittent} survey intermittent execution models.
Lucia and Ransford~\cite{lucia2015dino} propose DINO, a checkpoint-based
execution model for batteryless devices, and Hester et
al.~\cite{hester2017mayfly} add timekeeping guarantees for intermittently
powered sensors.

\subsection{Energy-Aware Scheduling}

Other energy-constrained scheduling models treat energy consumption or
arrival as stochastic, or choose execution parameters to reduce energy
use. Moser et al.~\cite{moser2006lazy,moser2006realtime,moser2007realtime}
propose the Lazy Scheduling Algorithm and prove it optimal for minimizing
missed jobs under energy clairvoyance with preemption. Maeng and
Lucia~\cite{maeng2020adaptive} develop task isolation for time-critical
workloads on energy-harvesting devices. Gopalakrishnan et
al.~\cite{gopalakrishnan2006finite} study finite-horizon scheduling under
joint energy and time constraints for radar systems. Huang and
Neely~\cite{huang2012utility} develop an online scheduling algorithm for
energy-harvesting packet networks via Lyapunov optimization. In the
embedded and real-time literature, Bambagini et
al.~\cite{bambagini2016energy} and Sheikh and
Pasha~\cite{sheikh2018multicore} survey energy-aware scheduling on uni-
and multiprocessor real-time systems; that work mostly minimizes energy
for a fixed workload through voltage and frequency scaling and power
management. We fix the energy demand of each job and choose job order and
recharging to minimize a scheduling objective.

\section{Conclusions}
\label{sec:conclusions}

We set out to determine which versions of the Scheduling-and-Charging
Problem are solvable in polynomial time and how well the others can be
approximated. The answer depends on whether a recharge has a fixed cost.

Without setup time, partial recharging reduces every objective to an
ordering problem, because the recharge durations follow from the job
order. Classical sequencing rules solve the structured job classes, EDF
order minimizes maximum tardiness for every job class, and every order has
the same makespan (\cref{thm:makespan-partial}). The two NP-hard variants
are the average completion time, which we analyzed through the
equivalent sum of completion times, and the number of tardy jobs, both
with arbitrary $(e_i,b_i)$. Under complete recharging, equal energy demands
give the only polynomial case for each objective. For the average completion time under partial recharging, we gave a 2-approximation whose
factor is tight for the algorithm (\cref{thm:2approx}) and a variant with
factor $2-2/(n+1)$ (\cref{thm:final-job}). For makespan under complete
recharging, we gave a $5/4$-approximation (\cref{thm:a2-guarantee}).

A positive setup time makes the number and the positions of the charging
sessions part of the scheduling decision under both policies. Equal energy
demands still admit polynomial exact algorithms for all four objectives.
For the other three job classes, all objectives are strongly NP-hard when
the setup time is an input, and ten variants that are polynomial without
setup time become hard. For these classes we gave subset algorithms that
find exact schedules without enumerating battery levels, makespan
approximations with factor two for any initial energy and $3/2$ for an
empty battery, and a sorting rule that depends on the setup time and has
factor three for the weighted average completion time when the battery
starts empty and every job needs energy.

The two parts of the paper are not equally complete. Without setup time,
the classification is complete, and the open questions concern
approximation. With setup time, our hardness result for the average completion time chooses the setup time from the input, our guarantee for
that objective requires an empty battery and positive demands, and for the
deadline objectives we can only optimize the recharges for candidate
orders, with no approximation factor. Our evaluation with setup time uses
six-job instances. All experiments use an offline model with a constant
charging rate, and we did not execute the schedules on a physical battery
and charger.

\subsection*{Open Problems}

\emph{Average completion time without setup time.}
Our NP-hardness result holds for binary-encoded inputs and is not a strong
NP-hardness result, so it does not exclude a fully polynomial-time
approximation scheme. Such a scheme would have to control how rounding changes both the
jobs selected for the free window and the SSF order of the remaining jobs.
We also do not know whether a polynomial algorithm achieves a constant
factor below two, independent of $n$.

\emph{Tightness of the $5/4$-approximation for makespan.}
The guarantee of $\mathcal{A}_2$ rests on a $5/4$ ratio for its SS-MIN
subroutine and a $3/4$ ratio for its SS-MAX subroutine
(\cref{thm:a2-guarantee}). The additive $\sum_i e_i$ term in $C_{\max}$
reduces the scheduling ratio relative to the subroutine ratio. A matching
family would therefore have to drive $\sum_i e_i$ to zero relative to the
optimal recharging time while respecting the battery capacity
(\cref{thm:ssmin-tight}), and we have not identified such a family.

\emph{Complete recharging without setup time.}
We do not have constant-factor approximations for the average completion time or the number of tardy jobs. The scheduler must assign jobs to
groups that fit within the battery capacity, so bin-packing techniques may
help to construct and analyze such approximations.

\emph{Average completion time with setup time.}
We do not know whether the factor three is tight, and the analysis does
not cover a positive initial energy or zero-demand jobs. Hardness for each
fixed positive setup time is also open. The zero-setup order by $e_i+b_i$
has ratio $\Omega(\sqrt n)$ even with optimal recharging
(\cref{thm:a1-order-lower-bound}), so any improvement must again use a
sorting rule that depends on the setup time.

\emph{Deadline objectives with setup time.}
Unless P$=$NP, no polynomial algorithm has a finite multiplicative
guarantee on all instances (\cref{cor:deadline-inapprox}). Additive
guarantees, and algorithms parameterized by the number of job types or
sessions, are more suitable directions.

\emph{Several devices and online arrivals.}
Even two devices with the makespan objective are open. Delivery requests
and warehouse orders also arrive during execution; an online scheduler
must order jobs and recharge without knowledge of future arrivals, and a
competitive analysis would quantify the cost of that missing information.

\emph{Battery and charger models.}
Lithium-ion charging is fast at a low state of charge and slows above
about 80\%, which suggests a recharge duration that is convex in the
energy acquired; this change may alter which variants are solvable in
polynomial time. Setup that consumes energy, charger availability that
varies over time, and chargers shared among devices would each require
additional constraints. Measured setup times and charging rates would show
whether the scheduling decisions studied here matter for a physical device.

\bibliographystyle{IEEEtran}
\bibliography{references}
\appendix

\section{Proofs}
\label{apx:proofs}

We give the group-exchange arguments, dynamic programs, and reduction
details omitted from the main text. We then prove the guarantees of the
two subset-sum routines that $\mathcal{A}_2$ uses.

\subsection{Proof of \cref{lem:sjf-groups} (SJF Within Groups)}
\label{apx:lem:sjf-groups}

\begin{proof}
We separate the recharge and execution contributions to total completion
time by expanding the completion formula (\ref{eq:completion}):
\begin{equation}
  \sum_{j=1}^n C_j
  = \underbrace{\sum_{j=1}^n (n-j+1)\, r_j}_{A}
  + \underbrace{\sum_{j=1}^n (n-j+1)\, e_j}_{B}.
  \label{eq:sct-expand}
\end{equation}
Since all jobs consume the same energy $b$, reordering jobs within any group
does not change which slots carry recharging intervals, so term $A$ is
unchanged.
SJF order minimizes term $B$ by placing smaller execution times first.
Any schedule not in SJF order therefore has a weakly larger $B$ and hence a
weakly larger sum of completion times.
\end{proof}

\subsection{Proof of \cref{lem:full-groups} (Full Interior Groups)}
\label{apx:lem:full-groups}

\begin{proof}
Assume $b>0$; if $b=0$, then the device needs no recharge.
Let $\Pi^*$ be optimal and let $G$ be a group strictly between the first
and last recharging intervals, with $m_a < m^*$ jobs.
Let $G'$ be the immediately succeeding group with $m_{a'}$ jobs.
If $m_a + m_{a'} \le m^*$, then we may merge $G$ and $G'$ into a single group,
strictly improving the objective (by reducing the number of recharging intervals),
contradicting optimality.
Hence $m_a + m_{a'} > m^*$.
Create schedule $\Pi^{**}$ by moving $m^* - m_a$ jobs from $G'$ to the end of
$G$.
The change in the total sum of completion times is:
\[
  C^{**} - C^* = (m_a - m^*)(m^* - m_{a'})\, b.
\]
This move preserves the job order and hence the execution contribution.
Both $G$ and $G'$ start from a full battery. Moving
$x=m^*-m_a$ jobs across their boundary replaces the recharge contribution
$m_{a'}m_a b$ by $(m_{a'}-x)(m_a+x)b$; the move preserves all earlier and later contributions. This gives the displayed difference.
Since $m_a < m^*$ and $m_{a'} \le m^*$ we have $C^{**} - C^* \le 0$, so
$\Pi^{**}$ is at least as good as $\Pi^*$.
Iterating this argument across all non-full interior groups yields the claim.
\end{proof}

\subsection{Proof of \cref{thm:sct-complete-nph} (\SCT{}, Complete, NP-hard)}
\label{apx:thm:sct-complete-nph}

\begin{proof}
We give separate decision thresholds for equal execution times and for
$e_i=b_i$. The source problem is EquiPartition: given $2n$ positive integers
$a_1,\ldots,a_{2n}$ with sum $2A$, decide whether exactly $n$ of them sum
to $A$. This problem is NP-hard: a Partition instance
$x_1,\ldots,x_n$ of sum $2S$ maps to the $n$ values $x_i+1$ and $n$
additional unit values, with target $S+n$ and cardinality $n$.
We may restrict to $n\ge2$, since we can decide smaller instances directly.

\emph{Construction and feasible groups.}
Set
\[
  L=6nA,\qquad b_i=L+a_i,\qquad \BM=\Bz=nL+A.
\]
The total energy is $2\BM$. Any $n+1$ jobs exceed capacity, since
$(n+1)L>nL+A$. Conversely, every set of at most $n-1$ jobs fits, since
its energy is at most $(n-1)L+2A<\BM$.
In particular all jobs are individually feasible. A schedule with exactly
two groups must therefore have $n$ jobs in each. Both groups must consume
exactly $\BM$, so such a schedule exists if and only if EquiPartition is YES.
Schedules with more than two groups are also feasible, since smaller
groups fit within the battery. We show that every such schedule exceeds
the decision threshold.

\emph{Cost of every other grouping.}
Omit zero-duration recharges and write $G_1,\ldots,G_t$ for the nonempty
groups, $m_g=|G_g|$, and $E_g=\sum_{i\in G_g}b_i$.
Because the battery starts full, the cumulative recharge before $G_g$ is
$\sum_{f<g}E_f$. Thus the recharge contribution to the sum of completion
times is
\begin{align*}
 R&=\sum_{f<g}m_g E_f\ \ge\ L K,\\
 K&=\sum_{f<g}m_fm_g
   =\frac{4n^2-\sum_g m_g^2}{2}.
\end{align*}
Every $m_g\le n$. Unless the sizes are exactly $(n,n)$, there are at least
three groups and
$\sum_g m_g^2\le n^2+(n-1)^2+1$.
To see this, allow zero entries and transfer one unit from a smaller
positive entry to a larger entry below $n$; this increases the sum of
squares. Continue while keeping at least three entries positive.
The terminal maximum is $(n,n-1,1)$, up to permutation.
Consequently every other size pattern satisfies
\[
 K\ge n^2+n-1,\qquad R\ge L(n^2+n-1).
\]
For a feasible $(n,n)$ schedule, $R=n\BM=n^2L+nA$.
The inequality $L(n-1)>5nA$ will separate these schedules from all other
groupings, even when execution times depend on the input.

\emph{Equal execution times.}
Set $e_i=1$. The execution contribution is always $n(2n+1)$.
Use decision threshold
\[
 H_{\rm eq}=n(2n+1)+n\BM.
\]
A YES instance has a feasible $(n,n)$ schedule of cost $H_{\rm eq}$.
Every other grouping has recharge contribution strictly larger than
$n\BM$, because $L(n-1)>nA$. Hence a schedule of cost at most
$H_{\rm eq}$ exists if and only if EquiPartition is YES.

\emph{Energy equal to execution time.}
Set $e_i=b_i=L+a_i$. For any job order $\pi$, the execution contribution is
\[
 L n(2n+1)+\sum_{j=1}^{2n}(2n-j+1)a_{\pi(j)}.
\]
The second term is positive and at most $2n\sum_i a_i=4nA$.
Use the separate threshold
\[
 H_{\rm prop}=L(3n^2+n)+5nA.
\]
If EquiPartition is YES, then a $(n,n)$ schedule has recharge contribution
$n^2L+nA$ and therefore total cost at most $H_{\rm prop}$.
Every other grouping has cost strictly greater than
\[
 L n(2n+1)+L(n^2+n-1)
 =L(3n^2+n)+L(n-1)>H_{\rm prop}.
\]
This proves the converse and the $e_i=b_i$ case independently.
All numbers have polynomial encoding length, and unrestricted $(e_i,b_i)$
include the equal-execution case, completing the reduction.
\end{proof}

\subsection{Proof of \cref{thm:maxt-complete-eq} ($T_{\max}$, Complete, $b_i=b$)}
\label{apx:thm:maxt-complete-eq}

\begin{proof}
By \cref{lem:edf-complete} we sort jobs in EDF order and index them
$J_1, \ldots, J_n$ by non-decreasing deadline.
A schedule partitions this sequence into consecutive \emph{groups}, with a
complete recharge before each group except possibly the first, and no recharging
within a group.

\emph{Completion times depend only on the group boundaries.}
If the scheduler recharges before a group beginning at position $k$, then just after
that recharge the battery is $\BM$ while the first $k-1$ jobs have drawn $(k-1)b$
from the initial charge $\Bz$; by energy conservation the recharging time elapsed
so far is
\[
  R_k \;=\; \BM - \Bz + (k-1)\,b,
\]
\emph{independent of how the first $k-1$ jobs are grouped}.
Job $J_j$ in this group ($k \le j < \ell$) therefore completes at
$\sum_{i\le j} e_i + R_k$.
If the scheduler starts the first group without recharging, then it uses
the initial charge $\Bz$, and job $J_j$ completes at $\sum_{i\le j} e_i$.

Writing $x^{+}=\max(0,x)$, define for $1 \le k < \ell$ the maximum tardiness of
the group $J_k,\ldots,J_{\ell-1}$ preceded by a recharge,
\[
  \tau_{k,\ell} \;=\; \max_{k \le j < \ell}\Bigl(\textstyle\sum_{i\le j} e_i + R_k - d_j\Bigr)^{+},
\]
with $\tau_{k,\ell}=\infty$ if $(\ell-k)\,b > \BM$, and, for the
recharge-free first group,
\[
  \tau_{0,\ell} \;=\; \max_{1 \le j < \ell}\Bigl(\textstyle\sum_{i\le j} e_i - d_j\Bigr)^{+},
\]
with $\tau_{0,\ell}=\infty$ if $(\ell-1)\,b > \Bz$.
It remains to choose the group boundaries. Let $\tau_\ell$ be
the minimum maximum tardiness over schedules of the first $\ell - 1$ jobs.
We consider every possible start of the last group and combine its cost
with the best schedule for the preceding jobs:
\begin{equation}
  \tau_\ell = \begin{cases}
    0 & \ell \le 1, \\
    \min_{0 \le k < \ell} \max(\tau_k,\, \tau_{k,\ell}) & \ell > 1.
  \end{cases}
  \label{eq:dp-maxt}
\end{equation}
The $k=0$ term starts the first group from $\Bz$ with no initial recharge
(capacity $\Bz$); the terms $k\ge 1$ place a recharge before position $k$
(capacity $\BM$).
Since $R_k$ depends only on $k$, we can combine a candidate group with any
feasible schedule of the preceding jobs without changing the group's
completion times. This is why the recurrence only needs the best maximum
tardiness for that prefix. To compute the group costs, we fix each possible
start and extend the group one job at a time, updating its maximum tardiness
in $O(1)$ per endpoint. Computing all group costs and evaluating the
recurrence therefore take $O(n^2)$ time in total.
\end{proof}

\subsection{Proof of \cref{thm:maxt-complete-nph} ($T_{\max}$, Complete, NP-hard)}
\label{apx:thm:maxt-complete-nph}

\begin{proof}
Take a makespan decision instance, with threshold $D$, of any of the three
variants proved NP-hard in \cref{thm:makespan-nph}. Assign every job
$d_i=0$, leaving all execution times, energy requirements, and battery
parameters unchanged. Each completion time is nonnegative, so every job's
tardiness equals its completion time. Thus maximum tardiness equals
makespan for every feasible schedule, and a schedule with maximum
tardiness at most $D$ exists if and only if the original makespan
instance is YES. The map is polynomial and preserves each of the
three execution-time/energy relationships.
\end{proof}

\subsection{Proof of \cref{thm:numtardy-complete-eq} ($U$, Complete, $b_i=b$)}
\label{apx:thm:numtardy-complete-eq}

\begin{proof}
We can execute an optimal set of on-time (\emph{timely}) jobs in EDF order
(\cref{lem:edf-complete}). We therefore index the jobs $J_1,\ldots,J_n$ by
non-decreasing deadline and seek a largest subset whose jobs all meet their
deadlines.
As in \cref{thm:maxt-complete-eq}, a schedule partitions the chosen jobs into
groups separated by complete recharges; within a group the battery falls by $b$
per job, and a recharge restores it to $\BM$.

A dynamic program over the EDF sequence chooses which jobs to include in
the timely subset and when to recharge between them
(\cref{alg:numtardy-complete}).
A state $(t,\beta)$ records that we have chosen $t$ timely jobs so far and that
the battery is $\beta$ immediately after the $t$-th of them; $D[t][\beta]$ is the
least possible completion time of that $t$-th job over all such choices.
We process jobs in EDF order and append each chosen job after the previous
ones. Its deadline is the largest so far, and appending it preserves earlier
completion times. We can therefore select it exactly when its own completion
meets its deadline.
From a state with completion $L$ and battery $\beta$, we have two ways to
append $J_i$. If $\beta\ge b$, then we can extend the current group,
giving completion $L+e_i$ and battery $\beta-b$. Alternatively, we can
recharge for $\BM-\beta$ time units and start a new group, giving completion
$L+(\BM-\beta)+e_i$ and battery $\BM-b$. The recurrence retains each
option that meets the deadline. It also permits us to defer $J_i$ to the
tardy jobs, leaving the state unchanged.
The initial state $(0,\Bz)$ lets the first group run from $\Bz$ with no initial
recharge.

\emph{Optimality.}
To justify retaining only one completion time for a fixed $(t,\beta)$,
consider two partial schedules with that state. If we start the same
continuation after the earlier completion, then every later job completes
earlier by the same amount, and the battery levels are unchanged.
Because $\beta$ alone determines the next recharge duration, any continuation
feasible from a later completion is also feasible from the minimum one.
Retaining, for each $(t,\beta)$, only the minimum completion $D[t][\beta]$
therefore discards no optimal solution, and the largest $t$ with a finite entry
is the maximum number of timely jobs; the minimum number of tardy jobs is $n$
minus this.

\emph{Running time.}
The battery after any job equals $\Bz - jb$ or $\BM - jb$ for some
$0 \le j \le n$, so only $O(n)$ values of $\beta$ occur; with $O(n)$ values of
$t$ there are $O(n^2)$ states, each updated in $O(1)$ per job, giving $O(n^3)$
time.
\end{proof}

\subsection{Proof of \cref{lem:structure} (Optimal Schedule Structure)}
\label{apx:lem:structure}

\begin{proof}
For any fixed ordering, minimum recharging gives the completion at position
$r$ as
\[
  C_r=\sum_{j\le r}e_{\pi_j}
       +\max\!\left(0,\sum_{j\le r}b_{\pi_j}-\Bz\right).
\]
Swapping adjacent jobs therefore leaves the completion at the second position,
and every later completion, unchanged. Only the first position can change.

If the total energy demand is at most $\Bz$, then there is no separation
job and the usual SJF exchange proves the assertion. Otherwise, let $k$ be
the separation position. Sorting the prefix by $e_i$ preserves its total
energy demand, hence all its jobs still fit without recharging, and
cannot increase its sum of completion times. Sorting the jobs strictly
after $k$ by $e_i+b_i$ cannot increase the objective either: the separation
job leaves zero battery, and each later job requires its full $b_i$.
These exchanges preserve the separation job and position. We may order tied jobs arbitrarily within either sorted part.

The boundary conditions come from adjacent exchanges at the separation
job, taking its remaining initial charge into account. Write $y=\pi_k$ and let $\beta$ be the
battery remaining just before its recharge, so $0\le\beta<b_y$.
If a left neighbor $x=\pi_{k-1}$ exists, then the battery before the pair is
$\beta+b_x$. The first completion increment is $e_x$ before the exchange
and $e_y+\max(0,b_y-\beta-b_x)$ after it. Non-improvement gives
\[
 e_x\le e_y+\max(0,b_y-\beta-b_x),
\]
which is the first inequality in (\ref{eq:boundary}).
If a right neighbor $z=\pi_{k+1}$ exists, then the first increment changes
from $e_y+b_y-\beta$ to $e_z+\max(0,b_z-\beta)$, giving the second inequality.
This reasoning allows the exchange to move the separation position, and
covers $\beta=0$, a zero-energy neighbor, and equality at the initial
battery boundary. When $k=1$ or $k=n$, we omit the inequality for the missing neighbor. The second inequality does not in general imply
$e_y+b_y\le e_z+b_z$ when $b_z<\beta$; this is why part~(b) excludes $y$.
\end{proof}

\subsection{Proof of \cref{lem:partial-auxiliary} (Exact Cost Identity)}
\label{apx:lem:partial-auxiliary}

\begin{proof}
Write $z_i=\ell_i-1$, so $H_{i+1}=H_i+z_i$, and let
$\bar\ell_i=2-\ell_i$ be the right-block multiplicity. The left and right counts
strictly below type $i$ are respectively $i+H_i$ and $i-H_i$.
Since $\delta=m\Lambda-2E_\ell-\Bz$, the middle term of
(\ref{eq:partial-auxiliary}) depends on $S_i$ through
$-(m+1)\ell_iS_i$. Including the two sorted-block terms, the coefficient
of $S_i$ in $F(\ell)$ is
\begin{align*}
 &\ell_i\left(m-(i+H_i)-\frac{\ell_i-1}{2}\right)-(m+1)\ell_i\\
 &\qquad-\bar\ell_i\left(i-H_i+\frac{\bar\ell_i+1}{2}\right)\\
 &\quad=-2(i+1)-2z_iH_i-z_i^2\\
 &\quad=-2(i+1)-(H_{i+1}^2-H_i^2).
\end{align*}
All other terms are independent of $\ell$: in particular, the right-block
$\Lambda$ contribution is $\Lambda m(m+1)/2$ for every assignment.
Subtracting the reference multiplicities $\ell_i=1$ and using
$H_0=H_m=0$ gives
\begin{align*}
 F(\ell)-C_{\mathrm{ref}}
   &=-\sum_{i=0}^n S_i(H_{i+1}^2-H_i^2)\\
   &=\sum_{k=1}^n(S_k-S_{k-1})H_k^2
   =\sum_{k=1}^n a_kH_k^2.
\end{align*}
The same summation by parts gives
$d=\sum_i z_iS_i=-\sum_k a_kH_k$.

It remains to identify the completion times that the auxiliary cost
misrepresents. The first $m-1=n$ non-anchor jobs require no recharge,
and the first right-block job brings total consumption above $\Bz$
(\ref{eq:partial-count-bounds}).
Only two completions can differ from the auxiliary expression: the last
left-block job needs an additional $\delta^+=\max(0,\delta)$, and the anchor
needs an additional $\delta^+-\delta=(-\delta)^+$.
For every right-block job, cumulative recharging is already the cumulative
energy demand minus $\Bz$, exactly as represented by the signed term
$\delta$ and the right-block $b_i$ terms. Thus the total correction is
$\delta^++(-\delta)^+=|\delta|$. This also covers $\delta=0$ and proves
(\ref{eq:partial-cost-identity}).

For completeness, a forward generalized exchange from type $v$ to type
$u>v$ replaces $H_k$ by $H_k-1$ exactly on $v<k\le u$.
Taking differences of the squared terms gives
(\ref{eq:partial-general-swap}); reversing the exchange replaces them by
$H_k+1$. Subtracting the two feasible-cost identities gives
(\ref{eq:partial-feasible-swap}). These calculations apply whenever the
outgoing and incoming copies exist, and do not assume that their windows
are disjoint or that the feasible cost increases.
\end{proof}

\subsection{Proof of \cref{lem:anchor-pos} (Anchor Position and Ties)}
\label{apx:lem:anchor-pos}

\begin{proof}
Choose an optimal schedule with minimum recharging and the structure in
\cref{lem:structure}. Total consumption exceeds $\Bz$, so a separation job
exists. The anchor cannot be that job because its energy demand is zero.
If it belongs to the prefix, then its strictly largest execution time places
it last there. If it occurs after the separation job, then SSF places it
before all jobs of types $0,\ldots,n-1$. We resolve ties with type $n$ by
moving the anchor first among those tied tail jobs; all are strictly after
the separation job, so exchanging them leaves the objective unchanged.
Thus the anchor is either last in the prefix or immediately after the
separation job. This tie rule never exchanges a job across the separation position.

Let $q$ be the number of non-anchor jobs before the first recharge.
By (\ref{eq:partial-count-bounds}), $q$ is $n$ or $n+1$: any $n$ fit,
and no $n+2$ fit. If the anchor is last in a prefix with $q=n+1$, then its position
is $n+2$. If it follows the separation job, then its position is $q+2$, giving
$n+2$ or $n+3$ as stated.

The remaining possibility, an anchor last in a prefix with $q=n$, cannot
be optimal. Exchange it with the following separation job $x$.
Since the anchor consumes no battery, $x$ has the same initial residual
battery before and after this exchange. Its recharge deficit is at most
$(n+1)\Lambda-\Bz=2\sum_{i=0}^nS_i+2T$. Consequently its execution plus
this recharge is at most
\[
  2T+2\sum_{i=0}^nS_i+2T\le(4n+4)T<A.
\]
The exchange strictly reduces the first of the two completions and leaves
the second and all later completions unchanged, a contradiction.
The estimates are strict under (\ref{eq:Lambdabound}); they also cover a
prefix that exhausts the initial battery exactly.
\end{proof}

\subsection{Proof of \cref{lem:only-if} (Only-If Direction)}
\label{apx:lem:only-if}

\begin{proof}
Suppose a feasible schedule has cost at most $C_{\mathrm{ref}}+T$.
Choose an optimal minimum-recharge schedule in the form guaranteed by
\cref{lem:anchor-pos}; its cost is no greater.
For every integer $H_k$, $H_k(H_k+1)\ge0$. Set
\[
  K=\sum_{k=1}^n a_kH_k(H_k+1)\ge0,
  \qquad \sum_{k=1}^n a_kH_k^2=d+K.
\]
Since all $a_k>0$, equality $K=0$ holds exactly when every
$H_k\in\{0,-1\}$.

\emph{Anchor at position $n+2$.}
There are $m=n+1$ non-anchor jobs on each side. Sort the left block by $e_i$
and the right block by $e_i+b_i$. This does not increase cost even if the
left block contains the separation job: by (\ref{eq:partial-count-bounds}),
its first $n$ jobs have no recharge in any ordering, its last completion
depends only on the block's total execution and energy, and all subsequent
completions depend only on the block totals. Every right-block job is past
the initial charge boundary, so SSF is valid throughout that block.
Let $\ell$ be the resulting counts. Its cost is
\[
 C(\ell)=C_{\mathrm{ref}}+d+K+2|T-d|.
\]
For $d\le T$, $d-T+2|T-d|=T-d$; for $d\ge T$, it equals $3(d-T)$.
Therefore $C(\ell)\le C_{\mathrm{ref}}+T$ forces $d=T$ and $K=0$.

\emph{Anchor at position $n+3$.}
Let $x$ be the separation job immediately before the anchor, and write
$s=e_x=S_i$ for its type's prefix sum. The first $m=n+1$ non-anchor jobs
form a recharge-free SJF block with counts $\ell$ and residual battery
\[
  \beta=\Bz-\sum_i\ell_i b_i=2(d-T)\ge0.
\]
Move the anchor just before $x$, leaving every other job fixed, to obtain a
central ordering $\sigma$. This move preserves the completion of the second position in the pair
and every later completion. The first
completion increases by
\[
  A-(e_x+b_x-\beta)
    =(\Lambda-2T)-(\Lambda-s-\beta)=\beta+s-2T.
\]
Consequently the original cost is
$C(\sigma)-\beta+2T-s$; we do not assume a sign for this change.
In $\sigma$, every job to the right of the anchor exceeds the initial
battery boundary, because the first such job is the $(n+2)$-nd non-anchor
job. Sorting that right block by $e_i+b_i$ cannot increase cost and yields
$F(\ell)+\beta$ by \cref{lem:partial-auxiliary}. Hence the original cost
is at least
\[
  F(\ell)+2T-s
   =C_{\mathrm{ref}}+d+K+2T-s
   \ge C_{\mathrm{ref}}+T.
\]
Here $d\ge T$, $K\ge0$, and $s\le2T$. A cost at most the threshold forces
$d=T$, $K=0$, and $s=2T$. The equality case includes an anchor after a type-$n$ separation job.
Resolving SSF ties only within the tail preserves this case.

In both cases $d=T$ and $H_k\in\{0,-1\}$. Thus
\[
  \sum_{k:H_k=-1}a_k=-\sum_{k=1}^n a_kH_k=d=T,
\]
which is a Partition solution.
\end{proof}

\subsection{Proof of \cref{thm:numtardy-dp} ($U$, Partial, Dynamic Program)}
\label{apx:thm:numtardy-dp}

\begin{proof}
The timely jobs can be scheduled in EDF order with the tardy jobs after
them in any order, so the task is to choose a largest subset that can meet
its deadlines. Whether adding a job preserves this property depends on
when the current timely sequence ends and on how much charge remains.

Define $\eta(m, f, B)$ as the minimum number of tardy jobs after processing
the first $m$ EDF-ordered jobs, where $f$ is the completion time of the last
timely job and $B$ is the battery level after that job. For each next job,
we compare deferring it to the tardy set with appending it to the timely
sequence, provided that it can complete by its deadline. The recurrence
accounts for any recharge needed before that execution (\cref{fig:bellman}).

The state space has $O(n\BM\sum_i(e_i+b_i))$ states; at states with $B = 0$ there are $O(\max_i b_i)$ options. The stated running time follows.
\end{proof}

\subsection{Subset-Sum Subroutines Used by $\mathcal{A}_2$}
\label{apx:a2-subroutines}

To obtain the guarantee for $\mathcal{A}_2$, we need a $5/4$-approximation for SS-MIN
and a $3/4$-approximation for SS-MAX (\cref{thm:a2-guarantee}). We describe
the two subroutines used in our experiments and prove that they meet
these requirements. Both divide items into classes according to their
weights, which lets us restrict the combinations we need to consider.

Let $w(X)=\sum_{i\in X}b_i$ for a set of item indices $X$. Equal weights
remain distinct items. For a positive target $C$, partition the positive
weights into
\begin{align*}
 G_0&=(0,C/4), &G_1&=[C/4,C/2),\\
 G_2&=[C/2,3C/4), &G_3&=[3C/4,C),\\
 G_4&=[C,\infty).
\end{align*}
Here an interval denotes the set of indices whose weights lie in it.
Our implementation compares integer multiples, such as $4b_i$ with $3C$,
to preserve these boundaries for arbitrary integer inputs.
Sort each class in non-decreasing weight, retaining input order within ties.
We may omit zero-weight items because they change neither subset objective.
$\mathcal{A}_2$ retains their jobs when it constructs the complete schedule.

\subsubsection{Minimum Head Energy}

To choose the head, we need a subset with enough energy to leave a final
group that fits in the battery, while keeping the head's energy as small
as possible. We use SS-MIN to make this choice.
We assume $0<C\le w(\mathcal J)$; $\mathcal{A}_2$ handles non-positive targets without
calling this subroutine. We choose an inclusion-minimal optimal subset:
removing any selected positive-weight item leaves a sum below $C$.
If it contains an item from $G_4$, then that item alone is optimal among subsets
containing an oversized item. Otherwise its large-item counts
$(n_3,n_2,n_1)$ belong to
\begin{align*}
 \mathcal M=\{&(0,0,0),(0,0,1),(0,0,2),(0,0,3),(0,0,4),\\
              &(0,1,0),(0,1,1),(0,1,2),(0,2,0),\\
              &(1,0,0),(1,0,1),(1,1,0),(2,0,0)\}.
\end{align*}
To obtain this list, suppose the subset contains a $G_1$ item. Removing one of these
items leaves weight at least $(3n_3+2n_2+n_1-1)C/4$, so minimality requires
$3n_3+2n_2+n_1\le4$. If the subset contains no $G_1$ item but contains a $G_2$ item, then removing
one $G_2$ item gives $3n_3+2n_2\le5$. With large items only from $G_3$,
we obtain $n_3\le2$. These conditions give the 13 tuples above, including
$(0,0,0)$ for a subset containing only small items. Small items can impose
additional restrictions, but cannot permit another large-item composition.

It remains to choose the items within each class. To avoid considering
every subset, we select a contiguous window of the prescribed size from each
of $G_1,G_2,G_3$. We start with the smallest items and slide the
$G_1$ window to its largest position, then the $G_2$ window, then $G_3$.
A slide removes the first item in one window and adds its next item.
The selected weight is non-decreasing along this path, and each increment
is strictly below $C/4$. The largest windows have at least as much weight
as any selection with the same counts. We stop at the first selection
whose weight, together with all small items, reaches $C$.

The algorithm also considers the smallest and largest windows directly.
For each selection, it adds small items in ascending order until it reaches
$C$, rejecting the selection if even all small items are insufficient
(\cref{alg:ssmin-implementation}). The two endpoint candidates preserve the
implementation's choices but are not needed for the bound.

\begin{algorithm}[tbp]
\caption{\textsc{SelectMinHead}$(\mathcal J,C)$}
\label{alg:ssmin-implementation}
\begin{algorithmic}[1]
\Require Positive weights with $0<C\le w(\mathcal J)$
\Ensure $H\subseteq\mathcal J$, $C\le w(H)\le(5/4)\operatorname{OPT}_{\min}$
\State Partition and sort $G_0,\ldots,G_4$; $S_0\gets w(G_0)$
\State $H\gets\mathcal J$
\If{$G_4\ne\emptyset$}
  \State $H\gets\{\text{smallest item in }G_4\}$
\EndIf
\For{each available composition $(n_3,n_2,n_1)\in\mathcal M$}
  \State $L_{\min},L_{\max}\gets$ smallest and largest windows with these counts
  \State \Call{Consider}{$L_{\min}$}; \Call{Consider}{$L_{\max}$}
  \State $L\gets L_{\min}$
  \For{$t=1,2,3$}
    \While{$w(L)+S_0<C$ and the nonempty $G_t$ window can slide}
      \State Slide that window one position to the right
    \EndWhile
  \EndFor
  \State \Call{Consider}{$L$}
\EndFor
\State \Return $H$
\Procedure{Consider}{$L$}
  \If{$w(L)+S_0\ge C$}
    \State Add the shortest ascending prefix of $G_0$ needed to reach $C$
    \If{$w(L)<w(H)$}
      \State $H\gets L$ \Comment{store a copy; do not modify other candidates}
    \EndIf
  \EndIf
\EndProcedure
\end{algorithmic}
\end{algorithm}

\begin{lemma}[Guarantee of \textsc{SelectMinHead}]
\label{lem:ssmin-implementation}
\textup{\textsc{SelectMinHead}} is a $5/4$-approximation for SS-MIN and runs in
$O(n\log n)$ time.
\end{lemma}
\begin{proof}
Let $H^*$ be an inclusion-minimal optimum. If $H^*$ contains an oversized
item, then the smallest item of $G_4$ is an optimal candidate.
Otherwise let $L^*$ be its large items. Their composition belongs to
$\mathcal M$, and
\[
 C-S_0\le w(L^*)\le w(H^*)=\operatorname{OPT}_{\min}.
\]
Consider the sliding path for this composition. Its largest selection has
weight at least $w(L^*)$, so the path reaches $C-S_0$. Let $L$ be its first
selection reaching that value. If $L$ is the initial selection, then
$w(L)\le w(L^*)\le\operatorname{OPT}_{\min}$. Otherwise its previous
selection has weight below $C-S_0$ and one slide adds less than $C/4$,
so $w(L)<C-S_0+C/4\le5C/4$.

If $w(L)\ge C$, then this candidate is already feasible and has weight
at most $\max(\operatorname{OPT}_{\min},5C/4)$.
If $w(L)<C$, then small-item augmentation can reach $C$; its last item
weighs less than $C/4$, giving final weight below $5C/4$.
Since $\operatorname{OPT}_{\min}\ge C$, either candidate satisfies the
claimed bound. Keeping the best candidate preserves it.

Sorting takes $O(n\log n)$ time. There are 13 compositions, each containing
at most four large items. Each sliding path advances each class window at
most its class length; a slide updates the weight by subtracting the removed
item and adding the new item, in $O(1)$ time. We augment at most three
candidates per composition, each in $O(n)$ time. Candidate copying also
costs at most $O(n)$ per candidate. The work after sorting is therefore
$O(n)$, with $O(n)$ storage.
\end{proof}

This bounds the energy of the selected head relative to its minimum.
Adding the same non-negative execution and charging terms to both values
cannot increase the ratio (\cref{thm:ssmin-reduct}), which gives the
makespan bound. Showing that the makespan ratio is tight would still need
instances that satisfy the capacity condition (\cref{thm:ssmin-tight}).

\subsubsection{Maximum Final-Group Energy}

To choose the final group, we instead need to select as much energy as
possible without exceeding one full charge. We use SS-MAX for this
selection. We discard items above
$C$ because they cannot fit, and retain items of weight exactly $C$ in $G_4$.
The routine enumerates
\[
 (n_4,n_3,n_2,n_1)\in\{0,1\}\times\{0,1\}\times\{0,1,2\}\times\{0,1,2,3\}.
\]
For each available tuple, it considers the smallest prescribed items from
every class and, separately, the largest prescribed items from every class.
It rejects an overweight selection; otherwise it scans $G_0$ in ascending
order, adding each item that fits (\cref{alg:ssmax-implementation}).
There are at most 96 initial selections. We do not need to enumerate
every feasible composition: we show below that these selections either
produce enough weight to meet the approximation guarantee or include
an optimal subset.

\begin{algorithm}[tbp]
\caption{\textsc{SelectMaxFinalGroup}$(\mathcal J,C)$}
\label{alg:ssmax-implementation}
\begin{algorithmic}[1]
\Require Non-negative weights and capacity $C>0$
\Ensure $w(F)\le C$ and $w(F)\ge(3/4)\operatorname{OPT}_{\max}$
\State Discard weights above $C$; partition and sort $G_0,\ldots,G_4$
\State $F\gets\emptyset$
\For{each available tuple $(n_4,n_3,n_2,n_1)$ in the stated ranges}
  \For{each of largest-first selection and smallest-first selection}
    \State $L\gets$ the prescribed number of items from each of $G_4,G_3,G_2,G_1$
    \If{$w(L)\le C$}
      \For{each $i\in G_0$, in ascending order}
        \If{$w(L)+b_i\le C$}
          \State $L\gets L\cup\{i\}$
        \EndIf
      \EndFor
      \If{$w(L)>w(F)$}
        \State $F\gets L$
      \EndIf
    \EndIf
  \EndFor
\EndFor
\State \Return $F$
\end{algorithmic}
\end{algorithm}

\begin{lemma}[Guarantee of \textsc{SelectMaxFinalGroup}]
\label{lem:ssmax-implementation}
\textup{\textsc{SelectMaxFinalGroup}} returns either a subset of weight at least
$3C/4$ or an optimal subset. It is therefore a $3/4$-approximation for
SS-MAX and runs in $O(n\log n)$ time.
\end{lemma}
\begin{proof}
Suppose every returned candidate has weight below $3C/4$; we show that the
best is optimal. In any feasible candidate's small-item scan, rejecting
an item of weight below $C/4$ would mean that the accumulated weight
already exceeds $3C/4$. Thus every feasible candidate includes all small
items. Write $S_0=w(G_0)$.

There can be no eligible item in $G_3$ or $G_4$, since its singleton
candidate already reaches $3C/4$. Nor can a feasible subset contain both
a $G_2$ item and a $G_1$ item: the two smallest items of those classes
would form a feasible candidate of weight at least $3C/4$.
Two $G_2$ items can fit only if both weigh exactly $C/2$; the smallest-pair
candidate would then reach $C$.
Finally, if any three $G_1$ items fit, then the three smallest fit and
have total weight at least $3C/4$, again a contradiction.

Consequently a feasible subset's large items consist of at most one
$G_2$ item or at most two $G_1$ items. The algorithm enumerates the largest $G_2$ singleton and the largest pair
of $G_1$ items when available; the pair
fits because each item weighs strictly less than $C/2$.
If there is only one $G_1$ item, then the algorithm enumerates its singleton
instead. It also enumerates the empty selection. Each of these candidates takes all
small items, as proved above. Taking the largest of their weights therefore
attains the optimum over every possible form of a feasible subset.

If a candidate instead reaches $3C/4$, then it has weight at least
$(3/4)\operatorname{OPT}_{\max}$ because $\operatorname{OPT}_{\max}\le C$.
This proves the bound in both cases. Sorting dominates the runtime:
the constant number of selections and their small-item scans take $O(n)$
time after sorting, and the storage is $O(n)$.
\end{proof}

Both subroutines meet the hypotheses of \cref{thm:a2-guarantee}, so the
guarantee applies to the implementation we evaluated. In the SS-MAX regime, the selected final group
falls short of the optimal energy by at most $\BM/4$. By energy
conservation, we need at most that much additional recharging time, which
gives the makespan error used in the theorem.

The evaluated schedule builder recharges only when needed within the head,
then fills the battery before a nonempty final group. When recharge is
necessary and the designated final group is nonempty, this has the same
makespan as \cref{alg:shtf}: both have total recharge
$\BM-\Bz+w(\mathcal H)$ by energy conservation. If the designated final
group is empty, then both builders omit terminal recharging and use at
most this amount, which suffices for the approximation upper bound.
This comparison concerns makespan; the intermediate completion times can
differ.

\end{document}